\documentclass[english,aps,prb,showpacs,superscriptaddress,floatfix,longbibliography, notitlepage,reprint,unicode=true,colorlinks=true,citecolor=Blue,linkcolor=RubineRed,urlcolor=Blue]{revtex4-2}
\usepackage[T1]{fontenc}
\usepackage{mathrsfs}
\usepackage{bm}
\usepackage{amsmath}
\usepackage{amsthm}
\usepackage{amssymb}
\usepackage{graphicx}
\usepackage{booktabs}
\usepackage[usenames,dvipsnames]{xcolor}
\usepackage[]{hyperref}
\makeatletter

\usepackage{braket}
\usepackage{txfonts} 

\theoremstyle{definition}
\newtheorem{theorem}{Theorem}
\newtheorem{prp}[theorem]{Principle}
\newtheorem{obs}[theorem]{Observation}
\date{\today}

\usepackage{babel}
\begin{document}
\title{Asymptotic Metrological Scaling and Concentration in Chaotic Floquet Dynamics}
\author{Astrid J.~M.~Bergman}
\affiliation{Department of Physics, KTH Royal Institute of Technology, Stockholm, Sweden.}
\author{Yunxiang Liao}
\email{liao2@kth.se}
\affiliation{Department of Physics, KTH Royal Institute of Technology, Stockholm, Sweden.}
\author{Jing Yang}
\email{jing.yang.quantum@zju.edu.cn}
\affiliation{Institute of Fundamental and Transdisciplinary Research, Institute of Quantum Sensing, and Institute for Advanced Study in Physics, Zhejiang University, Hangzhou 310027, China.}
\affiliation{Nordita, KTH Royal Institute of Technology and Stockholm University, Hannes Alfv\'ens v\"ag 12, 106 91 Stockholm, Sweden.}
\begin{abstract}
We study quantum sensing with Floquet chaotic dynamics generated by
Haar random unitary gates. The metrological resources consist of three
ingredients: A given initial state, a set number of Haar random
unitary gates and the sensing gates. There are two natural ways of
organizing the resources: the first one is the ``control'' protocol,
where the random unitary gates act as random controls and intertwine
with the deterministic sensing gates and the second one is the
``state-preparation'' protocol, where random unitary gates play
the role of preparing the metrological useful states. In each
protocol, we consider both global Haar random unitary gates and a set of local two-site Haar random unitary
gates that forms a Floquet random quantum circuit (RQC) respectively. We find linear, shot-noise scaling of the metrological precision, quantified
by the quantum Fisher information (QFI), in the asymptotic limit when the Hilbert space dimension
becomes large, and quantum advantages beyond linear scaling
in the non-asymptotic regimes. We also bound the fluctuation of the
QFI using concentration inequalities. Our analytical findings are corroborated
by numerical simulations. Finally, along the way of analyzing
the precision limit, we prove an empirical
conjecture of RQC: In the asymptotic limit of large local Hilbert
space dimension, the Floquet operator of a Floquet RQC essentially
behaves like a global unitary operator.
\end{abstract}
\maketitle

\section{Introduction}

Precision measurement plays the fundamental role in physics. The performance
of a measurement apparatus can be refined by cooling the temperature
down and hence reducing the technical noise. However, at the very
fundamental level, quantum mechanics leads to some intrinsic
noise in the precision measurement, due to the randomness of quantum
measurements. Quantum metrology and sensing is the subject on leveraging
quantum coherences and entanglement to reduce the quantum noise and
enhance the precision of measuring some physical quantities, such
as the magnetic field, frequency etc~\citep{pezz`e2018quantum,*degen2017quantum}. 

Recent progress in quantum metrology has witnessed a number of works
on quantum metrology with many-body quantum systems~\citep{montenegro2024reviewquantum},
where the signal is encoded into a many-body state through many-body dynamics or a thermal Gibbs state~\citep{yang2022,chu2023strongquantum,yang2022superheisenberg,shi2024universal,puig2025fromdynamical,abiuso2025fundamental}.
Depending on resources, including the initial entanglement and the availability
of many-body controls, the quantum Fisher information (QFI), characterizing
the precision of the signal, can display different scaling behaviors
with respect to the time and the number of sensors. In the absence of interactions between sensors, the QFI scales linearly in the number of sensors if the initial state is unentangled and quadratically if the initial state is a GHZ-state. These limits are known as the shot noise limit and the Heisenberg limit respectively~\citep{giovannetti2004quantumenhanced,*giovannetti2006quantum}.
Alternatively, one can also exploit the resources of many-body interactions
to engineer the initial state to an entangled state then encode the
signal via a non-interacting Hamiltonian~\citep{chu2023strongquantum}.
The circuit versions of both protocols are shown in Fig.~\ref{fig:metrology-protocols}(a)
and (b) respectively.
\begin{figure}
\centering
\includegraphics[scale=0.8]{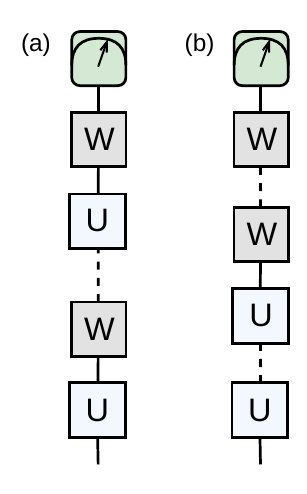}

\caption{\label{fig:metrology-protocols}Two protocols used in quantum metrology, where $W$ is the sensing gate depending on the signal while $U$ plays the role of the control gate in protocol (a) and the state-preparation gate in protocol (b).
(a) Control protocol where the evolution operator is $\mathcal{U}_{\mathrm{ctr}}(t)=(WU)^{t}$.
(b) State-preparation protocol where the evolution operator
is $\mathcal{U}_{\mathrm{sp}}(t)=W^{t}U^{t}$.}
\end{figure}
However, calculating the scaling of the QFI for exact many-body dynamics is notoriously difficult both analytically and numerically, except for a few integrable models. Instead, one can consider the average of the QFI taken over an ensemble of states or dynamics.  By averaging over the Haar-random unitary gates of the Circular Unitary Ensemble (CUE), we can bypass the microscopic complexities of specific Hamiltonians and instead derive universal scaling laws for the precision limits in chaotic or scrambling systems.
This transition to the ensemble averaged QFI allows us to leverage powerful tools from random matrix theory~\cite{haake2001quantum,DAlessio2016,Wigner_1951,*Wigner_1955,*Wigner_1957,*Wigner_1958,Dyson-I,*Dyson1962}. Furthermore, as we will see, this statistical approach is motivated by the fact that in the limit where the dimension of the Hilbert space is large, the QFI often exhibits concentration of measure. Therefore, the ensemble averaged QFI  becomes highly peaked around its mean value, rendering it a robust proxy for the performance of typical experimental realizations. More specifically, we consider Floquet dynamics in two cases; where
the Floquet operator is formed by global Haar random unitary gates, dubbed the case of Random Matrix Model (RMM); and
a set of local two-site Haar random unitary gates, dubbed the case of Random Quantum Circuits (RQC). Note that while local Hamiltonian dynamics can be approximated as a continuous limit of deep quantum circuits, the RQC case considers discrete layers of gates that are not necessarily close to the identity. This digitalized approach, is different from previous works on quantum sensing with chaotic Hamiltonian dynamics~\cite{li2023improving,fiderer2018quantum,wang2011chaosand,pavlov2024randommatrix}  and is particularly relevant for quantum sensing in noisy intermediate-scale quantum devices, where variational algorithms~\citep{zheng2022preparation,kaubruegger2021quantum} seek to optimize circuit parameters for metrological gain.

RQC (see review Ref~\cite{Fisher2023} and references therein) are built on the random matrix theory, which conjectures that the statistical properties of energy eigenvalues and eigenstates within a small enough window in  quantum chaotic systems can be modeled by a suitably chosen random matrix models. While numerous evidence show that RMM successfully captures the universal spectral statistics of chaotic many-body interacting systems~\cite{Bohigas1984,chan2018solution,Prosen-SFF1,Prosen-Boson,Prosen-Fermion,Chalker-2,Prosen-SFF2,Prosen-SFF3,Huse}, it neglects an important feature of physical many-body interacting systems - the spatial locality of interactions~\cite{DAlessio2016}. Random quantum circuits address this limitation by generalizing the ideas of RMM in a way that preserves locality. They are usually composed of sequences of random unitary gates acting on qubits or qudits on nearby sites, and sometimes  also intersected with local measurements. Random quantum circuits exhibit many universal properties observed across a wide of range of more realistic quantum many-body interacting systems, and serve as a powerful framework both in the area of condensed matter and quantum information. For example, they have been widely employed to investigate information scrambling~\cite{Nahum2018,*vonKeyserlingk2018,*Keyserlingk-2,Chalker-8,Huse-2018,Bertini}, entanglement dynamics~\cite{Nahum2017,Zhou2019}, thermalization~\cite{Chalker-ETH,Chalker-4,liao2022field,Chalker-9}, and localization~\cite{Chalker-3,Chalker-5,Chalker-6,Chalker-7,Huse-1,Huse-2} in many-body quantum interacting systems.

For the RMM case, in the asymptotic limit where the dimension of the
Hilbert space is large, the QFI for the control protocol and the state-preparation
protocol scale as $(\mathrm{Tr}H_{0}^{2})t/N$ and $(\mathrm{Tr}H_{0}^{2})t^{2}/N$,
respectively, where $t$ is the discrete time, $N$ is the dimension of the global
Hilbert space, and $H_{0}$ is the Hamiltonian that generates the
sensing gate whose trace is of the order of $N$. For the RQC case,
in the limit $q\to\infty$, the QFI for the control protocol and
the state-preparation protocol scales as $\mathrm{Tr}(h_0^{2})tL/q$
and $\mathrm{Tr}(h_0^{2})t^{2}L/q$, respectively, where again, 
$t$ is the discrete time, $L$ is the system size, $q$ is the local
Hilbert space, and $h_0$ is the local single-site Hamiltonian that generates the
local sensing gate whose trace is of the order of $q$. The results
are summarized in Table~\ref{tab:results-summary} and see Sec.~\ref{sec:QFI-scaling} for details. We further bound
the fluctuation of the QFI using concentration inequalities and our analytics
are confirmed by numerical simulations. 

\begin{table}
\centering
\begin{tabular}{ccc}
\toprule 
 & Control protocol & State-preparation protocol \\
 \midrule
RMT ($N\to\infty$) & $\mathrm{Tr}(H_{0}^{2})t/N$ & $\mathrm{Tr}(H_{0}^{2})t^{2}/N$ \\
RQC ($q\to\infty$) & $\mathrm{Tr}(h_0^{2})tL/q$ & $\mathrm{Tr}(h_0^{2})t^{2}L/q$ \\
\bottomrule
\end{tabular}

\caption{\label{tab:results-summary}The scaling of the average QFI for the control
protocol and the state-preparation protocol, shown in Fig.~\ref{fig:metrology-protocols} } 

\end{table}

On the technical side, the calculation of the ensemble average is
performed using the Weingarten formula and its diagrammatic representation.
The Weingarten calculus \cite{Weingarten,Samuel,Mello,collins2003,*collins2006,*collins2009,*collins2021} is a method for evaluating integrals of polynomials in matrix elements over compact groups with respect to the Haar measure. 
In particular, a diagrammatic approach was introduced to organize in a systematic way the combinatorial in the Weingarten calculus  \cite{brouwer1996diagrammatic}. This diagrammatic method was then further developed for random quantum circuits with Haar distributed unitary gates for various properties \cite{chan2018solution}. We employ this approach for the computation of the QFI for the Floquet RQCs depicted in Fig.~\ref{fig:RQC_notation}, and formulate a mathematical representation of the many-body diagrammatic technique for a wide range of random quantum circuits.
Here, we give a mathematical representation of the Weingarten
calculus for Floquet RQC. Recently, it was found that when the dimension
of local Hilbert space is large enough, locality in local Floquet
RQC disappears and the behavior of the characteristics of quantum
chaos, such as the spectral form factor and the eigenstate correlation function, can be captured by RMM~\citep{chan2018solution,liao2022effective, liao2022field}.
We prove this empirical belief by leveraging the many-body Weingarten
calculus and show that the Floquet operator of a Floquet RQC essentially
behaves like a global unitary gate drawn from the CUE in the asymptotic
limit of large dimension of the local Hilbert space. Furthermore,
we also show that the ensemble average of a trace product involving
the Floquet operator in the Floquet RQC factorizes into a single-site
trace product, where the latter bears a convenient diagrammatic representation
from RMM, under the constraint that the single-site diagrams must
be compatible with each other. These two observations play an essential role in obtaining the scaling of the QFI in the
RQC case, and holds not
only for brickwork Floquet RQC, but also generically for Floquet RQC with arbitary configurations of the nearest-neigbor two-site random gates.

This paper is organized as follows: In Sec.~\ref{sec:QFI-background}, we introduce the essential background for quantum sensing and metrology. In Sec.~\ref{sec:Weingarten-Pedagogical},  we  first review the Weingarten calculus and its diagrammatic representation. We then develop their many-body generalization for RQC, based on which we prove rigorously that local Haar random gates behave the same as a global CUE when the local Hilbert space dimension goes to infinity. In Sec.~\ref{sec:QFI-scaling}, we analyze the scaling the QFI in the RMM and RQC cases for both the control and state-preparation protocols. In Sec.~\ref{sec:QFI-fluctuations}, we show that the QFI becomes deterministic in the asymptotic limit when the Hilbert space dimension becomes infinite. In Sec.~\ref{sec:Dis-Con}, we discuss several aspects of the problem and conclude with future directions. 

\section{The quantum Cramer-Rao bound and the quantum Fisher information\label{sec:QFI-background}}

The idea of local quantum metrology is to encode the quantity or signal
to be measured, denoted as $\theta$, onto a quantum state, which
leads to the a parameter-dependent state $\rho_{\theta}$. Afterwards,
one performs multiple rounds of identical positive-operator-valued-measure
(POVM) measurements $\mathrm{E}\equiv\{E_{\omega}\}_{\omega=1}^{|\Omega|}$
to extract the information about the parameter. The precision of the
measurements is characterized by the so-called classical Cramer-Rao
bound, 
\begin{equation}
\check{\theta}(\omega_{1},\,\omega_{2},\cdots\omega_{\nu})\ge\frac{1}{\nu F^{C}(\rho_{\theta},\,\mathrm{E})},
\end{equation}
where $F^{C}(\rho_{\theta},\,\mathrm{E})=\sum_{\omega}[\partial_{\theta}p(\omega|\theta)]^{2}/p(\omega|\theta)$
is the classical Fisher information (CFI) and $p(\omega|\theta)=\mathrm{Tr}(\rho_{\theta}E_{\omega})$.
The classical Cramer-Rao bound is asymptotically saturable by the
maximum likelihood estimator in the limit $\nu\to\infty$. Further
optimizing over all possible POVM measurements leads to the so-called
quantum Fisher information (QFI) as an upper bound of the CFI, i.e.,
\begin{equation}
F^{C}(\rho_{\theta},\,\mathrm{E})\leq F^{Q}(\rho_{\theta})=\mathrm{Tr}\left(\rho_{\theta}L_{\theta}^{2}\right),\label{eq:QCRB}
\end{equation}
where $L_{\theta}$ is the symmetric logarithmic derivative defined
as $(\rho_{\theta}L_{\theta}+L_{\theta}\rho_{\theta})/2=\partial_{\theta}\rho_{\theta}$.
It should be noted that for the estimation of a single parameter,
the bound~(\ref{eq:QCRB}) is always achievable~\citep{braunstein1994statistical}. 

In particular, for pure states, the QFI simplifies to
\begin{equation}
F^{Q}(\ket{\psi_{\theta}})=4\left(\braket{\partial_{\theta}\psi_{\theta}|\partial_{\theta}\psi_{\theta}}-|\braket{\psi_{\theta}|\partial_{\theta}\psi_{\theta}}|^{2}\right).
\end{equation}
For unitary encoding with pure initial states, i.e., $\ket{\psi_{\theta}}=\mathcal{U}_{\theta}\ket{\psi_{0}}$,
we find 
\begin{equation}
F^{Q}(\mathcal{U}_{\theta},\,\ket{\psi_{0}})=4\mathrm{Var}[G]_{\ket{\psi_{0}}},\label{eq:QFI-def-in-G}
\end{equation}
where $G\equiv\mathrm{i}\partial_{\theta}\mathcal{U}_{\theta}^{\dagger}\mathcal{U}_{\theta}$
is known as the metrological generator~\citep{boixo2007generalized} and Var$[\bullet]$ denotes the variance of an operator.
Throughout this work, we shall consider the case of unitary encoding
where the encoding gates contain Haar random unitary gates. We investigate
the scaling behavior of the ensemble averaged QFI.
As we will show subsequently, for pure initial states $\langle F^{Q}(\mathcal{U}_{\theta},\,\ket{\psi_{0}})\rangle$
is independent of the initial state, which we denote as $\langle F^{Q}(\mathcal{U}_{\theta})\rangle$.

Then, for any initial mixed state $\rho_0$ with the the spectral decomposition $\rho_{0}=\sum_{n}p_{n0}\ket{\psi_{n0}}\bra{\psi_{n0}}$, thanks to the convexity property of the QFI~\citep{petz2011introduction}
\begin{equation}
F^{Q}(\mathcal{U}_{\theta},\,\rho_{0})\leq\sum_{n}p_{n0}F^{Q}(\mathcal{U}_{\theta},\,\ket{\psi_{n0}}),
\end{equation}
we find 
\begin{equation}
\langle F^{Q}(\mathcal{U}_{\theta},\,\rho_{0})\rangle\leq\sum_{n}p_{n0}\langle F^{Q}(\mathcal{U}_{\theta},\,\ket{\psi_{n0}})\rangle=\langle F^{Q}(\mathcal{U}_{\theta})\rangle.
\end{equation}
Thus, $\langle F^{Q}(\mathcal{U}_{\theta})\rangle$ also provides
the bounds of the ensemble averaged QFI for all initial states. 

\section{The Weingarten calculus and its many-body generalization for Floquet RQC\label{sec:Weingarten-Pedagogical}}
In this section, we first review the Weingarten calculus, which allows one to compute the ensemble average over Haar random matrices, and a diagrammatic method developed by Brouwer and Beenaker~\cite{brouwer1996diagrammatic} to bookeep the ensemble-averaged results. Then we propose the generalization of the Weingarten formula to the case of Floquet RQC, which allows one to compute the ensemble average of a set of independent two-site local Haar random unitaries. Based on the many-body Weingarten formula, we also give a generalization of the diagrammatic methods of Brouwer and Beenaker, which can be viewed as complimentary to the one proposed by Chan, De Luca and Chalker~\cite{chan2018solution}. Finally, we also prove rigorously the conjecture that in the limit of asymptotically large local Hilbert space dimension, local Haar random gates can be approximed by a global Haar random gate.

\subsection{The Weingarten formula and its diagrammatic representation in RMM\label{sec:Weingarten-RMT}}

Computing the ensemble average of the QFI requires the evaluation
of the average of polynomials of unitary matrix elements over the Haar measure.
This can be done via the celebrated Weingarten formula~\citep{weingarten1978asymptotic,samuel1980unintegrals,brouwer1996diagrammatic},
\begin{equation}
\langle U_{a_{1}b_{1}}\cdots U_{a_{t}b_{t}}U_{\alpha_{1}\beta_{1}}^{*}\cdots U_{\alpha_{t}\beta_{t}}^{*}\rangle=\sum_{P,\,P^{\prime}}V_{P^{-1}P^{\prime}}\prod_{\xi=1}^{t}\delta_{a_{\xi}\alpha_{P(\xi)}}\delta_{b_{\xi}\beta_{P^{\prime}(\xi)}},\label{eq:Wg}
\end{equation}
where $U$ is unitary matrix from the CUE, and the brackets represent ensemble averaging over the CUE. $P$ and $P^{\prime}$ are permutations of the numbers $\{1,2,\cdots t\}$.
$V_{P^{-1}P^{\prime}}$ is known as the Weingarten
function and only depends on the length of the cycles of the permutation
$P^{-1}P^{\prime}$. Note that the ensemble average over a quantity consisting of unequal number of $U$'s and $U^*$'s always vanishes due to the invariance of the Haar measure.

For practical applications, one may need to calculate a linear map
$f(UA_{1},\,\cdots,\,UA_{t},\,U^{\dagger}B_{1},\,\cdots,\,U^{\dagger}B_{t})$,
where $f:\mathcal{L}(\mathcal{H})\times\cdots\times\mathcal{L}(\mathcal{H})\to\mathbb{C}$
and $A_{j}$ and $B_{j}$ are referred to as operators. Here $f$
is linear in its input and hence a polynomial of $U$ of degree $t$.
Without loss of generality, the linear map $f(UA_{1},\,\cdots,\,UA_{t},\,U^{\dagger}B_{1},\,\cdots,\,U^{\dagger}B_{t})$
can be represented as a product of $g$ groups of traces, i.e.,
\begin{align}
 & f\left(UA_{1},\,\cdots,\,UA_{t},\,U^{\dagger}B_{1},\,\cdots,\,U^{\dagger}B_{t}\right)\nonumber \\
= & \prod_{j=1}^{g}\mathrm{Tr}\left(UA_{1}^{(j)}\cdots UA_{n_{j}}^{(j)}U^{\dagger}B_{1}^{(j)}\cdots U^{\dagger}B_{m_{j}}^{(j)}\right)\nonumber \\
= & \prod_{j=1}^{g}\left[U_{a_{1}^{(j)}b_{1}^{(j)}}\cdots U_{a_{n_{j}}^{(j)}b_{n_{j}}^{(j)}}U_{\alpha_{1}^{(j)}\beta_{1}^{(j)}}^{*}\cdots U_{\alpha_{m_{j}}^{(j)}\beta_{m_{j}}^{(j)}}^{*}\right.\nonumber \\
\times & [A_{1}^{(j)}]_{b_{1}^{(j)}a_{2}^{(j)}}[A_{2}^{(j)}]_{b_{2}^{(j)}a_{3}^{(j)}}\cdots[A_{n_{j}}^{(j)}]_{b_{n_{j}}^{(j)}\beta_{1}^{(j)}}\nonumber \\
\times & \left.[B_{1}^{(j)}]_{\alpha_{1}^{(j)}\beta_{2}^{(j)}}[B_{2}^{(j)}]_{\alpha_{2}^{(j)}\beta_{3}^{(j)}}\cdots[B_{m_{j}}^{(j)}]_{\alpha_{m_{j}}^{(j)}a_{1}^{(j)}}\right],\label{eq:f-rep}
\end{align}
where all indices appear twice, implying the Einstein summation convention, and
$\{n_{j}\}_{j=1}^{g}$ and $\{m_{j}\}_{j=1}^{g}$ satisfy
\begin{equation}
\sum_{j=1}^{g}n_{j}=\sum_{j=1}^{g}m_{j}=t.
\end{equation}

Direct application of the Weingarten formula~(\ref{eq:Wg}) to evaluate
the last line in Eq.~(\ref{eq:f-rep}) yields $(t!)^{2}$ terms,
which becomes intractable for large $t$ or large $g$. There is a
diagrammatic method developed by Brouwer and Beenakker~\citep{brouwer1996diagrammatic}
that can help track each term in the summation of Eq.~\eqref{eq:Wg}. With the diagrammatic methods, we
find 
\begin{equation}
\langle\prod_{j=1}^{g}\mathrm{Tr}\left(UA_{1}^{(j)}\cdots UA_{n_{j}}^{(j)}U^{\dagger}B_{1}^{(j)}\cdots U^{\dagger}B_{m_{j}}^{(j)}\right)\rangle=\sum_{PP^{\prime}}V_{P^{-1}P^{\prime}}T_{P,\,P^{\prime}},\label{eq:Wg-ProdTr}
\end{equation}
where $V_{P^{-1}P^{\prime}}$ is the Weingarten function mentioned previously,
and the derivation of and expression for $T_{P,\,P^{\prime}}$ can be found in Appendix~\ref{app:RMT-Diagrammatics}.

We illustrate the diagrammatic method by calculating the QFI for the RMM case at $t=1$. The control and state preparation protocols (Fig.~\ref{fig:metrology-protocols}) are identical when $t=1$, and the QFI can be calculated as follows:
\begin{align}
    F^Q(1)= 4\big( \underbrace{\text{Tr}[\rho_0 U^\dagger H_0^2 U]}_{(*)} - \underbrace{\text{Tr}[\rho_0 U^\dagger H_0 U]\text{Tr}[\rho_0 U^\dagger H_0 U]}_{(**)} \big),
    \label{eq:QFIt1}
\end{align}
where the generic expressions for the QFI for both protocols are given by Eq.~\eqref{eq:ctr-FQ} and Eq.~\eqref{eq:sp-FQ} respectively, $\rho_0$ is density operator for the initial pure state and $H_0$ is the sensing Hamiltonian. Here,  the term $(*)$ corresponds to $g=1$ and $(**)$ corresponds to $g=2$, both of which can be calculated separately using the diagrammatic method.

The diagrammatic method typically consists of the following three steps. A rigorous mathematical justification of the diagrammatic procedures below can be found in Appendix~\ref{app:RMT-Diagrammatics}.
\begin{itemize}
    \item First, express the linear functional $f$ in the canonical form $\prod_{j=1}^{g}\mathrm{Tr}(UA_{1}^{(j)}\cdots UA_{n_{j}}^{(j)}U^{\dagger}B_{1}^{(j)}\cdots U^{\dagger}B_{m_{j}}^{(j)})$, i.e., as in Eq.~\eqref{eq:QFIt1}.
    We can then draw the corresponding \textit{pre-contraction}
    diagrams, where $U$ and $U^{\dagger}$ are represented by dashed lines
    and star-dashed lines respectively and $A_{j}$ and $B_{j}$ are represented by
    the directed thick lines as shown in Fig.~\ref{fig:QFIt1}(a). 
    \item Second, one can pair an empty dot connected to a dashed line to another
    empty dot connected to a star-dashed lines. Similarly, a filled dot
    connected to a dashed line can be paired to another filled dot connected
    to a star-dashed lines. Note that empty or filled dots connected to
    the same type of dashed line (either dashed or star-dashed) are not allowed pairings. If all
    the empty and filled dots are paired, it is called a \textit{complete
    pairing}. A complete pairing corresponds to choosing one permutation
    $P$ and one permutation $P^{\prime}$ in Eq.~(\ref{eq:Wg-ProdTr}). 
    For a given complete pairing, one joins an undirected thin solid line, corresponding to the $\delta$-lines in Fig.~\ref{fig:QFIt1}(a), between the
    paired empty and filled dots to obtain a \textit{fully contracted} diagram, as seen in Fig.~\ref{fig:QFIt1}(b) and (c).
    \item Next, the Weingarten function $V_{P^{-1}P^{\prime}}$
    on the r.h.s.~of Eq.~(\ref{eq:Wg-ProdTr}) can be read off directly
    from a fully contracted diagram. Upon removing the thick directed
    lines in a fully contracted diagram, the rest of the diagram consists
    of a number of $U$-loops, where each $U$-loop is formed by alternating
    thin and dashed lines, e.g., $1\rightarrow2\rightarrow 3\to4$
    in Fig~\ref{fig:QFIt1}(b). One can then read off the length of disjoint
    cycles of $P^{-1}P^{\prime}$, denoted as $u_{1},u_{2},\,\cdots u_{s}$,
    which is equal to half the number of dashed lines and star-dashed lines in
    each $U$-loop. This knowledge of the length of all the cycles in the permutation $P^{-1}P^{\prime}$ completely determines the Weingarten function
    $V_{P^{-1}P^{\prime}}$. \\
    \item Finally, $T_{P,\,P^{\prime}}$ can be also be read off from a fully contracted diagram. Upon removing the dashed lines and star-dashed
    lines in a fully contracted diagram, the rest of the diagram consists
    of a number of $T$-loops, which is formed by alternating thick and
    thin lines, e.g., the loops $2\rightarrow5$ and $4\rightarrow6$
    in Fig~\ref{fig:QFIt1}(b). The contribution of each $T$-loop is the
    trace of the product of the operators along the $T$-loop. $T_{P,\,P^{\prime}}$
    is the product of the contributions from all  $T$-loops.
\end{itemize}

\begin{figure}
    \centering
    \includegraphics[scale=0.8]{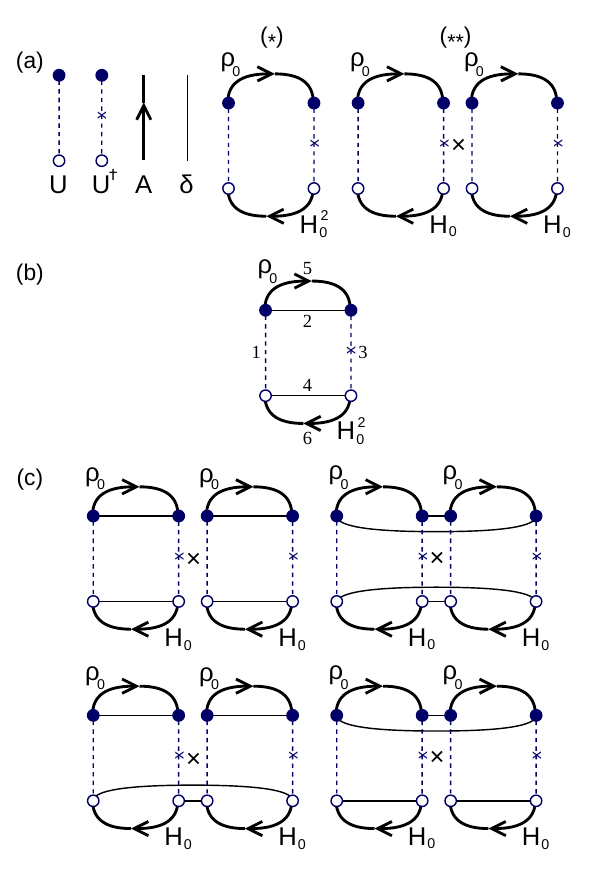}
    \caption{\label{fig:QFIt1}The contraction of the QFI at time $t=1$, as in Eq.~\eqref{eq:QFIt1}. (a) Diagrammatic notation used for the RMM case and pre-contracted diagrams for the terms $(*)$ and $(**)$, (b) only possible contracted diagram for $(*)$ and (c) the four possible contracted diagrams for $(**)$.}
\end{figure}
Following this recipe we find that the full expression of the QFI for the RMM case at $t=1$ is
\begin{align}
    \langle F^Q\rangle = 4\bigg( \frac{\text{Tr}[H_0^2]}{N} + \frac{\text{Tr}[H_0]^2 + \text{Tr}[H_0^2]}{N^2-1} - \frac{\text{Tr}[H_0]^2 + \text{Tr}[H_0^2]}{N(N^2-1)} \bigg),\label{eq:QFI-1time-exact}
\end{align}
where we utilized that $\rho_0$ is a pure state to set $\text{Tr}[\rho_0] = \text{Tr}[\rho_0^2] = 1$.

Note that above result, with $t=1$, has previously been calculated using the Schur-Weyl duality~\citep{oszmaniec2016randombosonic}. For the scaling of the QFI at late times, a pre-contracted diagram of degree
$t$ generates $(t!)^{2}$ fully contracted diagrams, which becomes intractable.
However, when $N$ is large, the number of diagrams contributing
to the leading-order in $N$ can be significantly reduced, which can be
addressed by counting the leading order in $V_{P,\,P^{\prime}}$ and
$T_{P,\,P^{\prime}}$ independently. 

For the Weingarten function, $V_{P^{-1}P^{\prime}}\sim N^{-2t+s}+\mathscr{O}(N^{-2t+s-2})$
if the length of disjoint cycles of $P^{-1}P^{\prime}$ is $u_{1},\,u_{2},\,\cdots,u_{s}$,
where $s\in[1,\,t]$. Clearly the leading-order contribution of the
Weingarten function is $V_{P^{-1}P^{\prime}}\sim N^{-t}$, which corresponds
to $u_{1}=u_{2}=\cdots=u_{s}=1$. Diagrammatically, such a choice
of $P$ and $P^{\prime}$ corresponds to a fully-contracted diagram
where an empty and a filled dot connected to the same dashed line
must pair with an empty and a filled dot connected to the same star-dashed
line. For example the U-loops in the top two contracted diagrams in Fig~\ref{fig:QFIt1}(c). 
Therefore, all the $U$-loops contain only one dashed line and one
star-dashed line, meaning the length of all the disjoint cycles is
$1$. Such diagrams are called \textit{Gaussian}.
Clearly, the number of Gaussian fully-contracted diagrams is $t!$
in a pre-contraction diagram of degree $t$ , which is less than the
total $(t!)^{2}$ diagrams, but still too many to investigate late-time
behaviors of physical quantities. 

For further reduction of the computational
efforts, we need to focus on the leading-order contribution of $T_{P,\,P^{\prime}}$. To
this end, we introduce the \textit{max-$T$ contraction}, which diagrammatically
corresponding to the maximum number of $T$-loops. We consider operators $\{A_j\}$ and $\{B_j\}$ that correspond to physical observables, whose the eignevalues or operator norms $\|A_i\|_{\infty},\|B_i\|_{\infty}$ scale at most as a constant when $N\to\infty$~\cite{reimann2008foundation,reimann2012equilibration}. Physically, RMM is an effective description of the spectral properties of an interacting many-body Hamiltonian. In realistic systems, physical observables are typically local, i.e., their support does not scale with $N$ as the system size goes to infinity, and clearly satisfy above assumptions. There can be multiple max-$T$ contractions and a max-$T$ contraction becomes the leading order of $T_{P,P^{\prime}}$ in the following two scenarios:
\begin{itemize}
\item When the trace of all the operators $A$'s, $B$'s scales linearly as $N$, the maximal possible scaling of each $T$-loop in the max-$T$ contraction is the linear scaling in $N$. Usually, this can be explicitly checked once the max-$T$ contraction is found. If this is the case, the max-$T$ contraction is the leading order contribution of $T_{P,\,P^{\prime}}$.
\item If the trace norm of some operators $C_i$ scale as a constant when $N\to\infty$, such as the density operator, then one can work with the unnormalized version $\widetilde{C_i}=NC_i$. Clearly, the $T$-loop involving only $\widetilde{C}_i$ scales as $N$, while the maximal scaling of any $T$-loop involving $\widetilde{C}_i$ and $A_i$'s or $B_i$'s is $N$, instead of $N^2$. This can be seen from the H\"older inequality
\begin{equation}
|\mathrm{Tr}(\widetilde{C_i}A_1\cdots A_t)|\leq \|\widetilde{C}_i\|_1\|A_1\|_{\infty}\cdots\|A_t\|_{\infty} \lesssim N.
\end{equation}
Then by working with the scaled operator $\widetilde{C}_i$, one can reduce this scenario to the first one.
\end{itemize}
The diagrammatic calculation of the QFI discussed in Sec.~\ref{sec:QFI-scaling}  falls into the second category, as the diagrams involve  the density matrix with $\text{Tr}[\rho_0]=1$. An example of a property that falls into the first category are spectral form factor like quantities (See Appendix~\ref{app:K-illustration} for detailed calculations).
In summary, to identify the leading contribution in the limit $N\to\infty$,
we can simply follow the principle below:

\begin{prp}\label{prp:RM}(\textit{max$-T$ and Gaussianity principle})
If a non-vanishing  max$-T$ contraction is Gaussian, then it must be of the leading
order because the Weingarten functions for the Gaussian diagrams has
the maximum scaling. In other words, to identify the leading-order
diagrams in the large $N$ limit, one should follow two guidelines
(a) Identify the recipe for max $-T$ contractions (\textit{max$-T$
principle}) (b) Choose the Gaussian contractions among all the max-$T$
contractions (\textit{Gaussianity principle}).
\end{prp}

Incidentally, if all the max-$T$, Gaussian contraction vanishes, then we can use the following observation to deduce the maximal possible scaling of $T_{P,P^{\prime}}$:
\begin{obs}\label{obs:vanishing-T}If all the max-$T$, Gaussian contractions vanish due to a $T$-loop consisting of a vanishing trace of one operator or a product of operators, the maximal possible scaling of $f(U)$ must be lower by at least an order $N$ compared to that of a non-vanishing max-$T$, Gaussian contraction.
\end{obs}

We now apply  Principle \ref{prp:RM} to the contractions shown in Fig~\ref{fig:QFIt1} in the limit $N\to\infty$. Due to the gauge-invariance of the QFI, we set $\text{Tr}[H_0] = 0$ and thus $\text{Tr}[H_0^2] \sim N$. We then start by imposing the Gaussianity principle, pairing empty and filled dots from one dashed line with empty and filled dots from one star-dashed line. This leaves us with the diagram in Fig~\ref{fig:QFIt1}(b) and the two upper diagrams in Fig~\ref{fig:QFIt1}(c). Next, we want to find the max$-T$ contractions, from which we deduce that only Fig~\ref{fig:QFIt1}(b) and the upper left diagram in Fig~\ref{fig:QFIt1}(c) have the maximum number of $T$-loops. We then note that the first diagram in Fig~\ref{fig:QFIt1}(c) is vanishing, as it contains the T-loop $\text{Tr}[H_0]$. Then according to Observation~\ref{obs:vanishing-T}, the second diagram in Fig~\ref{fig:QFIt1}(c) is of order $\mathscr{O}(1/N)$, lower than the scaling of the diagram in Fig~\ref{fig:QFIt1}(b). Therefore we can conclude that the leading order scaling of the QFI in $N$ is
\begin{align}
    \langle F^Q(1) \rangle = \frac{4\text{Tr}[H_0^2]}{N} + \mathscr{O}(1/N).
\end{align}

\subsection{Many-body Weingarten formula: Generalization of Weingarten calculus to Floquet RQC \label{sec:Weingarten-many-body}}
We consider a Floquet  random quantum circuit consisting of a set of Haar-distributed 
random \textit{nearest-neighbor} two-site unitary gates acting on
a one-dimensional periodic lattice of qudits. 
To build some intuition, we first
consider the case $L=2$ and apply Weingarten formula~(\ref{eq:Wg})
to the two-site unitary. Once this situation is understood, we generalize to the many-body case with $L>2$. To this end, let us clarify some notation that is used subsequently.
As shown in Fig.~\ref{fig:Floquet-RQC}(a), the matrix element of the two-site unitary acting on sites $\mu$ and
$\nu$ is denoted as $\mathsf{U}_{\bm{a}^{[\mu\nu]}\bm{b}^{[\mu\nu]}}$
or $\mathsf{U}_{a_{\sim\nu}^{[\mu]}a_{\sim\mu}^{[\nu]},\,b_{\sim\nu}^{[\mu]}b_{\sim\mu}^{[\nu]}}$,
where the indices of the two-site unitary are denoted as $\bm{a}^{[\mu\nu]}=(a_{\sim\nu}^{[\mu]},\,a_{\sim\mu}^{[\nu]})$
and $\bm{b}^{[\mu\nu]}=(b_{\sim\nu}^{[\mu]},\,b_{\sim\mu}^{[\nu]})$. We shall call $\bm{a}^{[\mu\nu]}$ or its components the input index and $\bm{{b}}^{[\mu\nu]}$ or its components the output index.
The superscript $[\mu]$ in the indices $a_{\sim\nu}^{[\mu]}$ , $b_{\sim\nu}^{[\mu]}$
means that they live in the Hilbert space associated with site $\mu$
while the subscript $\sim\nu$ in the indices $a_{\sim\nu}^{[\mu]}$
, $b_{\sim\nu}^{[\mu]}$ means that another support of the two-site
unitary is at site $\nu$. Clearly one needs to specify both the superscript
$[\mu]$ and the subscript $\sim\nu$ to uniquely determine the indices of the two-site unitary gates $\mathsf{U}$ or $\mathsf{U^*}$.
\begin{figure}
\centering
\includegraphics[scale=0.7]{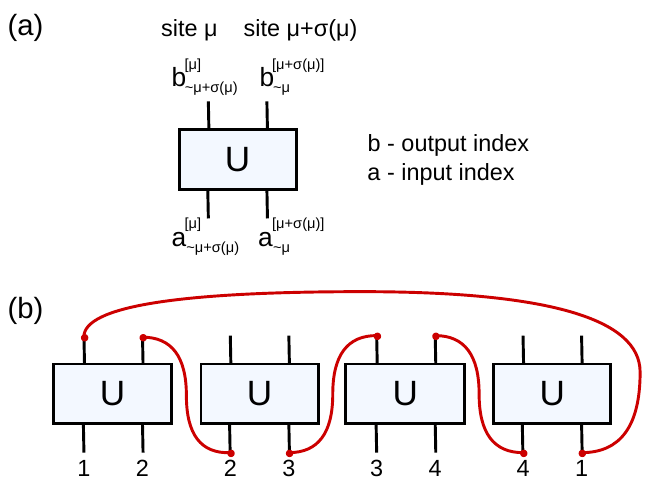}
\caption{\label{fig:Floquet-RQC}Illustration of a Floquet random quantum circuit
on a one-dimensional periodic lattice for $L=4$ with periodic boundary conditions. (a) The indices of a given two-site unitary gate at site $\mu$ and $\mu+\sigma(\mu)$, where $\sigma(\mu)=\pm1$. (b) The diagrammatic representation of the Floquet operator, i.e., the evolution operator of the Floquet RQC in one period, $U_{\{a_{\sim\mu+\sigma(\mu)}^{[\mu]}\},\,\{b_{\sim\mu-\sigma(\mu)}^{[\mu]}\}}$. For each site $\mu$, there is a contraction between the input index  $a_{\sim\nu-\sigma(\mu)}^{[\mu]}$ and the output index $b_{\sim\nu+\sigma(\mu)}^{[\mu]}$ described by the red lines.
In this particular case, $\sigma(1)=1,\,\sigma(2)=-1,\,\sigma(3)=-1,\,\sigma(4)=-1$, where we identify site "0" with site 4 and site "5" with site 1, due to the periodicity of the lattice.}
\end{figure}

For the two-site case, which is equivalent to the usual RMM except that a single index is represented by a two-component vector, the Weingarten formula becomes
\begin{align}
 & \langle\mathsf{U}_{\bm{a}_{1}^{[12]}\bm{b}_{1}^{[12]}}\cdots\mathsf{U}_{\bm{a}_{t}^{[12]}\bm{b}_{t}^{[12]}}\mathsf{U}_{\bm{\alpha}_{1}^{[12]}\bm{\beta}_{1}^{[12]}}^{*}\cdots\mathsf{U}_{\bm{\alpha}_{t}^{[12]}\bm{\beta}_{t}^{[12]}}^{*}\rangle\nonumber \\
= &\sum_{P^{[12]},\,P^{\prime[12]}}V_{P^{-1[12]}P^{\prime[12]}}\prod_{\xi=1}^{t}\delta_{\bm{a}^{[12]}_{\xi}\bm{\alpha}^{[12]}_{P^{[12]}(\xi)}}\delta_{\bm{b}^{[12]}_{\xi}\bm{\beta}^{[12]}_{P^{\prime[12]}(\xi)}},\label{eq:two-site-Wg}
\end{align}
where, again, the superscripts $[12]$ denote the support of the unitary
and the Weingarten functions $V_{P^{-1[12]}P^{\prime[12]}}$ are defined
the same as in the single-site case, i.e., they are permutations of the
$t$-tuple $\{1,2,\,\cdots,\,t\}$, except that the dimension of the
Hilbert space is $q^{2}$ instead of $q$ in the single-site case. 

For a generic case, with $L>2$, one can find the most general Floquet operator, i.e., the evolution operator in one period, as follows: First, one can draw the gates $\mathsf{U}_{\bm{a}^{[\mu,\mu+1]},\bm{b}^{[\mu,\mu+1]}}$ for $\mu=0,1,\cdots,L$,  as shown in Fig.~\ref{fig:Floquet-RQC}(b).
Then it is crucial to observe that to form a Floquet operator, for each site $\mu$, either output index $b_{\sim\mu-1}^{[\mu]}$ should be contracted with the input index $a_{\sim\mu+1}^{[\mu]}$ or the output index $b_{\sim\mu+1}^{[\mu]}$ should be contracted with the input index $a_{\sim\mu-1}^{[\mu]}$. As shown in Fig.~\ref{fig:Floquet-RQC}(b), the brickwork circuit with $L=4$ can be specified in this way, where contractions are denoted by the red lines. In total, there are $2^L$ different Floquet operators, each of which is specified uniquely by the binary function: $\sigma:\{1,\,2,\,\cdots,L\}\to\{1,\,-1\}$.

For site $\mu$, the contracted indices are $a_{\sim\mu-\sigma(\mu)}^{[\mu]}$ and $b_{\sim\mu+\sigma(\mu)}^{[\mu]}$ while the free indices are $a_{\sim\mu+\sigma(\mu)}^{[\mu]}$
and $b_{\sim\mu-\sigma(\mu)}^{[\mu]}$. 
As such the most generic Floquet operator can be written as
\begin{equation}
U_{\{a_{\sim\mu+\sigma(\mu)}^{[\mu]}\},\,\{b_{\sim\mu-\sigma(\mu)}^{[\mu]}\}}=\prod_{\mu=1}^{L}\left(\mathsf{U}_{\bm{a}^{[\mu,\,\mu+1]}\bm{b}^{[\mu,\,\mu+1]}}\delta_{a_{\sim\mu-\sigma(\mu)}^{[\mu]}b_{\sim\mu+\sigma(\mu)}^{[\mu]}}\right),\label{eq:U-ME}
\end{equation}
where we introduce the  compact notation
\begin{equation}
\prod_{\mu=1}^{L}\mathsf{U}_{\bm{a}^{[\mu,\,\mu+1]}\bm{b}^{[\mu,\,\mu+1]}}\equiv\mathsf{U}_{\bm{a}^{[L,\,L+1]}\bm{b}^{[L,\, L+1]}}\cdots\mathsf{U}_{\bm{a}^{[23]}\bm{b}^{[23]}}\mathsf{U}_{\bm{a}^{[12]}\bm{b}^{[12]}},
\end{equation}
with site $L+1$ being identified with site $1$, and where $\delta_{a_{\sim\mu+\sigma(\mu)}^{[\mu]}b_{\sim\mu-\sigma(\mu)}^{[\mu]}}$ implements the contractions between the input and output indices $a_{\sim\mu+\sigma(\mu)}^{[\mu]}$ and $b_{\sim\mu-\sigma(\mu)}^{[\mu]}$ at site $\mu$. 

The many-body analogy of Eq.~(\ref{eq:Wg}) in the most generic case
is
\begin{widetext}
\begin{align}
 & \langle U_{\{a_{1,\,\sim\mu+\sigma(\mu)}^{[\mu]},\,b_{1,\,\sim\mu-\sigma(\mu)}^{[\mu]}\}}\cdots U_{\{a_{t,\,\sim\mu+\sigma(\mu)}^{[\mu]},\,b_{t,\,\sim\mu-\sigma(\mu)}^{[\mu]}\}}U_{\{\alpha_{1,\,\sim\mu+\sigma(\mu)}^{[\mu]},\,\beta_{1,\,\sim\mu-\sigma(\mu)}^{[\mu]}\}}^{*}\cdots U_{\{\alpha_{t,\,\sim\mu+\sigma(\mu)}^{[\mu]},\,\beta_{t,\,\sim\mu-\sigma(\mu)}^{[\mu]}\}}^{*}\rangle\nonumber \\
= & \sum_{P^{[12]},\,P^{\prime[12]}}\sum_{P^{[23]},\,P^{\prime[23]}}\cdots\sum_{P^{[L1]},\,P^{\prime[L1]}}\left[\left(\prod_{\mu=1}^{L}V_{P^{-1[\mu,\,\mu+1]}P^{\prime[\mu,\,\mu+1]}}q^{\bar{s}^{[\mu,\,\mu\pm\sigma(\mu)]}}\right)\prod_{\xi=1}^{t}\prod_{\mu=1}^{L}\left(\delta_{a_{\xi,\,\sim\mu+\sigma(\mu)}^{[\mu]}\alpha_{P^{[\mu,\,\mu+\sigma(\mu)]}(\xi),\,\sim\mu+\sigma(\mu)}^{[\mu]}}\delta_{b_{\xi,\,\sim\mu-\sigma(\mu)}^{[\mu]}\beta_{P^{\prime[\mu,\,\mu-\sigma(\mu)]}(\xi),\,\sim\mu-\sigma(\mu)}^{[\mu]}}\right)\right],\label{eq:Wg-Manybody}
\end{align}
\end{widetext}
where $s^{[\mu,\,\mu+\sigma(\mu)]}$ denotes the number of disjoint cycles in the permutation ${P^{-1[\mu,\,\mu+\sigma(\mu)]}P^{\prime[\mu,\,\mu+\sigma(\mu)]}}$ that determines the scaling of the Weingarten function in the limit $q\to\infty$ and $\bar{s}^{[\mu,\,\mu\pm\sigma(\mu)]}$ denotes the number disjoint cycles in the permutation ${P^{-1[\mu,\,\mu-\sigma(\mu)]}P^{\prime[\mu,\,\mu+\sigma(\mu)]}}$. A few comments are in order: Eq.~\eqref{eq:Wg-Manybody} holds for arbitary $q$. Comparing to the Weingarten formula~\eqref{eq:Wg} for the RMM case, an additional term $q^{\bar{s}^{[\mu,\,\mu\pm\sigma(\mu)]}}$ appears due to the contribution from the contracted dummy indices between the two-site unitary gates $\mathsf{U}_{\bm{a}^{[\mu,\,\mu\pm1]}\bm{b}^{[\mu,\,\mu\pm1]}}$. The detailed proof is presented in Appendix~\ref{app:RQC-Weingarten}. The basic idea of the proof is the following: First by substituting Eq.~\eqref{eq:U-ME} into the l.h.s.~of Eq.~\eqref{eq:Wg-Manybody}, one can apply the Weingarten formula~\eqref{eq:Wg} to the ensemble average of each individual two-site unitary gate. This process can be carried out even without contracting the indices described by the red lines shown in Fig.~\ref{fig:Floquet-RQC}(b). After performing the index contraction, which is characterized by $\prod_{\mu=1}^L\delta_{a_{\sim\mu-\sigma(\mu)}^{[\mu]}b_{\sim\mu+\sigma(\mu)}^{[\mu]}}$, one can manipulate the dummy indices and show that they give rise to the contribution $\prod_{\mu=1}^L q^{\bar{s}^{[\mu,\,\mu\pm\sigma(\mu)]}}$.

\subsection{Proving the global CUE empirical conjecture \label{sec:globalCUE}}
Eq.~(\ref{eq:Wg-Manybody}) immediately allows us to prove an empirical
conjecture, that the ensemble average of the Floquet operators in a Floquet
RQC in the large $q$ limit behaves like a global CUE~\citep{chan2018solution}. Recall that in the limit $q\to\infty$, $V_{P^{-1[\mu,\,\mu+1]}P^{\prime[\mu,\,\mu+1]}}\propto (q^2)^{-2t+s^{[\mu,\,\mu+\sigma(\mu)]}}$,where $s^{[\mu,\,\mu+\sigma(\mu)]},\,\bar{s}^{[\mu,\,\mu\pm\sigma(\mu)]} \in [1,t]$. Thus, the leading order scaling of the prefactor in Eq.~(\ref{eq:Wg-Manybody}) is
\begin{equation}
V_{P^{-1[\mu,\,\mu+1]}P^{\prime[\mu,\,\mu+1]}}q^{\bar{s}^{[\mu,\,\mu\pm\sigma(\mu)]}}\propto q^{-t},\, \forall \mu\in [1,L],
\end{equation}
which corresponds to $s^{[\mu,\,\mu+\sigma(\mu)]}=\bar{s}^{[\mu,\,\mu\pm\sigma(\mu)]}=t$, i.e., the Gaussian contractions
\begin{equation}
P^{[\mu,\,\mu+1]}=P^{\prime[\mu,\,\mu+1]},\,P^{[\mu,\,\mu-\sigma(\mu)]}=P^{\prime[\mu,\,\mu+\sigma(\mu)]},\,\forall \mu\in [1,L].
\end{equation}
This essentially implies that the permutations for all two-site Haar random unitary gates collapse to a single-permutation $P^{[\mu,\,\mu-\sigma(\mu)]}=P^{\prime[\mu,\,\mu-\sigma(\mu)]}=P^{\prime[\mu,\,\mu+\sigma(\mu)]}=P$, which leads to
\begin{widetext}
\begin{equation}
\langle U_{\{a_{1,\,\sim\mu+\sigma(\mu)}^{[\mu]},\,b_{1,\,\sim\mu-\sigma(\mu)}^{[\mu]}\}}\cdots U_{\{a_{t,\,\sim\mu+\sigma(\mu)}^{[\mu]},\,b_{t,\,\sim\mu-\sigma(\mu)}^{[\mu]}\}}U_{\{\alpha_{1,\,\sim\mu+\sigma(\mu)}^{[\mu]},\,\beta_{1,\,\sim\mu-\sigma(\mu)}^{[\mu]}\}}^{*}\cdots U_{\{\alpha_{t,\,\sim\mu+\sigma(\mu)}^{[\mu]},\,\beta_{t,\,\sim\mu-\sigma(\mu)}^{[\mu]}\}}^{*}\rangle=q^{-t}\sum_P\prod_{\xi=1}^{t}\prod_{\mu=1}^{L}\left(\delta_{a_{\xi,\,\sim\mu+\sigma(\mu)}^{[\mu]}\alpha_{P(\xi),\,\sim\mu+\sigma(\mu)}^{[\mu]}}\delta_{b_{\xi,\,\sim\mu-\sigma(\mu)}^{[\mu]}\beta_{P(\xi),\,\sim\mu-\sigma(\mu)}^{[\mu]}}\right).\label{eq:global-CUE}
\end{equation}
On the other hand, let us first note that if
$\mathsf{U}_{\{a_{\xi,\,\sim\mu+\sigma(\mu)}^{[\mu]},\,b_{\xi,\,\sim\mu-\sigma(\mu)}^{[\mu]}\}}$ are the matrix elements of a global unitary gate that belongs to the
CUE of size $q^{L}$, we would obtain 
\begin{align}
 & \langle U_{\{a_{1,\,\sim\mu+\sigma(\mu)}^{[\mu]},\,b_{1,\,\sim\mu-\sigma(\mu)}^{[\mu]}\}}\cdots U_{\{a_{t,\,\sim\mu+\sigma(\mu)}^{[\mu]},\,b_{t,\,\sim\mu-\sigma(\mu)}^{[\mu]}\}}U_{\{\alpha_{1,\,\sim\mu+\sigma(\mu)}^{[\mu]},\,\beta_{1,\,\sim\mu-\sigma(\mu)}^{[\mu]}\}}^{*}\cdots U_{\{\alpha_{t,\,\sim\mu+\sigma(\mu)}^{[\mu]},\,\beta_{t,\,\sim\mu-\sigma(\mu)}^{[\mu]}\}}^{*}\rangle\nonumber \\
= & \sum_{P,\,P^{\prime}}V_{P^{-1}P^{\prime}}\prod_{\xi=1}^{t}\prod_{\mu=1}^{L}\left(\delta_{a_{\xi,\,\sim\mu+\sigma(\mu)}^{[\mu]}\alpha_{P(\xi),\,\sim\mu+\sigma(\mu)}^{[\mu]}}\delta_{b_{\xi,\,\sim\mu-\sigma(\mu)}^{[\mu]}\beta_{P^{\prime}(\xi),\,\sim\mu-\sigma(\mu)}^{[\mu]}}\right).\label{eq:Global-CUE-Ansatz}
\end{align}
\end{widetext}
Clearly, in the limit $q\to\infty$, the leading order of the Weingarten function agrees with Eq.~\eqref{eq:global-CUE}.

It was shown in Ref.~\cite{harrow2009randomquantum,brandao2016localrandom} that RQC becomes an approximate unitary $k$-design for time $t$ of the polynomial order of $k$ and $L$. Our results are orthogonal and complementary to their results in that for asymptotically large local Hilbert space dimension, Floquet RQC asymptotically becomes a global random Haar gate (i.e., unitary design with arbitrary $k$) for any time $t$.

\subsection{Diagrammatic representation of the many-body Weingarten formula \label{sec:Weingarten-Diagram-FloquetRQC}}
Having generalized the Weingarten formula to the case of Floquet circuits,
it is natural to ask how to use the diagrammatic methods to compute
the ensemble average over the Floquet operators in a Floquet RQC and what the corresponding many-body generalizations of Eq.~(\ref{eq:Wg-ProdTr}) is.

Here, we consider the case where the operators {$\{A_j\}$} and {$\{B_j\}$,} that were previously discussed in the context of RMM, are separable, i.e., they are only supported on single site and take the form $A=\prod_{\mu=1}^L A^{[\mu]}$ and $B=\prod_{\mu=1}^L B^{[\mu]}$. For the sake of simplicity, we consider a linear map of the form $f=\mathrm{Tr}(UA_{1}\cdots UA_{t}U^{\dagger}B_{1}\cdots U^{\dagger}B_{t})$. Upon substituting this into Eq.~\eqref{eq:U-ME} and exploiting the fact that the operators {$\{A_i\}$} and {$\{B_i\}$} are only supported on single sites, the ensemble average of $f$ factorizes into a product of single-site diagrams that is almost the same as the  RMM case, except that the contractions at each site are constrained due to the so-called "bond constraint" between sites~\cite{chan2018solution}. More quantitatively, 
in Appendix~\ref{app:RQC-Weingarten}, we show that in the limit
$q\to\infty$ the many-body generalization of Eq.~(\ref{eq:Wg-ProdTr})
becomes
\begin{widetext}
\begin{equation}
\langle\prod_{j=1}^{g}\mathrm{Tr}\left(UA_{1}^{(j)}\cdots UA_{n_{j}}^{(j)}U^{\dagger}B_{1}^{(j)}\cdots U^{\dagger}B_{m_{j}}^{(j)}\right)\rangle=\sum_{P_{1}^{[1]},\,P^{\prime[1]}}\sum_{Q^{[1]},\,Q^{\prime[1]}}\cdots\sum_{P_{1}^{[L]},\,P^{\prime[L]}}\sum_{Q^{[L]},\,Q^{\prime[L]}}\mathcal{B}[\{P^{[\mu]},\,Q^{[\mu]}\}]\prod_{\mu=1}^{L}V_{P^{[\mu]^{-1}}P^{\prime[\mu]}}V_{Q^{[\mu]-1}Q^{\prime[\mu]}}T_{P^{[\mu]}P^{\prime[\mu]}Q^{[\mu]}Q^{\prime[\mu]}},\label{eq:Wg-ProdTr-Manybody}
\end{equation}
\end{widetext}
where 
\begin{equation}
\mathcal{B}[\{P^{[\mu]},\,Q^{[\mu]}\}]\equiv\prod_{\mu=1}^{L}\delta_{P^{[\mu]},\,Q^{[\mu+1]}}\delta_{P^{\prime[\mu]},\,Q^{\prime[\mu+1]}},
\end{equation}
is the precise mathematical characterization of the bond constraints discussed in Ref~\cite{chan2018solution} and $T_{P^{[\mu]}P^{\prime[\mu]}Q^{[\mu]}Q^{\prime[\mu]}}$ denotes the $T$-loop for site $\mu$ (see Eq.~(\ref{eq:T-PP-Many-body}) for an explicit expression). Note that unlike Eq.~\eqref{eq:Wg-Manybody}, the contracted indices between the two-site unitary gates characterized by $\prod_{\mu=1}^L\delta_{a_{\sim\mu-\sigma(\mu)}^{[\mu]}b_{\sim\mu+\sigma(\mu)}^{[\mu]}}$ are contained in $T_{P^{[\mu]}P^{\prime[\mu]}Q^{[\mu]}Q^{\prime[\mu]}}$, rather than explicitly evaluated.   

Similarly to the RMM case $T_{P^{[\mu]}P^{\prime[\mu]}Q^{[\mu]}Q^{\prime[\mu]}}$
is a product of traces of the multiplication of operators, since all the indices for the matrices appear
only twice and the Einstein summation convention is used. Eq.~(\ref{eq:T-PP-Many-body})
suggest the following diagrammatic representation in the limit $q\to\infty$:
\begin{itemize}
\item First, one can draw the many-body pre-contraction diagrams for
$f\left(UA_{1},\,\cdots,\,UA_{t},\,U^{\dagger}B_{1},\,\cdots,\,U^{\dagger}B_{t}\right)$,
as shown in Fig.~\ref{fig:many-body-diagram}. Thanks to the fact
that the operators are separable, such a many-body pre-contraction
diagram factorizes into a product of single-site pre-contraction diagrams.
\item Then, in principle, one can start the contraction for each single-site
pre-contraction diagram, where $T_{P^{[\mu]}P^{\prime[\mu]}Q^{[\mu]}Q^{\prime[\mu]}}$
is represented by the $T$-loops at each site, as in the RMM case. Note that for a given site $\mu$, there are two Weingarten functions $V_{P^{[\mu]^{-1}}P^{\prime[\mu]}}V_{Q^{[\mu]-1}Q^{\prime[\mu]}}$, corresponding the two independent unitary gates acting on site $\mu$, $\mu+1$ and $\mu-1$, $\mu$. Similarly to
the RMM case, they are still represented by the $U$-loops in the
single-site diagrams.
\item Finally, if the contractions at all sites satisfy the bond constraints the global many-contraction is legitimate with $\mathcal{B}[\{P^{[\mu]},\,Q^{[\mu]}\}]=1$,
otherwise $\mathcal{B}[\{P^{[\mu]},\,Q^{[\mu]}\}]=0$.
\end{itemize}
The above procedure can be summarized in the following principle:

\begin{prp}\label{prp:many-body}(Factorization Principle) Assuming
all the operators are single site operators, then, in the large
$q$ limit, the contribution of a fully-contracted many-body diagram
is the multiplication of the contributions of single-site diagrams,
provided the contractions at different sites meet the bond constraint.

\end{prp}

In principle, if one were to use Eq.~(\ref{eq:Wg-ProdTr-Manybody}),
the number of permutations for all the sites that needs to be
taken into account is $(t!)^{2L}$, which is impractical. {Nevertheless, if we consider the limit $q\to \infty$ in combination with the RMM case max$-T$ and Gaussianity principle, it is possible to identify the dominant, leading order contributions.} Note that the ensemble average involving generic, non-separable operators can also be evaluated in this way by expressing them in the basis of separable operators, see Sec.~\ref{sec:ent-RQC} for an example.

We conclude this section with a few remarks. First, Eqs.~\eqref{eq:Wg-Manybody} and (\ref{eq:Wg-ProdTr-Manybody}), the equivalence to global CUE
in the large $q$ limit, and the factorization principle do not
only apply to the brickwork random quantum circuits discussed in~\citep{chan2018solution},
but for all the $2^{L}$ different configuration of Floquet RQC with random nearest neighbor interactions discussed previously.
Furthermore, Eq.~(\ref{eq:Wg-Manybody})
can also be used to identify non-Gaussian corrections beyond the prediction by the global CUE equivalence. Secondly, it is also possible to further generalize the above results
to two-site random unitary gates acting on distant sites, beyond the
nearest-neighbor setup. 

\section{Asymptotoic scaling of the QFI in chaotic Floquet dynamics modeled by RMM and RQC\label{sec:QFI-scaling}}
We now address the scaling of the QFI in the control and state-preparation protocols for the RMM and RQC cases, separately. From previous discussions in Sec.~\ref{sec:globalCUE}, we expect that the QFI for the two protocols applied to RQC can be captured by modeling the two layers of Haar-distributed local unitaries acting on even and odd bonds sequently by a global CUE in the RMM. Indeed, as will see later, {the QFI for the RQC case coincide with the one for the RMM to the leading order in large $q$.}

\subsection{Asymptotic scaling of QFI in RMM\label{sec:scaling-RMT}}
 
We consider two fundamental metrological protocols, as shown in Fig.~\ref{fig:metrology-protocols}.
For the control protocol, at time $t$, the unitary operator is given
by $\mathcal{U}_{\mathrm{ctr}}(t)=(WU)^{t}$, where the gate used
to sense the signals is given by $W=\mathrm{e}^{-\mathrm{i}H_{0}\theta}$.
{We model $U$ by a global CUE matrix acting on the Hilbert space of the entire system. The gate used
to sense the signals $W=\bigotimes W_0^{[\mu]}$, where $W_0^{[\mu]}$ is a single-qudit sensing gate acting on qudit $\mu$ only. We can express $W$ as $W=\mathrm{e}^{-\mathrm{i}H_{0}\theta}$, where
\begin{equation}
    H_0=\sum_{\mu} \mathbb{I}\otimes \cdots \otimes h_0^{[\mu]}\otimes \cdots \mathbb{I},\label{eq:H0-h0-connection}
\end{equation}
is given by the summation of local Hamiltonian $h_0^{[\mu]}$ acting on qudit $\mu$.
We consider an initial pure state denoted as $\rho_{0}$. Without loss of generality, we can take
$\mathrm{Tr}[H_{0}]=0$ without changing the scaling of the QFI. Furthermore, we note that all of the eigenvalues of $H_0^2$ are non-negative and the majority of them are of the order $1$, because $H_0$ correspond to some local observable. Therefore, $\mathrm{Tr}[H_{0}^{2}]\sim N$.
According to Eq.~\eqref{eq:QFI-def-in-G}, the QFI is~\citep{boixo2007generalized}
\begin{equation}
F_{\mathrm{ctr}}^{Q}(t)=4\mathrm{Var}[G_{\mathrm{ctr}}(t)]_{\rho_{0}}\label{eq:ctr-FQ},
\end{equation}
where $G_{\mathrm{ctr}}(t)\equiv\mathrm{i}\partial_{\theta}\mathcal{U}_{\mathrm{ctr}}^{\dagger}(t)\mathcal{U}_{\mathrm{ctr}}(t)=\sum_{s=1}^{t}H_{\mathrm{ctr}}(s)$
is the metrological generator and $H_{\mathrm{ctr}}(s)\equiv(U^{\dagger}W^{\dagger})^{s}H_{0}(WU)^{s}$
is the representation of $H_{0}$ in the Heisenberg picture at time
$s$. Then it can be readily obtained that 
\begin{equation}
F_{\mathrm{ctr}}^{Q}(t)=4\sum_{s,\,\tau=1}^{t}\mathrm{Cov}[H_{\mathrm{ctr}}(s),\,H_{\mathrm{ctr}}(\tau)]_{\rho_{0}},\label{eq:FQ-ctr-corr}
\end{equation}
where $\mathrm{Cov}[AB]_{\rho_{0}}\equiv\frac{1}{2}\mathrm{Tr}[\rho_{0}\{A,B\}]-\mathrm{Tr}[\rho_{0}A]\mathrm{Tr}[\rho_{0}B]$.

For the state-preparation protocol, we know that $\mathcal{U}_{\mathrm{sp}}(t)=W^{t}U^{t}$. It is straightforward to
calculate that 
\begin{equation}
F_{\mathrm{sp}}^{Q}(t)=4\mathrm{Var}[G_{\mathrm{sp}}(t)]_{\rho_{0}},\label{eq:sp-FQ}
\end{equation}
where $G_{\mathrm{sp}}(t)\equiv t U^{\dagger t} H_{0}U^{t}=t H_{\mathrm{sp}}(t)$ and $H_{\mathrm{sp}}(t)\equiv U^{\dagger t}H_{0}U^{t}$. Alternatively, 
\begin{equation}
    F_{\mathrm{sp}}^{Q}(t)=4t^{2}\mathrm{Var}[H_{\mathrm{sp}}(t)]_{\rho_{0}}.\label{eq:FQ-sp-corr}
\end{equation} 

In the local estimation approach, one
is usually concerned with the small parameter limit where $\theta$
is very close to zero. If $s\theta\Delta_{H_{0}}\ll1$, where $\Delta_{H_{0}}$
is the spectral width of $H_{0}$, i.e., the difference between the
maximum and minimum eigenvalues of $H_{0}$, we also have $H_{\mathrm{sp}}(s)=H_{\mathrm{ctr}}(s)$ . In general, $H_{\mathrm{sp}}(t)\neq H_{\mathrm{ctr}}(t)$ unless $t=0$ or $t=1$. Here, we make the useful observation that upon setting $\theta=0$ and $W=\mathbb{I}$,  $H_{\mathrm{sp}}(t)=H_{\mathrm{ctr}}(t)$ for arbitrary $t$. 
Therefore, the scaling for the state-preparation protocol can be obtained as a special case of the control protocol, as we will see later. 

Now we are in a position to employ the diagrammatic method discussed in Sec.~\ref{sec:Weingarten-RMT} to calculate the average value of the QFI in the limit of large Hilbert space dimension $N$. Let us first state the results of the scaling of the QFI for the control protocol and the state preparation protocol, respectively,
\begin{align}
    \langle F^Q_{\text{ctr}}(t)\rangle =& \frac{4\text{Tr}[H_0^2]}{N} t + \mathscr{O}(1/N),
    \label{eq:QFI_rmt_ctr} \\
    \langle F^Q_{\text{sp}}(t)\rangle =& \frac{4\text{Tr}[H_0^2]}{N} t^2 + \mathscr{O}(1/N). 
    \label{eq:QFI_rmt_sp}
\end{align}
i.e., the control protocol on average scales linearly with time whereas the state preparation protocol will scale quadratically.

To show the above results, let us first focus on the control protocol. We separate the sum in Eq.~\eqref{eq:FQ-ctr-corr} into two parts
\begin{widetext}
\begin{align}
    \sum_{s,\tau=1}^t \text{Cov}[H_{\text{ctr}}(s), H_{\text{ctr}}(\tau)]_{\rho_0} = & \sum_{s,\tau=1}^t \quad\text{Tr}[\rho_0H_{\text{ctr}}(s) H_{\text{ctr}}(\tau)] + \sum_{s,\tau=1}^t \text{Tr}[\rho_0 H_{\text{ctr}}(s)]\text{Tr}[\rho_0 H_{\text{ctr}}(\tau)] \nonumber \\
    =&  \sum_{s,\tau=1}^t \underbrace{\text{Tr}[\rho_0 (U^\dagger W^\dagger)^{s-1} U^\dagger H_0 (W^\dagger U^\dagger)^{\tau-s} H_0 U (WU)^{\tau-1}]}_{\text{(i)}} \nonumber \\
    &+ \sum_{s,\tau=1}^t \underbrace{\text{Tr}[\rho_0 (U^\dagger W^\dagger)^{s-1} U^\dagger H_0 U (WU)^{s-1}]  \text{Tr}[\rho_0 (U^\dagger W^\dagger)^{\tau-1} U^\dagger H_0 U (WU)^{\tau-1}]}_{\text{(ii)}}, \label{eq:rmt_ctr}
\end{align}
\end{widetext}
where we have inserted $H_{\text{ctr}}(s) \equiv (U^\dagger W^\dagger)^sH_0 (WU)^s$. We can now employ the diagrammatic notation to evaluate the terms (i) and (ii) separately.

\begin{figure}[h]
    \centering    \includegraphics[width=0.49\textwidth]{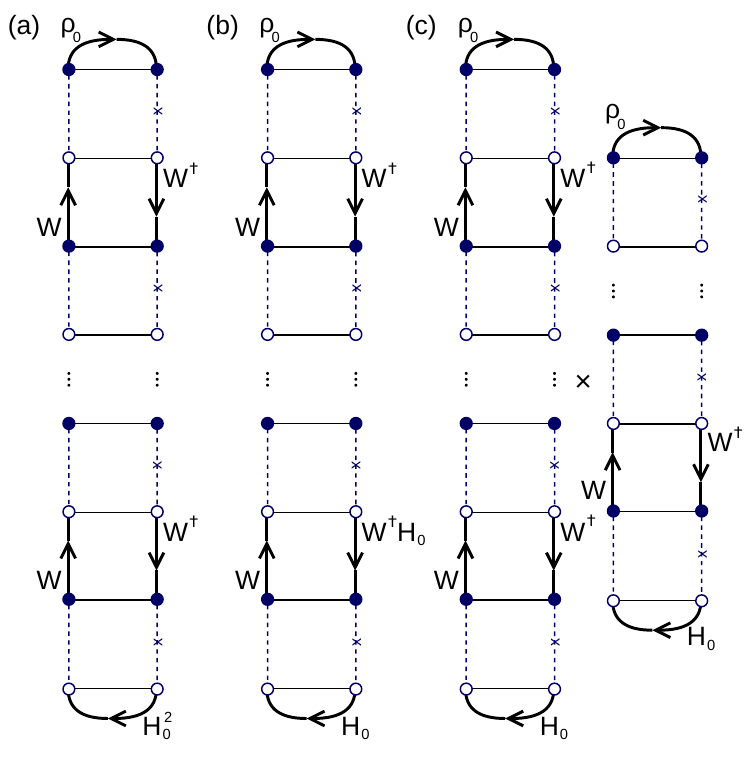}
    \caption{Leading order diagrams in $N$ of the QFI for the control protocol (Eq.~\eqref{eq:FQ-ctr-corr}) for arbitrary time $t$, (a) term (i) when $\tau=s$, (b) term (i) when $\tau\neq s$ and c) term (ii) respectively. Diagrams (b) and (c) are always vanishing due to the T-loop $\text{Tr}[H_0]=0$; diagram (b) have a non-vanishing contribution.}
    \label{fig:RMT_ctr_diagram}
\end{figure}
We note that $\mathrm{Tr}[\rho_0]=1$, which falls into the second scenario discussed in Sec.~\ref{sec:Weingarten-RMT}.
The leading order contracted diagrams for term (i) at $\tau=s$ and $\tau\neq s$ can be seen in Fig.~\ref{fig:RMT_ctr_diagram}(a) and (b) respectively, and term (ii) can be seen in Fig.~\ref{fig:RMT_ctr_diagram}(c), all for an arbitrary time $t$.
These diagrams are the only ones that both fulfill the Gaussianity principle and maximize the number of T-loops, and in accordance with Sec.~\ref{sec:Weingarten-RMT}, we therefore conclude that if they are non-vanishing, they are of the leading order. 
Note that with this contraction, the diagrams in Fig.~\ref{fig:RMT_ctr_diagram}(b) and (c) will always contain the T-loop $\text{Tr}[H_0]=0$, which means that these terms are vanishing for all $\tau \neq s$ for term (i) and all $\tau,\, s$ for term (ii). Further, they follow the scaling in Observation~\ref{obs:vanishing-T}, which is of subleading order compared to Fig.~\ref{fig:RMT_ctr_diagram}(a). 
For Fig.~\ref{fig:RMT_ctr_diagram}(a), we find a non-vanishing diagrammatic contribution of $N^{-1}\text{Tr}[H_0^2]$, which is the only leading order contributions in $N$. We therefore note that the scaling of the two terms are
\begin{align}
\langle\mathrm{Tr}[\rho_{0}H_{\mathrm{ctr}}(s)H_{\mathrm{ctr}}(\tau)]\rangle & \sim\frac{\mathrm{Tr}(H_{0}^{2})}{N}\delta_{t,\,s}+\mathscr{O}(1/N),\label{eq:RMT-Moment}\\
\langle\mathrm{Tr}[\rho_{0}H_{\mathrm{ctr}}(s)]\mathrm{Tr}[\rho_{0}H_{\mathrm{ctr}}(\tau)]\rangle & \sim\mathscr{O}(1/N),\,\forall\tau,\,s,\label{eq:RMT-Average}
\end{align}
which in the end gives the result in Eq.~\eqref{eq:QFI_rmt_ctr}.

Let us now discuss the state preparation protocol. Eq.~\eqref{eq:FQ-sp-corr} has two contributing terms
\begin{align}
    \text{Var}[H_{\text{sp}}(t)]_{\rho_0}
    =&\underbrace{\text{Tr}[\rho_0 U^{\dagger t} H_0^2 U^t]}_{\text{(iii)}} + \underbrace{\text{Tr}[\rho_0 U^{\dagger t} H_0 U^t]^2}_{\text{(iv)}},
\end{align}
where we have used that $H_{\text{sp}}(t) = U^{\dagger t}H_0 U^t$. We note that the correlation function for the state preparation protocol can be obtained as a special case from the one in the control protocol by setting $\tau=s=t$ and $W=\mathbb{I}$, therefore the diagrams for terms (iii) and (iv) are identical to those in Fig.~\ref{fig:RMT_ctr_diagram}(a) and (c) except that $W=\mathbb{I}$. Therefore, similarly to the control protocol, the term (iv) will vanish as it contains the T-loop $\text{Tr}[H_0]=0$. The term (iii) contributes a factor of $N^{-1}\text{Tr}[H_0^2]$, giving a final contribution from each term of
\begin{align}
\langle\mathrm{Tr}[\rho_{0}H_{\mathrm{sp}}(t)^2]\rangle & \sim\frac{\mathrm{Tr}(H_{0}^{2})}{N}+\mathscr{O}(1/N),\label{eq:RMT-Moment}\\
\langle\mathrm{Tr}[\rho_{0}H_{\mathrm{sp}}(t)]^2\rangle & \sim\mathscr{O}(1/N),\,\forall\tau,\,s.\label{eq:RMT-Average}
\end{align}
By including the scaling with time, $t^2$, we find the final expression as in Eq.~\eqref{eq:QFI_rmt_sp}.

Numerical simulations shown in Fig~\ref{fig:RMT_numerical} are in perfect agreement with the analytical results, where the sensing Hamiltonian is, without loss of generality,
\begin{align}
    H_0 = 
    \begin{pmatrix}
        \mathbb{I}_{N/2} & 0 \\
        0 & -\mathbb{I}_{N/2}
    \end{pmatrix}.
\end{align}
As seen in the graph, the scaling in time is linear for the control protocol and quadratic for the state preparation protocol, and the simulated scaling is in great agreement with the fits.
\begin{figure}[h]
    \centering    \includegraphics[width=0.48\textwidth]{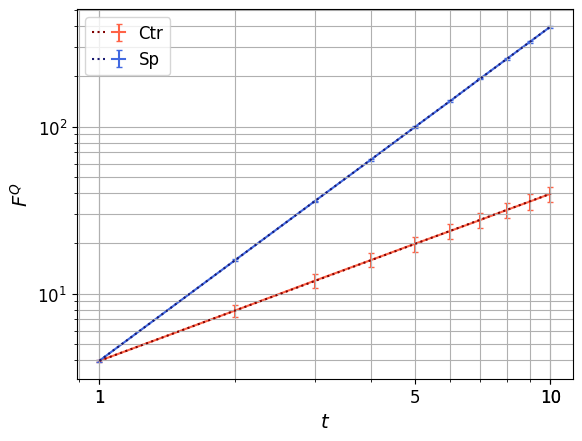}
    \caption{Numerical results for the RMM case. The QFI as a function of time for the control protocol (Ctr) and state preparation protocol (Sp), including linear fit for control protocol and quadratic fit for the state preparation protocol (dotted lines). As predicted by the analytics in the main text, the average QFI scaling with time $t$ is linear for the control protocol and quadratic for the state preparation protocol. Parameters: Hilbert space dimension $N=100$, encoding parameter $\theta=1$ and maximum time step $t=10$.
    The average was done over $100$ simulation runs.}
    \label{fig:RMT_numerical}
\end{figure}

We conclude by making a few remarks. First, it can be shown that (see Appendix~\ref{sec:bounds-QFI}) 
\begin{equation}
F_{\mathrm{ctr}}^{Q}(t)\!\leq \!4(\sum_{s=1}^{t}\!\sqrt{\mathrm{Var}[H_{\mathrm{ctr}}(s)]_{\rho_{0}}})^{2}\!\leq \!4t\sum_{s=1}^{t}\mathrm{Var}[H_{\mathrm{ctr}}(s)]_{\rho_{0}}\!\leq\!t^{2}\Delta_{H_{0}}^{2},\label{eq:FQ-bound}
\end{equation}
and 
\begin{equation}
F_{\mathrm{sp}}^{Q}(t)\leq t^2 \Delta_{H_{0}}^{2}.
\end{equation}
The bound $F_{\mathrm{ctr}}^{Q}(t),\,F_{\mathrm{sp}}^{Q}(t)\leq t^{2}\Delta_{H_{0}}^{2}$
in the case of Hamiltonian dynamics was first found by Boixo et al~\citep{boixo2007generalized}.
Its saturability condition was later discussed in Ref.~\citep{yuan2015optimal,pang2017optimal,yang2017quantum}.
It can be achieved if $[U,\,W]=0$ and $\rho_{0}$ is prepared in
the equal superposition of the eigenstates of $H_{0}$ corresponding
to maximum and minimum eigenvalues. 

Next, in general there is
no universal ordering between $F_{\mathrm{ctr}}^{Q}(t)$ and $F_{\mathrm{sp}}^{Q}(t)$.
For example, consider the limit $1\leq t\ll\lfloor1/\theta\Delta_{H_{0}}\rfloor$, where $H_{\mathrm{sp}}(t)=H_{\mathrm{ctr}}(t)$.
One then can choose $\rho_{0}$ and $U$ such that $U^{t}\rho_{0}U^{\dagger t}$ is either the superposition of the eigenstates of $H_{0}$ or simply one of the eigenstates of $H_{0}$.
In the former case, $\mathrm{Var}[H_{\mathrm{sp}}(t)]_{\rho_{0}}=\Delta_{H_{0}}^{2}/4$ and $F_{\mathrm{ctr}}^{Q}(t)\leq F_{\mathrm{sp}}^{Q}(t)=t^{2}\Delta_{H_{0}}^{2}$
while in the latter case $F_{\mathrm{sp}}^{Q}(t)=0\leq F_{\mathrm{ctr}}^{Q}(t)$. 
However, since the quantum coherence is destroyed by the randomness of the Haar gates in the control protocol, their ensemble average in the
limit of $N\to\infty$ displays a universal  ordering, i.e., $\langle F_{\mathrm{ctr}}^{Q}(t)\rangle < \langle F_{\mathrm{sp}}^{Q}(t)\rangle$. 

\subsection{Asymptotic scaling of QFI in RQC\label{sec:ent-RQC}}
In Ref.~\citep{shi2024universal,chu2023strongquantum}, it was shown
that when the initial state is separable or has low entanglement, the QFI in locally interacting systems scales linearly with the system size $L$ at short times, due to the Lieb-Robinson bound~\cite{lieb1972thefinite,anthonychen2023speedlimits}. While the scaling with respect to time is analytically hard to calculate for many-body Hamiltonian dynamics, it would be natural to consider the average scaling of the QFI for chaotic dynamics, generated from the random quantum circuit (RQC) model shown in Figs.~\ref{fig:RQC_notation}(a) and (b),
where tools from RQC theory can be leveraged. To mimic local interactions, we consider a model with two layers of two-site unitary gates.

We can extend the definitions of the QFI in Eqs.~(\ref{eq:ctr-FQ}) and (\ref{eq:sp-FQ}) to apply to the RQC case by making the following replacements:
$G_{\mathrm{ctr}}(t)=\sum_{s,\mu}h_{\mathrm{ctr}}^{[\mu]}(s)$
and $G_{\mathrm{sp}}(t)=t \sum_{\mu}h_{\mathrm{sp}}^{[\mu]}(t)$, where $h_{\mathrm{ctr}}^{[\mu]}(s)=(U^{\dagger}W^{\dagger})^{s}h_{0}^{[\mu]}(WU)^{s}$
and $h_{\mathrm{sp}}^{[\mu]}(t)\equiv U^{\dagger t}h_{0}^{[\mu]}U^{t}$. We further define $U$ as two layers of $q^2\times q^2$ nearest-neighbor unitary gates $U=\prod_{k=1}^{L/2}\mathsf{U}^{[2k,\,2k+1]}\prod_{\ell=1}^{L/2}\mathsf{U}^{[2\ell-1,\,2\ell]}$ and the encoding $W$ acts on each individual site as
$W=\prod_{\mu}\mathsf{W}_0^{[\mu]}$. We then find 
\begin{align}
F_{\mathrm{ctr}}^{Q}(t) & =4\sum_{s,\,\tau}\sum_{\mu,\,\nu}\mathrm{Cov}[h_{\mathrm{ctr}}^{[\mu]}(s),\,h_{\mathrm{ctr}}^{[\nu]}(\tau)]_{\rho_{0}}, \label{eq:Fq_rqc_ctr}\\
F_{\mathrm{sp}}^{Q}(t) & =4t^{2}\sum_{\mu,\,\nu}\mathrm{Cov}[h_{\mathrm{sp}}^{[\mu]}(t),\,h_{\mathrm{sp}}^{[\nu]}(t)]_{\rho_{0}}.\label{eq:Fq_rqc_sp}
\end{align}
\begin{figure}[h]
    \centering    \includegraphics[width=0.47\textwidth]{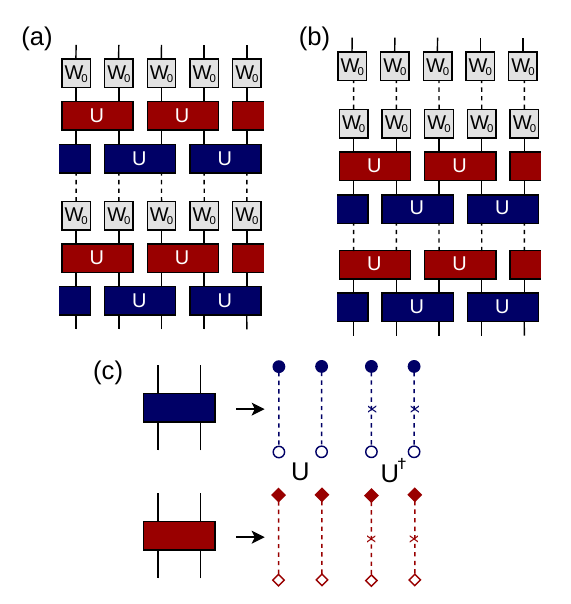}
    \caption{The sensing protocol with RQC and updated notations for two-site CUE random gates. Set of gates applied at each time step with color coded layers of unitary gates for (a) the control protocol and (b) the state preparation protocol. (c) Diagrammatic notation introduced to distinguish between unitary gates acting on neighboring sites in different directions. The first layer of unitary gates are represented by blue dashed lines with circles for indices and the second layer of unitary gates are represented by red dashed lines with rhombi for indices.}
    \label{fig:RQC_notation}
\end{figure}
To understand the behavior of the QFI in a random quantum circuit we need to modify the diagrammatic notation from the RMM case to the RQC case. We want to consider the contributions from each site individually and then employ Principle \ref{prp:many-body} to obtain the final QFI. Therefore, we need to extend our previous diagrammatic description, Fig.~\ref{fig:QFIt1}(a), to distinguish between random unitary gates acting on different neighboring sites. This leads to the updated notation for the unitary gates as seen in Fig.~\ref{fig:RQC_notation}, where the layers of unitary gates are separated by different colors (blue and red) and different index shapes (circles and rhombi). Indices from a blue(red) dotted line can then only be connected to indices from a blue(red) star-dotted line. We also need to ensure the bond constraints are fulfilled, i.e., that two neighboring sites have the same index contractions for their shared unitary gates. However, for the same reason as in the RMM case, leading order contracted diagrams in $q$ have a ``ladder'' structure (see Fig.~\ref{fig:RQC_diagrams}) and this structure needs to be imposed on all sites, ensuring the fulfillment of the bond constraints.

\subsubsection{Initial product state}
As seen in Fig.~\ref{fig:RQC_diagrams}, with the ``ladder'' structure of the diagram the local density operator $\rho^{[\mu]}$ at site $\mu$ appears alone in a T-loop as $\text{Tr}[\rho^{[\mu]}]$ . If the initial state is a product state we will always have $\text{Tr}[\rho^{[\mu]}]=1$, whereas for an initially entangled state it is possible that $\text{Tr}[\rho^{[\mu]}]=0$ for one or more sites. We will therefore start by considering an initial product state, $\rho_0=\otimes_{\mu=1}^L\rho^{[\mu]}$, which ensures that the ``ladder'' structure does not vanish due to the initial state. We then get the following scaling for the two protocols
\begin{align}
    \langle F^Q_{\text{ctr}}(t)\rangle =& \frac{4\text{Tr}[h_0^2]}{q} Lt + \mathscr{O}(1/q), 
    \label{eq:QFI_rqc_ctr} \\
    \langle F^Q_{\text{sp}}(t)\rangle =& \frac{4\text{Tr}[h_0^2]}{q} Lt^2 + \mathscr{O}(1/q),
    \label{eq:QFI_rqc_sp}
\end{align}
where $h_0=h^{[\mu]}_0$. We find a linear scaling with time for the control protocol and a quadratic scaling in time for the state preparation protocol. The scaling with lattice size $L$ is linear for both protocols. 

To show the above equations, we start by considering the control protocol and the QFI as defined in Eq.~\eqref{eq:Fq_rqc_ctr}, and rewrite the expression as 
\begin{widetext}
\begin{align}
    F_{\mathrm{ctr}}^{Q}(t)=\sum_{\mu,\,\nu=1}^L\bigg( \sum_{s,\,\tau=1}^t \underbrace{\text{Tr}[\rho_0 h_{\mathrm{ctr}}^{[\mu]}(s)\,h_{\mathrm{ctr}}^{[\nu]}(\tau)]}_{\text{(I)}} + \sum_{s,\,\tau=1}^t \underbrace{\text{Tr}[\rho_0h_{\mathrm{ctr}}^{[\mu]}(s)]\text{Tr}[\rho_0h_{\mathrm{ctr}}^{[\nu]}(\tau)]}_{\text{(II)}} \bigg).
    \label{eq:QFI_rqc}
\end{align}
We can note the similarity between the expression in the parenthesis and that in Eq.~\eqref{eq:rmt_ctr}, for the RMM case. Utilizing our knowledge from the RMM case, we therefore conclude that if we impose the leading order diagram structure, the term (I) is vanishing if $\tau\neq s$ and the term (II) is vanishing for all $\tau, \, s$. 
Next, we further rewrite (I) for $\tau=s$ to separately consider the two cases $\mu= \nu$ and $\mu\neq \nu$ 
\begin{align}
    \sum_{\mu,\, \nu = 1}^L \text{Tr}[\rho_0h_{\mathrm{ctr}}^{[\mu]}(\tau)^2]
    =& \sum_{\mu=1}^L \underbrace{ \text{Tr}[\rho_0(U^{\dagger}W^{\dagger})^{\tau} (h_{0}^{[\nu]})^2 (WU)^{\tau}]}_{\text{(III)}} + \sum_{\mu,\,\nu=1;\, \mu\neq \nu}^L \underbrace{\text{Tr}[\rho_0 (U^{\dagger}W^{\dagger})^{\tau} h_{0}^{[\mu]}  h_{0}^{[\nu]} (WU)^{\tau}]}_{\text{(IV)}} . \nonumber \\
    \label{eq:QFI_rqc2}
\end{align}
\end{widetext}

\begin{figure}[h]
    \centering    \includegraphics[width=0.48\textwidth]{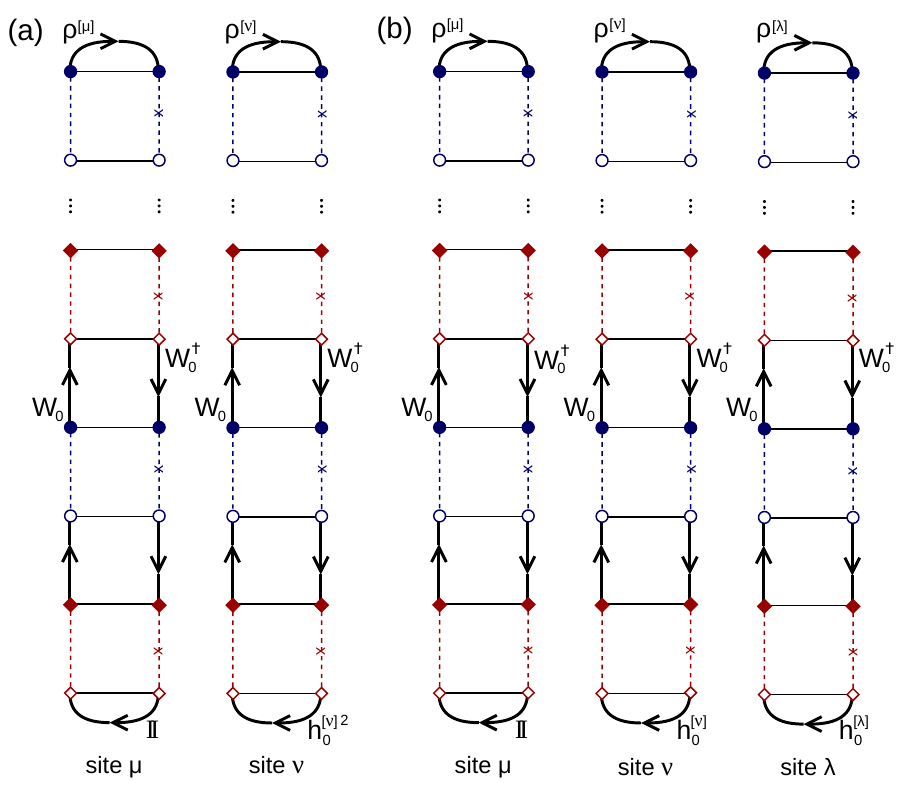}
    \caption{Leading order contracted diagrams in $q$ for the control protocol, as defined in Eq.~\eqref{eq:QFI_rqc2}, for each site. (a) For term (III) $\mu=1, \dots, L; \, \mu\neq \nu$ and (b) for term (IV) $\mu=1, \dots, L; \, \mu\neq \nu,\lambda$. Only the diagrams in (a) are non-vanishing at all sites.}
    \label{fig:RQC_diagrams}
\end{figure}
The leading $N$-order contracted diagrams for terms (III) and (IV) can be seen in Fig.~\ref{fig:RQC_diagrams}(a) and (b) respectively. Consider first the term (IV), for which $\mu\neq \nu$. As seen in the diagram, the single site diagrams at sites $\nu$ and $\lambda$ contain a T-loop of $\text{Tr}[h_0]=0$, which means that their single site contribution is vanishing. Since the full diagrammatic contribution from all sites is a product of the single site contributions, this means that all of these terms are vanishing to the leading order in $q$.
The term (III), on the other hand, is not vanishing; at site $\nu$ the diagrammatic contribution is $q^{-1}\text{Tr}[h_0^2]$ and for $\mu\neq \nu$ the diagrams contribute $1$.
The terms (I) and (IV) do have non-vanishing diagrams with either a lower number of U- or T-loops, which means that these diagrams, where either $\tau\neq s$ or $\mu \neq \nu$ or both, scale at best as $\mathscr{O}(1/q)$. We therefore conclude the following scaling for the two terms in Eq.~\eqref{eq:QFI_rqc}
\begin{align}
\langle\mathrm{Tr}[\rho_{0}h_{\mathrm{ctr}}^{[\mu]}(s)h_{\mathrm{ctr}}^{[\nu]}(\tau)]\rangle & =\frac{\mathrm{Tr}\left[(h_0^{[\mu]})^{2}\right]}{q}\delta_{\tau,\,s}\delta_{\mu,\,\nu}+\mathscr{O}(1/q),\label{eq:RQC-Moment}\\
\langle\mathrm{Tr}[\rho_{0}h_{\mathrm{ctr}}^{[\mu]}(s)]\mathrm{Tr}[\rho_{0}h_{\mathrm{ctr}}^{[\nu]}(\tau)]\rangle & =\mathscr{O}(1/q).\label{eq:RQC-Average}
\end{align}
By including the sums over the time $\tau=1,...,t$ and the sites $\mu=1,...,L$, we finally obtain the expression in Eq.~\eqref{eq:QFI_rqc_ctr}.



For the state preparation protocol, we can rewrite the sum defined in Eq.~\eqref{eq:Fq_rqc_sp} as
\begin{widetext}
\begin{align}
    \sum_{\mu,\,\nu=1}^L\mathrm{Cov}[h^{[\mu]}_{\mathrm{sp}}(t),\,h^{[\nu]}_{\mathrm{sp}}(t)]_{\rho_{0}}
    = \sum_{\mu,\,\nu=1}^L\ \text{Tr}[\rho_0 h^{[\mu]}_{\mathrm{sp}}(t)\,h^{[\nu]}_{\mathrm{sp}}(t)] + \sum_{\mu=1}^L \underbrace{\text{Tr}[\rho_0h^{[\mu]}_{\mathrm{sp}}(t)^2]}_{\text{(V)}} +\sum_{\mu,\,\nu=1}^t \text{Tr}[\rho_0h^{[\mu]}_{\mathrm{sp}}(t)]\text{Tr}[\rho_0h^{[\nu]}_{\mathrm{sp}}(t)].
\end{align}
\end{widetext}
Drawing on our knowledge from the above control protocol and state preparation protocol for the RMM case, we can conclude that the terms (V) are the only ones that are non-vanishing to the leading order in $q$. We therefore conclude the following scaling for the two terms in the covariance
\begin{align}
\langle\mathrm{Tr}[\rho_{0}h^{[\mu]}_{\mathrm{sp}}(t)h^{[\nu]}_{\mathrm{sp}}(t)]\rangle & =\frac{\mathrm{Tr}(h_0^{2})}{q}\delta_{\mu,\,\nu}+\mathscr{O}(1/q),\label{eq:RQC-Moment}\\
\langle\mathrm{Tr}[\rho_{0}h^{[\mu]}_{\mathrm{sp}}(t)]\mathrm{Tr}[\rho_{0}h^{[\nu]}_{\mathrm{sp}}(t)]\rangle & =\mathscr{O}(1/q),\label{eq:RQC-Average}
\end{align}
from which we can readily obtain the final scaling in Eq.~\eqref{eq:QFI_rqc_sp}.

\subsubsection{Initial entangled state}
{We now extend the above results derived for an initial product state to that of an arbitrary entangled initial state. Let {\{$| e_{i_{\mu}}\rangle$\}} be a complete set of orthogonal eigenvectors of $h_0^{[\mu]}$ at site $\mu$ ($i = 1,\dots, q$) and rewrite the density matrix of the whole system in terms of the local basis vectors as $\rho_0 = \sum_{\{i_\mu,\,j_\mu\}}^q p_{i_1\dots i_L j_1\dots j_L} \bigotimes_{\mu}^L |e_{i_\mu}\rangle\langle e_{j_\mu}|$. Consider the diagrammatic method for finding the QFI; for one or more sites $\mu$ the local density operator may now be $\rho^{[\mu]}=|e_{i_\mu}\rangle \langle e_{j_\mu}|$ with $i_\mu\neq j_\mu$, where the true contribution can be weighted by the factor $p_{i_1\dots i_L j_1\dots j_L}$ in the end. For this site the ``ladder'' diagram structure is vanishing, as it contains the T-loop $\text{Tr}[\rho^{[\mu]}] = \text{Tr}[|e_{i_\mu}\rangle \langle e_{j_\mu}|]=0$. 
\begin{figure}[h]
    \centering    \includegraphics[width=0.47\textwidth]{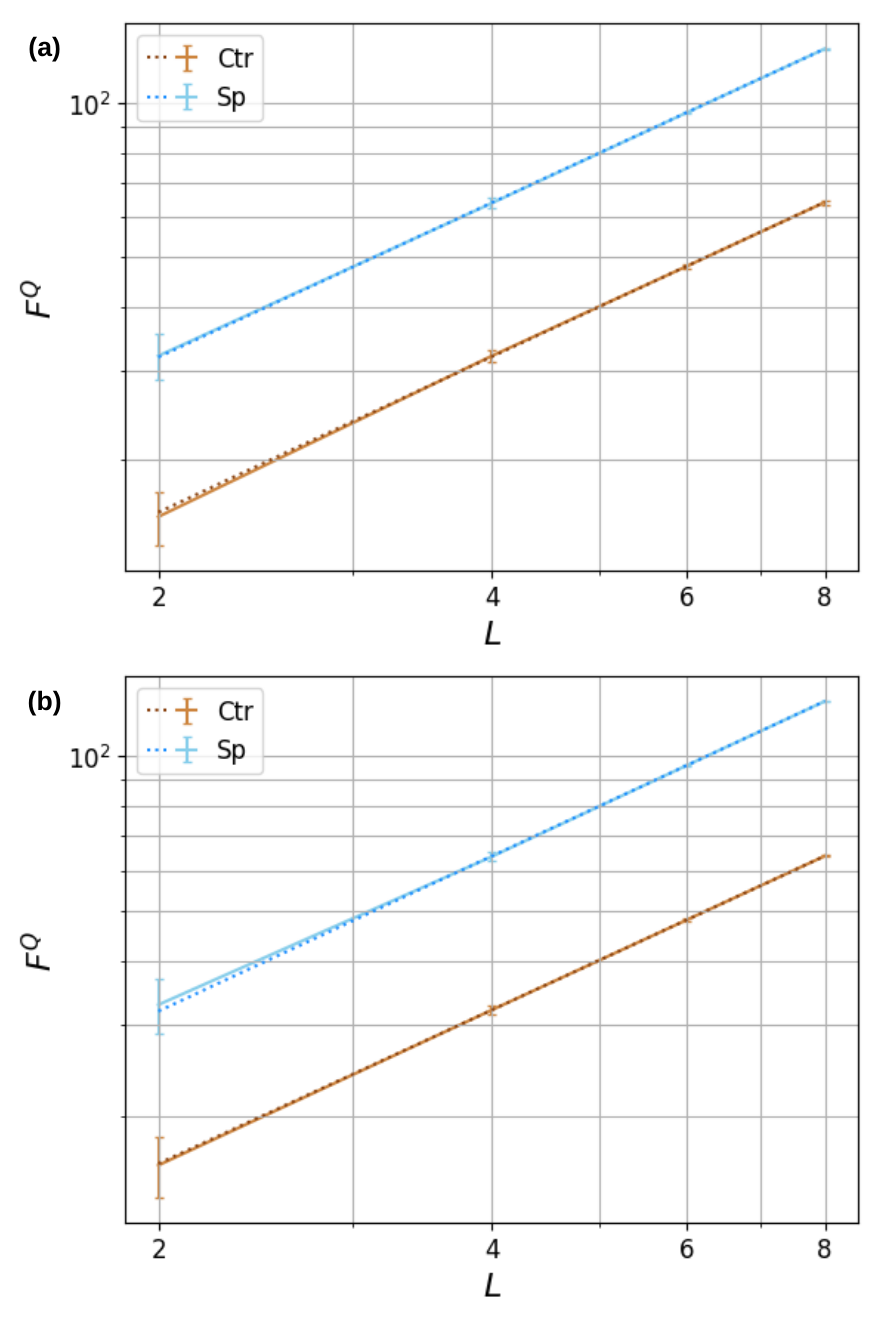}
    \caption{Numerical results for the QFI as a function of the system size $L$ for (a) and initial product state and (b) and initial entangled state. Solid lines are simulations and dotted lines are results from linear fit to the data for both the control (Ctr) and state preparation (Sp) protocol. The scaling behavior is the same regardless of choice of initial state and protocol. Parameters: encoding parameter $\theta=1$, local dimension $q=8$ and time $t=2$. The average was done over $100$ simulation runs.}
    \label{fig:RQC_graph_L}
\end{figure}
However, recall that $[h_0^{[\mu]}, W_0^{[\mu]}]=0$, $\mathsf{W}_0^{[\mu]\dagger} \mathsf{W}^{[\mu]}_0 = \mathbb{I}$ and that for every operator $W_0$ in a T-loop there will be equally many operators $\mathsf{W}_0^{[\mu]\dagger}$, due to the symmetry of their appearance in the diagrams. Then, as we have expresses the density operator in the basis of the local Hamiltonian, we can conclude that as long as $i_\mu\neq j_\mu$, a T-loop involving $\rho^{[\mu]} = |e_{i_\mu}\rangle \langle e_{j_\mu}|$ must vanish because $\rho^{[\mu]}$ commutes with $h_0^{[\mu]}$ and $\mathsf{W}_0^{\mathrm{[\mu]}}$. As such, the off-diagonal contributions from an initially entangled state do not contribute to the final scaling of the QFI. The final scaling is therefore the same as for the initial product state for both the control and state preparation protocols, Eqs.~\eqref{eq:QFI_rqc_ctr} and \eqref{eq:QFI_rqc_sp} respectively. }

\subsubsection{Numerical results}
Numerical simulations of the QFI were done with local Hamiltonian
\begin{align}
    h_0 = 
    \begin{pmatrix}
        \mathbb{I}_{q/2} & 0 \\
        0 & -\mathbb{I}_{q/2}
    \end{pmatrix},
\end{align}
where $q$ is the local dimension of a site. 
\begin{figure}[h]
    \centering    \includegraphics[width=0.47\textwidth]{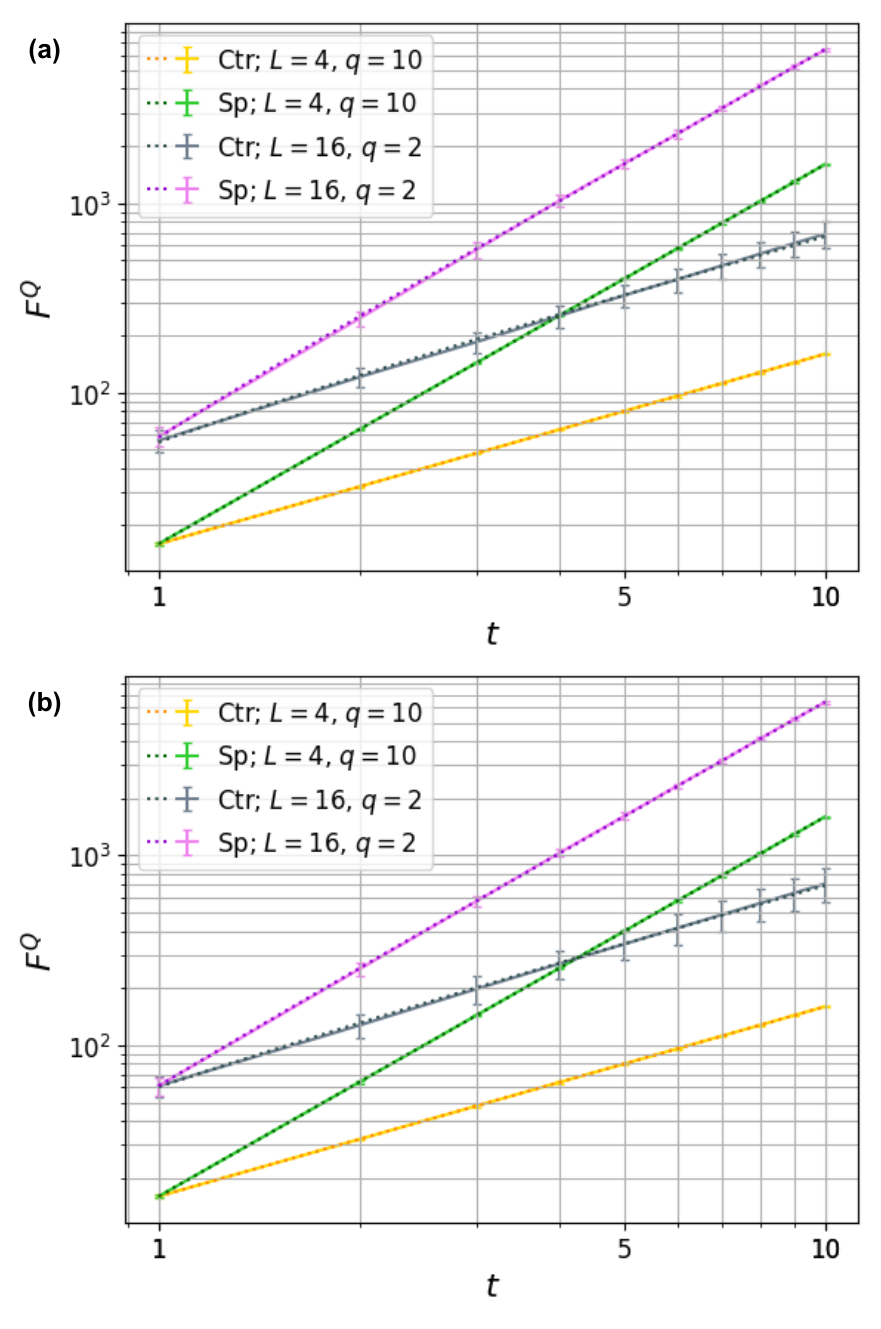}
    \caption{Numerical results for the QFI as a function of time for (a) an initial product state and (b) an initial entangled GHZ-state.
    Solid lines are simulations and dotted lines are results from fitting a linear and quadratic curve to the control protocol (Ctr) and state preparation (Sp) protocol respectively. The graphs show that the scaling behavior is the same regardless of the initial state, but differ between the two protocols. The scaling is the same as long as the total system size $N=q^L$ is large. Parameters: encoding parameter $\theta=1$, the average was done over $100$ simulation runs.
    }
    \label{fig:RQC_graph_t}
\end{figure}
Once again, the numerical results are in perfect agreement with the analytical results. as seen in Figs.~\ref{fig:RQC_graph_L} and ~\ref{fig:RQC_graph_t}: The scaling of the QFI is linearly with respect to. system size for both protocols while it scales linearly with respect to. time for the control protocol and quadratic with respect to time for the state-preparation protocol, regardless of initial state. 

We conclude by discussing the connection between the RMM and RQC cases in the asymptotic regimes where $N\to\infty$ and $q\to\infty$. Let us focus on the control protocol for both cases. As we argue in Sec.~\ref{sec:globalCUE},  in the limit $q\to\infty$, global RMM unitary gate can predict the behavior of the local interactions to the leading order. As such, we can intuitively treat the effect of local CUE gates in the RQC case as a global CUE gate,  the encoding Hamiltonian is given Eq.~\ref{eq:H0-h0-connection} and apply Eq.~\eqref{eq:QFI_rmt_ctr}. Note that $N=q^L$,  $\text{Tr}[h_0^{[\mu]}]=0$, and $\text{Tr}[H_0^2]=Lq^{L-1}\text{Tr}[h_0^2]$, where $h_0\equiv h_0^{[\mu]}$. With these observations, it can be readily observed that
\begin{equation}
    \langle F^Q_\text{ctr}(t)\rangle_\text{RMM} = 4t\frac{\text{Tr}[H_0^2]}{N} = 4tL\frac{\text{Tr}[h_0^2]}{q} = \langle F^Q_\text{ctr}(t)\rangle_\text{RQC},\label{eq:RMM-RQC-Equiv}
\end{equation}
which confirms the RMM description of RQC in the limit $q\to\infty$. The same calculation can be done for the state preparation protocol, with the same conclusion.

\section{Bounding the fluctuation of the QFI using concentration inequalities\label{sec:QFI-fluctuations}}
Having discussed the scaling of the ensemble average of the QFI, one might wonder what the fluctuation of the QFI in a typical realization of the control and state-preparation protocols is. In this section, we employ the concentration inequality for the special unitary group to bound the fluctuation of the QFI. As we shall see in what follows, in the limit $N\to\infty$ for the RMM case and $q\to\infty$ for the RQC case, the fluctuation of the QFI is exponentially suppressed. The QFI for any realization drawn from the random ensemble converges to the average one.

We note that $\mathrm{U}(N)\cong \mathrm{U}(1)\times \mathrm{SU}(N)/\mathbb{Z}_N$, where $\mathbb{Z}_N={0,1,\cdots, N}$ is the Abelian group with the usual addition as the group multiplication. Therefore, the probability measure $d\mu(U)=d\mu(V) d\phi/2\pi$, where $U\in \mathrm{U}(N)$, $V\in \mathrm{SU}(N)$ and $\phi\in [0,2\pi]$. Clearly one can see that as long as $f(U)$ is gauge-invariant, the concentration inequality for the average over the $\mathrm{U}(N)$ and $\mathrm{SU}(N)$ groups are the same. Therefore, we shall focus on the $\mathrm{SU}(N)$ group in what follows. 

Before we discuss the concentration inequality, let us review the Riemann geometry of the $\mathrm{SU}(N)$ manifold. Consider
a point $U\in \mathrm{SU}(N)$, then the tangent vector at $U$ is defined as follows
\begin{equation}
\mathcal{T}_{U}\mathrm{SU}(N)\equiv\left\{ \left.\frac{dx(\epsilon)}{d\epsilon}\right|_{\epsilon=0}, \begin{matrix}x(\epsilon)\text{is a smooth curve }\\
\text{on }\mathrm{SU}(N),\;x(0)=U
\end{matrix}\right\} .
\end{equation}
Note that $x(\epsilon)$ can be parameterized as follows: $x(\epsilon)=e^{-\mathrm{i}X\epsilon}U$
with $X=X^{\dagger}$ and $\mathrm{Tr}X=0$. Any vector in $\mathcal{T}_{U}\mathrm{SU}(N)$
can be identified with $\mathcal{X}_{U}=-\mathrm{i}XU$. Given two
tangent vectors $\mathcal{X}_{U}$ and $\mathcal{Y}_{U}$ in $\mathcal{T}_{U}\mathrm{SU}(N)$,
upon viewing $\mathrm{SU}(N)$ as a the submanifold of $\mathbb{R}^{2N^{2}}$
the natural inner product on $\mathcal{T}_{U}\mathrm{SU}(N)$ inherited from
the Euclidean inner product is 
\begin{equation}
\langle\mathcal{X}_{U},\,\mathcal{Y}_{U}\rangle\equiv\mathrm{Re}\mathrm{Tr}(\mathcal{X}_{U}^{\dagger}\mathcal{Y}_{U})=\mathrm{Tr}(XY).\label{eq:IP-def}
\end{equation}
 The gradient of any real-valued smooth function $f[U]:\mathrm{SU}(N)\rightarrow\mathbb{R}$
is defined as 
\begin{equation}
\langle\nabla_{U}f,\,\mathcal{X}_{U}\rangle\equiv\frac{df[e^{-\mathrm{i}X\epsilon}U]}{ d\epsilon}\bigg|_{\epsilon=0},\,\forall\mathcal{X}_{U}=-\mathrm{i}XU.\label{eq:nabla-f-def-RMT}
\end{equation}
Clearly the above definition is reminiscent of the gradient in multi-variable
calculus in the following sense 
\begin{equation}
\nabla_{\mathcal{Y}}f\cdot\mathcal{X}=\frac{df[\mathcal{Y}+\epsilon \mathcal{X}]}{d\epsilon}\Bigg|_{\epsilon=0}.
\end{equation}
Operationally, when $f[U]$ takes the form of a product of the traces
of matrices that are polynomial of $U$, one can easily identify $\nabla_{U}f$
from the r.h.s.~of Eq.~\eqref{eq:nabla-f-def-RMT}. For example,
for $f[U]=\mathrm{Tr}(UAU^{\dagger}B)$ where $A$ and $B$ are Hermitian
matrices, using Eq.~\eqref{eq:nabla-f-def-RMT}, we obtain 
\begin{align}
\langle\nabla_{U}f,\,\mathcal{X}_{U}\rangle & =\frac{d\mathrm{Tr}(e^{-\mathrm{i}X\epsilon}UAU^{\dagger}e^{\mathrm{i}X\epsilon}B)}{d\epsilon}\bigg|_{t=0}\nonumber \\
 & =\mathrm{Tr}\left([UAU^{\dagger},B](-\mathrm{i}X)\right).\label{eq:nabla-f-Example}
\end{align}
Then immediately one can recognize that 
\begin{equation}
\nabla_{U}f=[B,UAU^{\dagger}]U.
\end{equation}

On the product manifold $\underbrace{\mathrm{SU}(q)\times \mathrm{SU}(q)\times\cdots\times \mathrm{SU}(q)}_{\text{number of }\mathrm{SU}(q)=L}$, 
we note that the corresponding tangent space has the structure 
\begin{equation}
\mathcal{T}_{\{\mathsf{U}_{[\alpha]}\}}\mathrm{SU}(q)\times \mathrm{SU}(q)\times\cdots\times \mathrm{SU}(q)=\oplus_{\alpha}\mathcal{T}_{U^{[\alpha]}}\mathrm{SU}(q).
\end{equation}
This implies that any tangent vector $\mathcal{X}_{\{\mathsf{U}^{[\alpha]}\}}$
on $\mathcal{T}_{\{\mathsf{U}^{[\alpha]}\}}\mathrm{SU}(q)\times \mathrm{SU}(q)\times\cdots\times \mathrm{SU}(q)$
can be identified with $\{\mathcal{X}_{\mathsf{U}^{[\alpha]}}=-\mathrm{i}X_{\alpha}\mathsf{U}^{[\alpha]}\}=\{-\mathrm{i}X_{1}\mathsf{U}^{[1]},\,\cdots,\,-\mathrm{i}X_{1}\mathsf{U}^{[L]}\}$.
Therefore, given two tangent vectors $\mathcal{X}_{\{\mathsf{U}^{[\alpha]}\}}$
and $\mathcal{Y}_{\{\mathsf{U}^{[\alpha]}\}}$, there is a natural
inner product inherited from $\mathrm{SU}(q)$ defined as follows:
\begin{equation}
\langle\mathcal{X}_{\{\mathsf{U}^{[\alpha]}\}},\,\mathcal{Y}_{\{\mathsf{U}^{[\alpha]}\}}\rangle\equiv\sum_{\alpha=1}^{L}\langle\mathcal{X}_{\mathsf{U}^{[\alpha]}},\,\mathcal{Y}_{\mathsf{U}^{[\alpha]}}\rangle=\sum_{\alpha=1}^{L}\mathrm{Tr}(X_{\alpha}Y_{\alpha}).
\end{equation}
Similarly, the gradient of any real-valued smooth function on $\mathrm{SU}(q)\times \mathrm{SU}(q)\times\cdots\times \mathrm{SU}(q)$
is defined as 
\begin{equation}
\langle\nabla_{\{\mathsf{U}^{[\alpha]}\}}f,\,\mathcal{X}_{\{\mathsf{U}^{[\alpha]}\}}\rangle=\frac{df[\{e^{-\mathrm{i}X_{\alpha}\epsilon}\mathsf{U}^{[\alpha]}\}]}{d\epsilon}\bigg|_{\epsilon=0}.
\end{equation}
Furthermore, when $f[\{\mathsf{U}^{[\alpha]}\}]$ takes the form of
a product of traces of matrices that are a polynomials of $\{\mathsf{U}^{[\alpha]}\}$,
one can identify $\nabla_{\{\mathsf{U}^{[\alpha]}\}}f$ operationally
from its definition, as shown in Eq.~\eqref{eq:nabla-f-Example}.

Now we are in a position to discuss the concentration equality for
smooth real-valued functions on the manifold $\mathrm{SU}(N)$ and $\underbrace{\mathrm{SU}(q)\times \mathrm{SU}(q)\times\cdots\times \mathrm{SU}(q)}_{\text{number of }\mathrm{SU}(q)=L}$
associated with the Haar measure and the product Haar measure. Consider
two point $U$ and $V$ on $\mathrm{SU}(N)$ and the geodesic $\gamma(s)$
that connects them satisfying $\gamma(0)=U$ and $\gamma(1)=V$. From
the fundamental theorem of calculus, we know that
\begin{equation}
f[U]-f[V]=\int_{0}^{1}ds\frac{df\left(\gamma(s)\right)}{ds}=\int_{0}^{1}ds\langle\nabla_{\gamma(s)}f,\,\Gamma_{\gamma(s)}\rangle,
\end{equation}
where $\Gamma_{\gamma(s)}$ is the unit tangent vector at $\gamma(s)$
along the geodesic and we have used the definition in Eq.~\eqref{eq:nabla-f-def-RMT}.
Therefore, we obtain
\begin{align}
|f[U]-f[V]| & \leq\int_{0}^{1}ds\big|\langle\nabla_{\gamma(s)}f,\,\Gamma_{\gamma(s)}\rangle\big|\nonumber \\
 & \leq\int_{0}^{1}ds\|\nabla_{\gamma(s)}f\|_{\mathrm{HS}}\|\Gamma_{\gamma(s)}\|_{\mathrm{HS}}\nonumber \\
 & \leq\mathcal{L}d[U,V],
\end{align}
where we have used the Cauchy-Schwarz inequality to obtain the second
last inequality, $\|\bullet\|_{\mathrm{HS}}$ is the Hilbert-Schmidt
norm (also known as the Frobenius norm or Schatten $2-$norm) and
\begin{equation}
\mathcal{L}\equiv\max_{U}\|\nabla_{U}f\|_{\mathrm{HS}},\label{eq:L-RMT}
\end{equation}
is the called the Lipschitz constant and $d[U,V]=\int_{0}^{1}ds\|\Gamma_{\gamma(s)}\|_{\mathrm{HS}}$
is the geodesic distance between $U$ and $V$. Clearly, since the
manifold $\mathrm{SU}(N)$ is compact for any smooth function, the maximization
in Eq.~\eqref{eq:L-RMT} always exists. In general, functions
defined on a manifold are Lipschitz if $|f[U]-f[V]|\leq\mathcal{L}d[U,V]$. 

On the other hand the Haar measure on $\mathrm{SU}(N)$ satisfies the logarithmic
Sobolev inequality, which leads to the concentration inequality for
a Lipschitz function $f[U]$~\cite{anderson2010anintroduction,vershynin2018highdimensional}:
\begin{equation}
\mathrm{\mu}\left(\big|f[U]-\mathbb{E}f[U]\big|\ge\delta\right)\leq2\exp\left(-\frac{N\delta^{2}}{4\mathcal{L}^{2}}\right),\label{eq:Con-Ineq-RMT}
\end{equation}
where the Lipschitz constant is defined as in Eq.~\eqref{eq:L-RMT}.
Generally, a smooth function $f[\{\mathsf{U}^{[\alpha]}\}_{\alpha=1}^{L}]$
on the manifold $\mathrm{SU}(q)\times \mathrm{SU}(q)\times\cdots\times \mathrm{SU}(q)$ also
satisfies the following concentration inequality~\cite{anderson2010anintroduction}:
\begin{equation}
\mathrm{\mu}\left(\big|f[\{\mathsf{U}^{[\alpha]}\}]-\mathbb{E}f[\{\mathsf{U}^{[\alpha]}\}]\big|\ge\delta\right)\leq2\exp\left(-\frac{q\delta^{2}}{4\mathcal{L}^{2}}\right),\label{eq:Con-Ineq-RQC}
\end{equation}
where 
\begin{equation}
\mathcal{L}\equiv\max_{\{\mathsf{U}^{[\alpha]}\}}\|\nabla_{\{\mathsf{U}^{[\alpha]}\}}f\|_{\mathrm{HS}},
\end{equation}
and we have used the mathematical fact that if a measure on a manifold
$M$ satisfies the logarithmic Sobolev inequality, the product measure
on the product manifold $M\times M\times\cdots\times M$ also satisfies
the logarithmic Sobolev inequality~\cite{anderson2010anintroduction}.

Before we apply the concentration inequality to the RMM and RQC
cases, let us note an important observation:

\begin{obs}\label{obs:bi-invariance}Given a smooth, real-valued function
$f[U]$ on $\mathrm{SU}(N)$ and a unitary matrix $W$ in $\mathrm{SU}(N)$, the gradient
of $\tilde{f}[U]=f[UW]$ is related to the gradient of $f[U]$ via
\begin{equation}
\nabla_{U}\tilde{f}=(\nabla_{UW}f)W^{\dagger}\label{eq:grad-relation},
\end{equation}
 and hence their Lipschitz constant are the same. 

\end{obs}
\begin{proof}
We prove the above by exploiting the inner product structure~\eqref{eq:IP-def}.
By definition 
\begin{equation}
\langle\nabla_{U}\tilde{f},\,\mathcal{X}_{U}\rangle\equiv\frac{df[e^{-\mathrm{i}X\epsilon}UW]}{d\epsilon}\bigg|_{\epsilon=0}=\langle\nabla_{UW}f,\,\mathcal{X}_{UW}\rangle,
\end{equation}
where $\mathcal{X}_{U}=-\mathrm{i}XU$ and $\mathcal{X}_{UW}=-\mathrm{i}XUW$.
Therefore, we conclude that 
\begin{align}
\langle\nabla_{UW}f,\,\mathcal{X}_{UW}\rangle & =\mathrm{Tr}\left((\nabla_{U}f)^{\dagger}(-\mathrm{i}XUW)\right)\nonumber \\
 & =\mathrm{Tr}\left((\nabla_{U}fW^{\dagger})^{\dagger}(-\mathrm{i}XU)\right),
\end{align}
which confirms Eq.~\eqref{eq:grad-relation}. It then follows
naturally that 
\begin{equation}
\max_{U}\|\nabla_{U}\tilde{f}\|_{\mathrm{HS}}=\max_{U}\|\nabla_{UW}f\|_{\mathrm{HS}}=\max_{U}\|\nabla_{U}f\|_{\mathrm{HS}},
\end{equation}
which implies that their Lipschitz constants are the same.
\end{proof}
A few comments are in order. First, the same observations hold for the
real-valued smooth functions $\tilde{f}(\{U^{[\alpha]}\})\equiv f(\{U^{[\alpha]}W^{[\alpha]}\})$
and $f(\{U^{[\alpha]}\})$, where $W^{[\alpha]}$ is a given unitary
matrix in $\mathrm{SU}(N)$. Thanks to these observations, one can simplify the
analysis of the Lipschitz constants, as we will see shortly. Finally,
from Eqs.~\eqref{eq:Con-Ineq-RMT} and~\eqref{eq:Con-Ineq-RQC},
it is clear that a large Lipschitz constant corresponds to large fluctuations.
As such, in practice, in order to bound the fluctuations of $f(U)$,
it is sufficient and technically feasible to compute the upper bound
of the Lipschitz constant. To this end, let us note that
\begin{align}
\|\nabla_{U}f\|_{\mathrm{HS}} & \leq\max_{\|\mathcal{X}_{U}\|_{\mathrm{HS}}=1}|\langle\nabla_{U}f,\,\mathcal{X}_{U}\rangle|\nonumber \\
 & =\max_{\|\mathcal{X}_{U}\|_{\mathrm{HS}}=1}\bigg|\frac{df(e^{-\mathrm{i}X\epsilon}U)}{d\epsilon}\big|_{\epsilon=0}\bigg|.
\end{align}
If we find 
\begin{equation}
\bigg|\frac{df(e^{-\mathrm{i}X\epsilon}U)}{d\epsilon}\big|_{\epsilon=0}\bigg|\leq\mathcal{M}\|X\|_{\mathrm{HS}},
\end{equation}
then we immediately know that $\mathcal{L}\leq\mathcal{M}$. A similar
trick also applies to the multi-variable case. 

In Appendix~\ref{sec:upper-bounds}, we employ this
method to show that for the control and state-preparation
protocols in the RMM and RQC cases, 
\begin{equation}
\mathcal{L}_{\mathrm{ctr-RMM}}\leq4t^{2}(t+1)\Delta_{H_{0}},\,\mathcal{L}_{\mathrm{sp-RMM}}\leq8t\Delta_{H_{0}},\label{eq:Lip-Bound-RMT}
\end{equation}
and 
\begin{equation}
\mathcal{L}_{\mathrm{ctr-RMM}}\leq4t^{2}(t+1)L\Delta_{h_{0}},\,\mathcal{L}_{\mathrm{sp-RMM}}\leq8tL\Delta_{h_{0}},\label{Lip-Bound-RQC}
\end{equation}
where $h_0\equiv h_{0}^{[\mu]}$ and the corresponding concentration inequalities are given by Eq.~\eqref{eq:Con-Ineq-RMT} and Eq.~\eqref{eq:Con-Ineq-RQC}, respectively.

It is clear that for any physical sensing Hamiltonian $H_0$ and $h_0$, their spectral widths $\Delta_{H_{0}}$ and $\Delta_{h_0}$
scale as constants as $N\to\infty$ and $q\to\infty$, respectively. Therefore the upper bounds of the corresponding Lipschitz constant remains finite in these limits.
This implies that in the limit $N\to\infty$ and $q\to\infty$ there is no fluctuation of the QFI.

\section{Discussion and Conclusion\label{sec:Dis-Con}}
We have seen from both the numerical and analytical results, the equivalence between the RMM and RQC cases in the asymptotic regime. Let us now discuss the non-asymptotic regime of small $q$ in the RQC case, which is beyond the reach of the asymptotic large $q$ analysis here. Note that for the state-preparation, the scaling of the QFI with respect to time $t$ is always the optimal quadratic scaling, regardless of the values of $q$ and $L$. Therefore, in what follows, we shall focus on the non-asymptotic $t$ scaling for the control protocol and the non-asymptotic $L$ scaling for both protocols.

First, we consider the small $q$ and large $L$ regime. It is natural to extend the asymptotic equivalence between the RMM and RQC cases discussed in Sec.~\ref{sec:globalCUE} and the end of Sec.~\ref{sec:QFI-scaling} to the regime where $q$ is small and $L$ is large such that $N=q^L\gg t$, though we do not have rigorous analytical proof in this regime.  If such an extended equivalence holds, one can extrapolate Eq.~\eqref{eq:RMM-RQC-Equiv} from the regime of large $q$  to the regime of large $q^L$. We then conclude that for the control protocol the $t$-scaling of the QFI is linear both in time $t$ and system size $L$. Indeed, for $q=2$, the numerical calculations in Fig.~\ref{fig:RQC_graph_t} confirm that the linear $t$ scaling for the control protocol while Fig.~\ref{fig:RQC_num_smalltq}(a) confirms the linear $L$ scaling for both protocols.  Note that such results are consistent with the observation that for short times the QFI   scales at most linearly with $L$ due to the Lieb-Robinson bound~\cite{chu2023strongquantum,shi2024universal}.  

\begin{figure}[h]
    \centering    \includegraphics[width=0.48\textwidth]{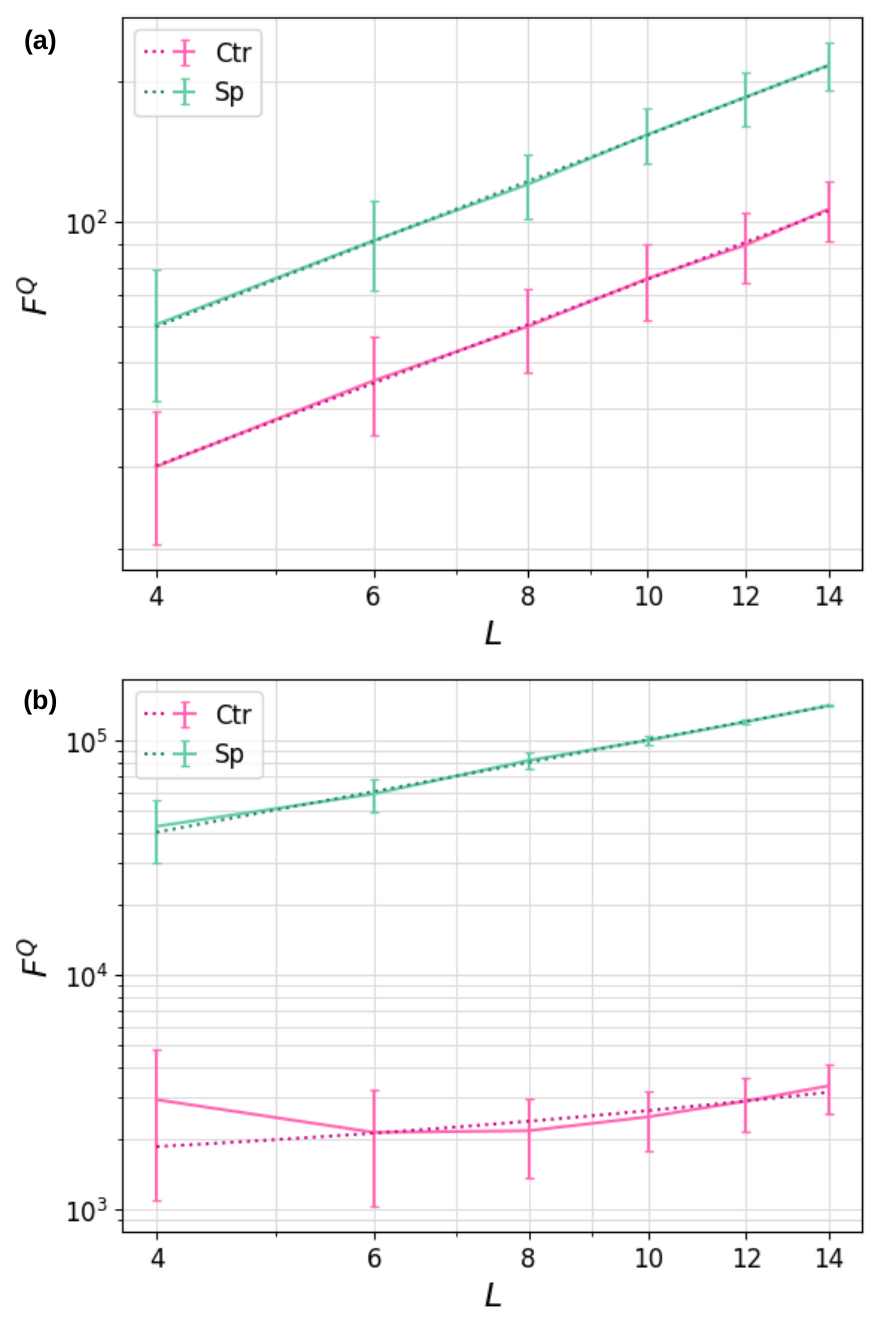}
    \caption{Numerical results for the QFI as a function of system size $L$ for the control (Ctr) and state preparation (Sp) protocols for small $q=2$ at time (a) $t=2$ and (b) $t=50$. The QFI scales linearly for both protocols at small times. At large times the scaling for the state preparation protocol is still linear, whereas the scaling for the control protocol appears to deviate from linearity. Parameters: encoding parameter $\theta=1$. The average was done over $200$ simulation runs for all plots except the control protocol at $t=50$, which was done for $1000$ simulation runs.
    \label{fig:RQC_num_smalltq}}
\end{figure}

Next, we consider the regime small $q$ and small $L$ with {$t\gtrsim q^L$}. We note that numerically fully reaching this regime to investigate the $L$-scaling is very expensive. Nevertheless, in Fig.~\ref{fig:RQC_num_smalltq}(b), We note that this regime is achieved for $t=50$ with $q=2$ and $L=4,6,8$. Due to limited data points, one cannot tell the $L$-scaling conclusively. However, analytical findings~\cite{harrow2009randomquantum,brandao2016localrandom} indicate that when $t\sim L k^{c}$ with constant $c$ that does not scales with $L$, RQC becomes a unitary $k$-design. Therefore, we expect linear scaling for the state-preparation protocol. On the other hand, as shown in Fig.~\ref{fig:RQC_num_smallLq}, in the regime where $t\sim q^L$, we observe nonlinear $t$ scaling for the control protocol, indicating the quantum advantage. The underlying intuition is that non-Gaussian and non-max $T$ contractions contribute to the QFI significantly, and the final scaling is in between the classical standard quantum limit and Heisenberg scaling, as seen from the dotted black lines that envelope the scaling in Fig.~\ref{fig:RQC_num_smallLq}. Quantitative characterization of scaling is left for future study.
\begin{figure}[h]
    \centering    \includegraphics[width=0.5\textwidth]{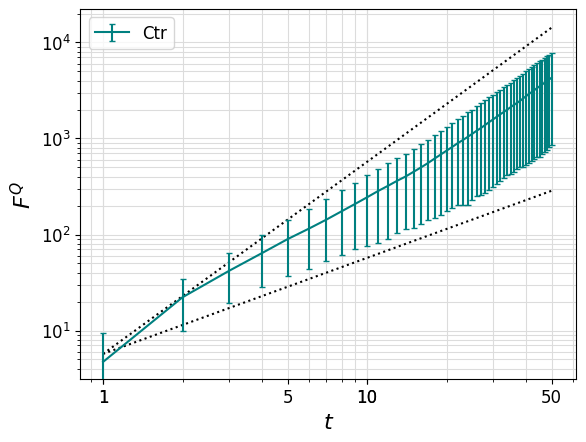}
    \caption{Numerical result for the QFI as a function of time for the control protocol with small $L=4$ and $q=2$. In this regime the QFI scales better than the standard quantum limit. Dotted lines are linear and quadratic scaling respectively. Parameters: encoding parameter $\theta=1$ and the average was done over $100$ simulation runs.}
    \label{fig:RQC_num_smallLq}
\end{figure}

Let us now discuss the fluctuations of the QFI in these non-asymptotic regimes when $q$ is finite. As $L$ or $t$ increases, the upper bounds of the Lipschitz constants Eq.~\eqref{eq:Lip-Bound-RMT} and Eq.~\eqref{Lip-Bound-RQC} becomes loose and cannot bound the fluctuations tightly. Nevertheless, as shown from Fig.~\ref{fig:RQC_num_smalltq}(a) and (b), for a given number of realizations, the error bars of the numerical calculation of the QFI decreases as $L$ increases, which hints that large system size tend to suppress the fluctuations. Meanwhile, as shown in Fig.~\ref{fig:RQC_num_smallLq}, the errors bars tends to increase as $t$ increases, which hints that the long-time sensing tends to increase the fluctuations.

Furthermore, we would like to point out that for the RQC protocol, the restriction of the initial states to symmetric subspaces still lead to the linear $L$ scaling in the asymptotic regime $q\to\infty$. Ref.~\cite{oszmaniec2016randombosonic} studied some "mixed" RMM-RQC case with $t=1$, $L$ qudits and the Haar random gate implemented by a global CUE. Clearly, in this case,  the control and state-preparation protocol coincide. It is found that Haar random gates on the full Hilbert space lead to an almost linear scaling of the average QFI while Haar random gates on the symmetric subspace lead to an average QFI that scales quadratically with the system size $L$, i.e., reaches the Heisenberg limit. The former scaling can be readily seen from Eq.~\eqref{eq:QFI-1time-exact} while the latter scaling is due to the following relation:
\begin{equation}
\mathrm{tr}_{\mathcal{S}_L}(H_0^2) = \frac{L(L+q)\mathrm{tr}(h_0^2)}{q(q+1)}|\mathcal{S}_L|,\label{eq:TrH0-sym}
\end{equation}
where $H_0$ is given by Eq.~\eqref{eq:H0-h0-connection}, $h_0\equiv h_0^{[\mu]}$, $\mathcal{S}_L$ denotes the symmetric subspace and $|\mathcal{S}_L|=\binom{L+q-1}{L}$ denotes its dimension. 

In the asymptotic regime $q\to\infty$, we know $|\mathcal{S}_L|\sim q^L/L!$. For the RQC case in this regime, if we restrict the initial states to the symmetric subspace, the scaling of the QFI is
\begin{align}
    \langle F^Q_{\text{ctr}}(t)\rangle =& \frac{4\text{Tr}_{\mathcal{S}_L}[H_0^2]}{|{\mathcal{S}_L}|} t + \mathscr{O}(1/|{\mathcal{S}_L}|),
    \label{eq:QFI_rqc_sym_ctr} \\
    \langle F^Q_{\text{sp}}(t)\rangle =& \frac{4\text{Tr}_{\mathcal{S}_L}[H_0^2]}{|{\mathcal{S}_L}|} t^2 + \mathscr{O}(1/|{\mathcal{S}_L}|). 
    \label{eq:QFI_rqc_sym_sp}
\end{align}
Upon substituting Eq.~\eqref{eq:TrH0-sym} into the above expression and noting that $q$ is the largest scale, we immediately see the resurgence of the linear scaling since Eq.~\eqref{eq:QFI_rqc_sym_ctr} and Eq.~\eqref{eq:QFI_rqc_sym_sp} reduce to Eq.~\eqref{eq:QFI_rqc_ctr} and Eq.~\eqref{eq:QFI_rqc_sp} in the limit $q\to \infty$, respectively. 

Many open questions and future directions remain unexplored. For example, whether the symmetry and conserved quantities~\cite{li2025sudsymmetric,marvian2022restrictions} lead to different scaling of the QFI; the scaling of QFI when RQC is under continuous or probabilistic measurements~\cite{rossi2020noisyquantum,yang2023efficient,yang2024quantum}; how to design a randomized metrological protocol that escapes the linear scaling, etc.~\cite{xu2025informationally,imai2026referenceframeindependent}.

\section*{Acknowledgement}
All authors contributed equally. Codes used for the current study are available from the authors upon request. We thank Satoya Imai and Yijian Zou for helpful comments.  J.Y. is supported by Zhejiang University Start-up Grants,
Zhejiang Key Laboratory of R\&D and Application of Cutting-edge Scientific
Instruments, and Wallenberg Initiative on Networks and Quantum Information
(WINQ).

\appendix

\section{Bounds on the QFI\label{sec:bounds-QFI}}

The QFI can be expressed as 
\begin{align}
F_{\mathrm{ctr}}^{Q}(t) & =4\mathrm{Tr}[\rho_{0}\sum_{s}\delta H_\mathrm{ctr}(s)]^{2}=4\sum_{s,\tau}\braket{\psi_{0}|\delta H_\mathrm{ctr}(s)\delta H_\mathrm{ctr}(\tau)|\psi_{0}}\nonumber \\
 & \leq4\sum_{s,\tau}\big|\braket{\psi_{0}|\delta H_\mathrm{ctr}(s)\delta H_\mathrm{ctr}(\tau)|\psi_{0}}\big|\nonumber \\
 & \leq4\sum_{s,\tau}\sqrt{\braket{\psi_{0}|[\delta H_\mathrm{ctr}(s)]^{2}|\psi_{0}}}\sqrt{\braket{\psi_{0}|[\delta H_\mathrm{ctr}(\tau)]^{2}|\psi_{0}}}\nonumber \\
 & =4\left(\sum_{s}\sqrt{\mathrm{Var}[H_\mathrm{ctr}(s)]_{\rho_{0}}}\right)^{2},
\end{align}
where $\delta H_\mathrm{ctr}(s)\equiv H_\mathrm{ctr}(s)-\mathrm{Tr}[\rho_{0}H_\mathrm{ctr}(s)]$ and we have
used the Cauchy-Schwarz inequality in the second inequality. Using
the Cauchy-Schwarz inequality $(\sum_{s=1}^t x_{s})^{2}\leq tx_{s}^{2}$
once again we obtain 
\begin{equation}
\left(\sum_{s=1}^t\sqrt{\mathrm{Var}[H_\mathrm{ctr}(s)]_{\rho_{0}}}\right)^{2}\le t\sum_{s=1}^t \mathrm{Var}[H_\mathrm{ctr}(s)]_{\rho_{0}}.
\end{equation}
Finally, we note that 
\begin{equation}
\mathrm{Var}[H_\mathrm{ctr}(s)]_{\rho_{0}}\leq\braket{\psi_{0}|H_\mathrm{ctr}^{2}(s)|\psi_{0}}\leq\|H_\mathrm{ctr}(s)\|_{\infty}^{2}=\|H_{0}\|_{\infty}^{2},
\end{equation}
where $\|\cdot\|$ is the operator norm or the Schatten-$\infty$
norm. In fact, if we shift $H_{0}$ by a constant, we can still obtain
similar bound, i.e., 
\begin{equation}
\mathrm{Var}[H_\mathrm{ctr}(s)]_{\rho_{0}}\leq\|H_{0}-c\|_{\infty}^{2},
\end{equation}
where $c\in\mathbb{R}$. Thus we find 
\begin{equation}
\mathrm{Var}[H_\mathrm{ctr}(s)]_{\rho_{0}}\leq\inf_{c\in\mathbb{R}}\|H_{0}-c\|_{\infty}^{2}=\frac{1}{4}\Delta_{H_{0}}^{2},
\end{equation}
which concludes the proof of Eq.~(\ref{eq:FQ-bound}). The bound
on $F_{\mathrm{sp}}^{Q}(t)$ can be derived in a similar way.

\section{The diagrammatic methods for RMM\label{app:RMT-Diagrammatics}}

\subsection{\label{subsec:diagrammatic-methods-review}Pedagogical review of
the diagrammatic methods for RMM}

Recall that in Sec.~\ref{sec:Weingarten-Pedagogical}, the operators $A_{k}^{(j)}$ and $B_{k}^{(j)}$ implies they are in the $k$-th operator in $j$-th group of trace products. Here,  we introduce the $\xi$, $\eta$-labels. They are functions on the tuple $(jk)$ defined
as
\begin{equation}
\xi(jk)\equiv\sum_{i=0}^{j-1}n_{i}+k,\quad\eta(jk)\equiv\sum_{i=0}^{j-1}m_{i}+k,
\end{equation}
with $k\in[1,\,n_{j}],\,k\in \mathbb{Z}$ in the definition of $\xi(jk)$, $k\in[1,\,m_{j}],\,k\in \mathbb{Z}$ in the definition of $\eta(jk)$, and
$n_{0}=m_{0}=0$. Clearly $\xi(jk),\,\eta(jk)\in[1,t]$ and 
\begin{equation}
A_{k}^{(j)}=A_{\xi(jk)},\,B_{k}^{(j)}=B_{\eta(jk)}.
\end{equation}
We now count the operators $A$ and $B$ on l.h.s.~of Eq.~(\ref{eq:f-rep}), 
in term of $\xi$, $\eta$-labels in the main text and  ones on the r.h.s.~in terms of the $(jk)$ labels. Similarly, one can track
the indices $a$'s, $b$'s, $\alpha$'s and $\beta$'s from the $(jk)$-counting
system to the $\xi$, $\eta$-counting system using the following
map: 
\begin{equation}
a_{k}^{(j)}\to a_{\xi(jk)},\,\alpha_{k}^{(j)}\to\alpha_{\eta(jk)},\;b_{k}^{(j)}\to b_{\xi(jk)},\,\beta_{k}^{(j)}\to\beta_{\eta(jk)}.\label{eq:mapping}
\end{equation}

As manifested by Eq.~(\ref{eq:Wg}) in the main text, in the $\xi$, $\eta$-counting
system it is more straightforward to express the permutations in the
Weingarten formulas, i.e.,
\begin{widetext}
\begin{equation}
\langle\prod_{j=1}^{g}U_{a_{1}^{(j)}b_{1}^{(j)}}\cdots U_{a_{n_{j}}^{(j)}b_{n_{j}}^{(j)}}U_{\alpha_{1}^{(j)}\beta_{1}^{(j)}}^{*}\cdots U_{\alpha_{m_{j}}^{(j)}\beta_{m_{j}}^{(j)}}^{*}\rangle=\sum_{P,\,P^{\prime}}\prod_{j=1}^{g}\left(\prod_{k=1}^{n_{j}}\delta_{a_{\xi(jk)}\alpha_{P\left(\xi(jk)\right)}}\right)\left(\prod_{k=1}^{m_{j}}\delta_{b_{\xi(jk)}\beta_{P^{\prime}\left(\xi(jk)\right)}}\right).\label{eq:Wg-mix}
\end{equation}
To convert the r.h.s.~back to the $(jk)$ counting system, we use
the inverse of the mapping~(\ref{eq:mapping}) and introduce
\begin{equation}
\left(J_{P}(jk)K_{P}(jk)\right)=\eta^{-1}\circ P\circ\xi(jk)\equiv\eta^{-1}\left(P(\xi(jk))\right).
\end{equation}
Then the Weingarten formula becomes
\begin{equation}
\langle\prod_{j=1}^{g}U_{a_{1}^{(j)}b_{1}^{(j)}}\cdots U_{a_{n_{j}}^{(j)}b_{n_{j}}^{(j)}}U_{\alpha_{1}^{(j)}\beta_{1}^{(j)}}^{*}\cdots U_{\alpha_{m_{j}}^{(j)}\beta_{m_{j}}^{(j)}}^{*}\rangle=\sum_{P,\,P^{\prime}}\prod_{j=1}^{g}\left(\prod_{k=1}^{n_{j}}\delta_{a_{k}^{(j)}\alpha_{_{K_{P}(jk)}}^{(J_{P}(jk))}}\right)\left(\prod_{k=1}^{m_{j}}\delta_{b_{k}^{(j)}\beta_{_{K_{P^{\prime}}(jk)}}^{(J_{P^{\prime}}(jk))}}\right),\label{eq:Wg-jk}
\end{equation}
which looks less straightforward than Eq.~(\ref{eq:Wg-mix}).

However, the $(jk)$-counting system has its own advantage in representing
the ensemble average of $\langle f\left(UA_{1},\,\cdots,\,UA_{t},\,U^{\dagger}B_{1},\,\cdots,\,U^{\dagger}B_{t}\right)\rangle$
as we shall see. To this end, let us note that Eq.~(\ref{eq:f-rep})
can be represented by the so-called pre-contraction diagram shown
in Fig.~\ref{fig:Diagram-Review}(a): The row indices of $U$ and
the column indices of $U^{\dagger}$ are represented by empty dots
while the column indices of $U$ and the row indices of $U^{\dagger}$
are represented by filled dots; $U$ is represented by the dashed line
while $U^{\dagger}$ is represented by the dashed line with a star (star-dashed
line); the thick directed solid lines represent the operators
$A_{k}^{(j)}$ and $B_{k}^{(j)}$ with the arrows always pointing
from its row index to its column index. Our next step is to apply
the Weingarten formula~(\ref{eq:Wg-jk}) to evaluate the ensemble
average, which immediately leads to Eq.~(\ref{eq:Wg-ProdTr}) in
the main text, where
\begin{equation}
T_{P,\,P^{\prime}}\equiv\prod_{j=1}^{g}\left[\left(\prod_{k=1}^{n_{j}}\delta_{a_{k}^{(j)}\alpha_{_{K_{P}}(jk)}^{(J_{P}(jk))}}\right)\left(\prod_{k=1}^{m_{j}}\delta_{b_{k}^{(j)}\beta_{_{K_{P^{\prime}}}(jk)}^{(J_{P^{\prime}}(jk))}}\right)[A_{1}^{(j)}]_{b_{1}^{(j)}a_{2}^{(j)}}[A_{2}^{(j)}]_{b_{2}^{(j)}a_{3}^{(j)}}\cdots[A_{n_{j}}^{(j)}]_{b_{n_{j}}^{(j)}\beta_{1}^{(j)}}[B_{1}^{(j)}]_{\alpha_{1}^{(j)}\beta_{2}^{(j)}}[A_{2}^{(j)}]_{\alpha_{2}^{(j)}\beta_{3}^{(j)}}\cdots[B_{m_{j}}^{(j)}]_{\alpha_{m_{j}}^{(j)}a_{1}^{(j)}}\right].\label{eq:T-ProdTr}
\end{equation}
\end{widetext}
Eq.~(\ref{eq:Wg-ProdTr}) is the results of applying the Weingarten
formula to a product of traces, which we shall refer to as the \textit{Weingarten
trace formula}. The symbolic representation of $T_{P,\,P^{\prime}}$,
Eq.~(\ref{eq:T-ProdTr}) looks very abstract and cumbersome and the
labels are counted in the $(jk)$-system. Nevertheless, it bears a
very simple, elegant diagrammatic representation. This is in analogy
with the case where symbolically representing all the $t!$ permutations
of a $t$-tuple is very abstract and cumbersome but there is an intuitive
way to represent them diagrammatically. 

To find the diagrammatic representation of $T_{P,\,P^{\prime}}$ ,
let us come back to the pre-contraction diagram. For a given $P$
and $P^{\prime}$, we note that $\delta_{a_{k}^{(j)}\alpha_{_{K_{P}(jk)}}^{(J_{P}(jk))}}$
is equivalent to an identity matrix $\mathbb{I}$ with indices $a_{k}^{(j)}$
and $\alpha_{\alpha_{_{K_{P}}(jk)}^{(J_{P}(jk))}}$, which is represented
by the empty dots in the pre-contraction diagram. As
such, it can be represented by drawing thin solid lines that connect
the two empty dots, as shown in Fig.~\ref{fig:Diagram-Review}(b).
Similarly, $\delta_{b_{k}^{(j)}\beta_{_{K_{P^{\prime}}}(jk)}^{(J_{P^{\prime}}(jk))}}$
is equivalent to an identity matrix $\mathbb{I}$ with indices $b_{k}^{(j)}$
and $\beta_{_{K_{P^{\prime}}}(jk)}^{(J_{P^{\prime}}(jk))}$, which
is represented by the filled dots in the pre-contraction diagram. Thus
$\delta_{b_{k}^{(j)}\beta_{_{K_{P^{\prime}}}(jk)}^{(J_{P^{\prime}}(jk))}}$
can be represented by a thin solid line that connects the two filled
dots. In fact, $T_{P,\,P^{\prime}}$ has the following structure: 

\begin{obs} $T_{P,\,P^{\prime}}$ must be expressed in terms of a product of
traces of the product of the operators $A$, $B$ and the identity matrix $\mathbb{I}$.
Each trace is represented by a $T$-loop, consisting of alternating
thick and thin lines. Within the same $T$-loop, the number of thick
lines must be equal to the number of thin lines. 

\end{obs}
\begin{proof}
Since all the indices appear only twice and they are dummy indices, $T_{P,\,P^{\prime}}$
must be expressed in terms of a product of trace. We observe that since the indices
of $A$'s and $B$'s do not overlap with each other, the only possible
way for them to form a trace is of the type $\mathrm{Tr}(A\mathbb{I}A\cdots)$,
$\mathrm{Tr}(A\mathbb{I}B\cdots)$, $\mathrm{Tr}(B\mathbb{I}A\cdots)$,
$\mathrm{Tr}(B\mathbb{I}B\cdots)$. This means each $T$-loop must
have alternating thick and thin lines. For exactly the same reason,
the number of identity matrices $\mathbb{I}$ and the number of operators
$A$ or $B$ in a single $T$-loop must be equal. Otherwise, the operators
must have indices in common.
\end{proof}
From the above discussion, it is clear that $T_{P,\,P^{\prime}}$ can
be diagrammatically represented as the product of the $T$-loops.
The diagram representation of $T_{P,\,P^{\prime}}$, as shown in
Fig.~\ref{fig:Diagram-Review}(b) is called the ``fully contracted''
diagram in that all the indices of $U$ and $U^{\dagger}$ are contracted.
We shall distinguish it from the pre-contraction diagrams, as shown
in Fig.~\ref{fig:Diagram-Review}(a). 

On the other hand, the Weingarten function $V_{P^{-1}P^{\prime}}$
is also encoded in the fully contracted diagrams and is fully determined by the length of the disjoint cycles in the permutation $P^{-1}P^{\prime}$. More precisely,
the cycles of $P^{-1}P^{\prime}$ have the following structure:

\begin{obs} The cycles in the permutation $P^{-1}P^{\prime}$ is
given by the $U$-loops, which are defined as the loops consisting of
alternating dashed or star-dashed lines and thin lines. Within
a $U$-loop, the total number of dashed lines and star-dashed
lines must be equal to the number of thin lines; the number of dashed lines must be equal to the number of
star-dashed lines.

\end{obs}
\begin{proof}
It is more straightforward to see the cycle structures in the permutation
of $P^{-1}P^{\prime}$ from the fully contracted diagram if the $\xi$,$\eta$-counting
system is used. We note that in the $\xi$,$\eta$-counting systems
the row index $a$ and column index $b$ belonging to the same $U$
has the label $\xi$. That is, the matrix element of $U$ is always
$U_{a_{\xi}b_{\xi}}$. Similar observations also hold for $U^{\dagger}$.
This means that in both the pre-contraction and in the fully-contracted diagrams, an
empty dot and a filled dot connected through the dashed line or
star-dashed line must have the same label, as shown Fig.~\ref{fig:Diagram-Review}(c). 

Let us now break down each segment in a $U$-loop. (i) Start with an index $b_{\xi}$, the column index of $U$, as shown
in Fig.~\ref{fig:Diagram-Review}(c). (ii) The thin solid line that
connects $b_{\xi}$ to the filled dot on the r.h.s.~(the row index
of $U^{\dagger}$) must be $P^{\prime}(\xi)$. (iii) Passing through
the star-dashed line, the empty dot must have the same label $P^{\prime}(\xi)$
because it is connected to the previous filled dot through the
same $U^{\dagger}$. (iv) Passing again through the thin solid line, the
label of the empty dot on the l.h.s.~is $P^{-1}\circ P^{\prime}(\xi)$.
(v) The $U$-loop continues until it returns to the empty dot with
index $a_{\xi}$ and is then closed by passing through the dashed line connecting
$a_{\xi}$ to $b_{\xi}$. Through above process, we immediately
observe that

\[
P^{-1}\circ\cdots P^{\prime}\circ P^{-1}\circ P^{\prime}(\xi)=\xi,
\]
which says that $(P^{-1}P^{\prime})^{u}(\xi)=\xi$, where $u$ is equal
to the number of thin solid lines or dashed and star-dashed lines in the $U$-loop.
In steps (i)-(v), a $U$-loop starts with a thin solid line but ends
with a dashed line. Since the thin solid lines and the dashed/star-dashed
lines are alternating, their numbers must be equal. 
\end{proof}
$V_{P,\,P^{\prime}}$ is fully determined by the length of the cycles
$u_{1},\,u_{2},\,\cdots,\,u_{s}$ in the permutation, i.e., $V_{P,\,P^{\prime}}=V_{u_{1},\,u_{2},\,\cdots,\,u_{s}}$,
which can be read off from the number of thin solid lines or the number
of dash/star-dashed lines, in the $U$-loop.

To summarize: the diagrammatic representation of a linear functional
$f$ is called a ``pre-contraction'' diagram. When the dots on the
pre-contraction are connected through thin solid lines according to
the diagrammatic representation of $T_{P,\,P^{\prime}}$, it is called a
``fully contracted'' diagram. Diagrams that are of the
$t$-th power in both $U$ and $U^\dagger$
are called diagrams of degree $t$.
\begin{itemize}
\item Upon removing the dashed lines and star-dashed lined representing
$U$ and $U^{\dagger}$ in the fully contracted diagrams, $T_{P,\,P^{\prime}}$
can be obtained as the product of the $T$-loop, 
\item Upon removing the directed solid lines representing the operators
$A$'s and $B$'s in the fully contracted diagrams, $u_{1},\,u_{2},\,\cdots,\,u_{s}$
are half of the number of the thin solid lines in the $U$-loops.
\end{itemize}
\begin{figure}
\centering
\includegraphics[scale=0.7]{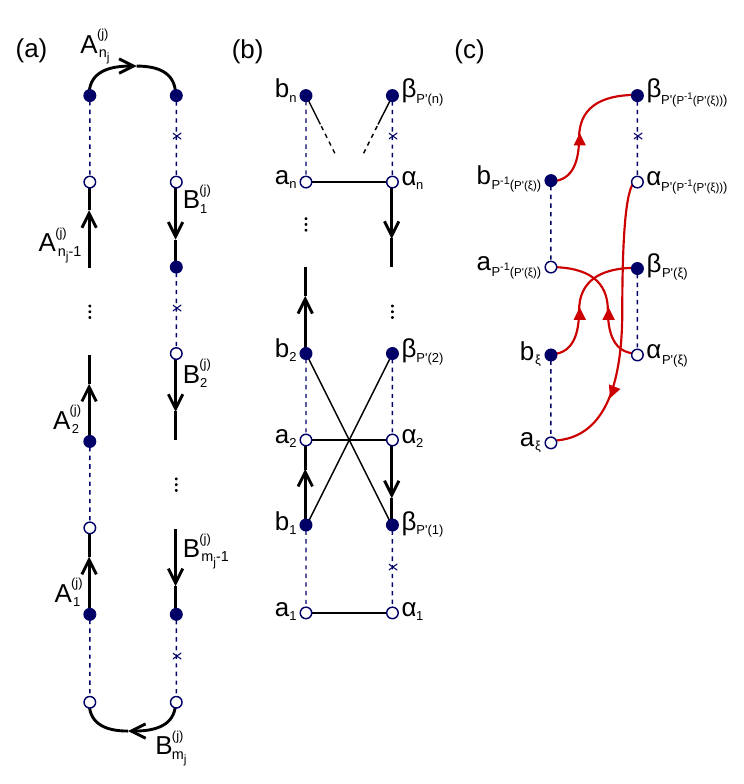}

\caption{\label{fig:Diagram-Review}(a) A generic diagrammatic representation
of $f(UA_{1},\,\cdots,\,UA_{t},\,U^{\dagger}B_{1},\,\cdots,\,U^{\dagger}B_{t})$
can consist of $g$ groups of pre-contraction diagrams, where each part
represents $\mathrm{Tr}(UA_{1}^{(j)}\cdots UA_{n_j}^{(j)}U^{\dagger}B_{1}^{(j)}\cdots U^{\dagger}B_{m_j}^{(j)})$
with $j=1,2,\cdots g$. (b) The schematic representation of a fully-contracted
diagrams with $P=\begin{pmatrix}1 & 2 & \cdots & t\\
1 & 2 & \cdots & t
\end{pmatrix}$ i.e., the identity in the permutation group. As such, the thin solid that connects
the empty dots denote $\delta_{a_{\xi}\alpha_{\xi}}$ with $\xi=1,2,\cdots,t$.
The the thin solid lines that connects the filled dots denotes $\delta_{b_{\xi}\beta_{P^{\prime}(\xi)}}$ with $\xi=1,2,\cdots,t$. (c) An example $U$-loop, which represents a cycle in the permutation $P^{-1}P^{\prime}$: $b_\xi \to \beta_{P^{\prime}(\xi)}\to a_{P^{\prime}(\xi)}\to \cdots \to \alpha_{P^{\prime}(P^{-1}(P^{\prime}(\xi)))}\to a_\xi\to b_\xi$. This loop implies that $P^{-1}(P^{\prime}(P^{-1}(P^{\prime}(\xi))))=(P^{-1}P^{\prime})^2(\xi)=\xi$.}  
\end{figure}

\subsection{Useful observations for extracting the leading-order fully contracted
diagrams for RMM}

\begin{obs} The total number of dashed lines in a pre-contraction diagram
is $t$, the same as the number of star-dash lines; the total number
of thick solid lines is $2t$, the same as the number of thin solid lines
in each fully-contracted diagram.

\end{obs}

\begin{obs}Assuming each $T$-loop scales linearly in $q$ in
a fully-contracted diagram, then the fully-contracted diagram scales
as $q$ to the power of a positive or negative integer and the
fractional power is forbidden. The maximum scaling for a fully-contracted
diagram is linear in $q$, i.e., $\mathscr{O}(q)$.

\end{obs}

\begin{obs}The top(bottom) thick lines in the pre-contraction diagrams
always connect two filled(empty) dots. Therefore, in any fully-contracted
diagram, the top and bottom thick lines cannot form a $T$-loop
without involving the lateral thick lines. They can form a $T$-loop
in only two ways: (a) the top(bottom) thick line forms a
$T$-loop on its own by joining its start and end filled(empty)
dots or (b) the top(bottom) thick line form a $T$-loop together with
the lateral thick lines, which involves \textit{at least two} lateral
thick lines on the opposite sides.

\end{obs}

\begin{obs}The lateral thick lines cannot form a $T$-loop on their
own. The loop must involve the thick lines on opposite lateral sides.
Therefore, except for the $T$-cycles formed by the top or bottom
thick lines, any $T$-loop will contain at least two thick lines
from the two opposite sides. 
\end{obs}

\subsection{Pedagogical illustration of the diagrammatic method for large $t$ \label{app:K-illustration}}
As a demonstration of the effectiveness of the diagrammatic method introduced above, we define the following quantities for both protocols:
\begin{subequations}
\begin{equation}
K_{\mathrm{ctr}}(t) =\langle\mathrm{Tr}[(WU)^{t}]\mathrm{Tr}[(U^{\dagger}W^{\dagger})^{t}],
\end{equation}
\begin{equation}
K_{\mathrm{sp}}(t) =\langle\mathrm{Tr}[W^{t}U^{t}]\mathrm{Tr}[W^{\dagger t}U^{\dagger t}]\rangle.
\end{equation}
\label{eq:K-def}
\end{subequations}
Note that $K_\mathrm{ctr}$ is the spectral form factor (SFF), a measure of spectral statistics~\cite{haake2001quantum}, of the control protocol, i.e., RQC with Floquet operator $WU$, while $K_\mathrm{sp}$ is a generalization of  $K_\mathrm{ctr}$ to the state preparation protocol but can not be interpreted directly as a SFF at time t. Nevertheless, we would like to illustrate how Principle~\ref{prp:RM} in the main text can be applied to extract the leading order scaling of these quantities.
\begin{figure}
\centering
\includegraphics[scale=0.69]{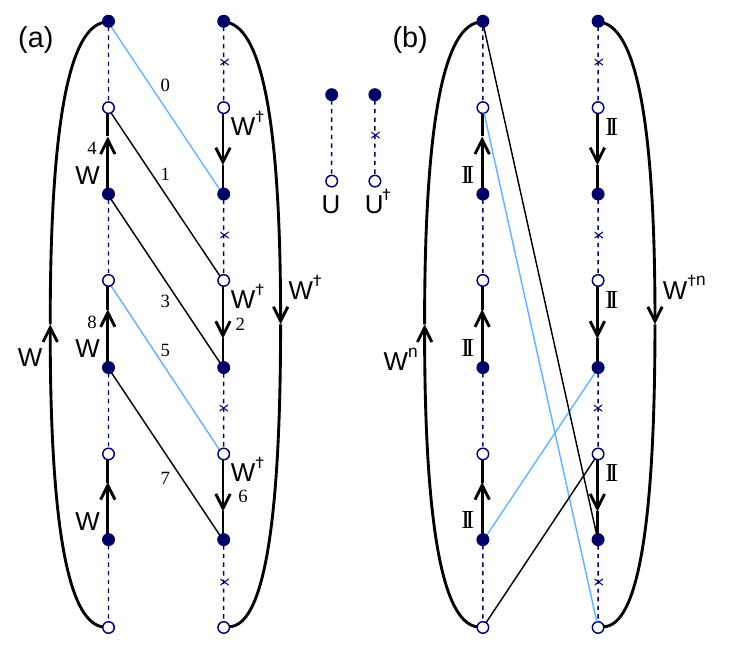}

\caption{\label{fig:SFF}The contraction of the quantity defined in Eq.~\eqref{eq:K-def}  for (a)
the control protocol (b) the state-preparation based protocol.}
\end{figure}
To this end, we consider Fig.~\ref{fig:SFF}(a)
in the limit $N\to\infty$ for the calculation of $K_{\mathrm{ctr}}(t)$:
\begin{itemize}

\item We start with the filled dot on the left top corner. There are $t$ possible ways to contract it with the filled dots on the right. In Fig.~\ref{fig:SFF}(a), we choose the blue line $0$. The Gaussianity principle then immediate leads to the contraction given by the black line $1$. 

\item Then, we consider the max-$T$ principle and conclude that if
each thick line on the left is paired with one thick line on
the right to form a $T$-loop, this is the only way that gives rise to the maximum number
of $T$-loops, which is $t$. This leads to the contraction $1\rightarrow2\rightarrow3\rightarrow4\rightarrow1$
in Fig~\ref{fig:SFF} (a). 

\item Next, we consider the Gaussianity principle again.
This immediately leads to the contraction 5.

\item Finally, we consider the max-$T$-principle again, which leads to
the $T$-loop $5\rightarrow6\rightarrow7\rightarrow8\rightarrow5$.
\end{itemize}
By considering the two principles alternately, one should be able
to contract all the empty and filled dots in Fig~\ref{fig:SFF}
(a), which in the end leads to 
\begin{equation}
K_{\mathrm{ctr}}(t)=t,\label{eq:RMT-Ktr}
\end{equation}
which agrees with the SFF of CUE before the plateau regime~\cite{haake2001quantum}, demonstrating RMM statistics - a defining signature of quantum chaos. For the plateau, a nonperturbative treatment is required~\cite{chan2018solution,liao2022effective}.
The leading order contractions in Fig.~\ref{fig:SFF}(b) can be also found analogously, which gives
\begin{equation}
K_{\mathrm{sp}}(t)=(t-1)\frac{|\mathrm{Tr}W^{t}|^{2}}{N^{2}}+1.\label{eq:RMT-Ksp}
\end{equation}

\section{Details of the many-body Weingarten formula and its diagrammatic methods for
Floquet RQC}
\label{app:RQC-Weingarten}

\subsection{Proof of the many-body Weingarten formula~\eqref{eq:Wg-Manybody} in the main text}
We note that if the gates $\mathsf{U}_{\bm{a}_{1}^{[\mu,\,\mu+1]}\bm{b}_{1}^{[\mu,\,\mu+1]}}$ for different $\mu$'s are independent of each other, the global average of the Floquet operator completely factorizes into the product of local Weingarten averages over the each two-site unitary gate. Thus, we apply the Weingarten formula separately for each two-site unitary in the Floquet
operator and obtain
\begin{widetext}
\begin{align}
 & \langle U_{\{a_{1,\,\sim\mu+\sigma(\mu)}^{[\mu]},\,b_{1,\,\sim\mu-\sigma(\mu)}^{[\mu]}\}}\cdots U_{\{a_{t,\,\sim\mu+\sigma(\mu)}^{[\mu]},\,b_{t,\,\sim\mu-\sigma(\mu)}^{[\mu]}\}}U_{\{\alpha_{1,\,\sim\mu+\sigma(\mu)}^{[\mu]},\,\beta_{1,\,\sim\mu-\sigma(\mu)}^{[\mu]}\}}^{*}\cdots U_{\{\alpha_{t,\,\sim\mu+\sigma(\mu)}^{[\mu]},\,\beta_{t,\,\sim\mu-\sigma(\mu)}^{[\mu]}\}}^{*}\rangle\nonumber \\
= & \prod_{\mu=1}^{L}\langle\mathsf{U}_{\bm{a}_{1}^{[\mu,\,\mu+1]}\bm{b}_{1}^{[\mu,\,\mu+1]}}\cdots\mathsf{U}_{\bm{a}_{t}^{[\mu,\,\mu+1]}\bm{b}_{t}^{[\mu,\,\mu+1]}}\mathsf{U}_{\bm{\alpha}_{1}^{[\mu,\,\mu+1]}\bm{\beta}_{1}^{[\mu,\,\mu+1]}}^{*}\cdots\mathsf{U}_{\bm{\alpha}_{t}^{[\mu,\,\mu+1]}\bm{\beta}_{t}^{[\mu,\,\mu+1]}}^{*}\rangle\nonumber \\
\times & \prod_{\mu=1}^{L}\delta_{a_{1,\,\sim\mu-\sigma(\mu)}^{[\mu]}b_{1,\,\sim\mu+\sigma(\mu)}^{[\mu]}}\cdots\prod_{\mu=1}^{L}\delta_{a_{t,\,\sim\mu-\sigma(\mu)}^{[\mu]}b_{t,\,\sim\mu+\sigma(\mu)}^{[\mu]}}\prod_{\mu=1}^{L}\delta_{\alpha_{1,\,\sim\mu-\sigma(\mu)}^{[\mu]}\beta_{1,\,\sim\mu+\sigma(\mu)}^{[\mu]}}\cdots\prod_{\mu=1}^{L}\delta_{\alpha_{t,\,\sim\mu-\sigma(\mu)}^{[\mu]}\beta_{t,\,\sim\mu+\sigma(\mu)}^{[\mu]}}\nonumber \\
= & \prod_{\mu=1}^{L}\left(\sum_{P^{[\mu,\, \mu+1]},\,P^{\prime[\mu,\,\mu+1]}}V_{P^{-1[\mu,\,\mu+1]}P^{\prime[\mu,\,\mu+1]}}\prod_{\xi=1}^{t}\delta_{a_{\xi,\,\sim\mu+1}^{[\mu]}\alpha_{P^{[\mu,\,\mu+1]}(\xi),\,\sim\mu+1}^{[\mu]}}\delta_{a_{\xi,\,\sim\mu}^{[\mu+1]}\alpha_{P^{[\mu,\,\mu+1]}(\xi),\,\sim\mu}^{[\mu+1]}}\delta_{b_{\xi,\,\sim\mu+1}^{[\mu]}\beta_{P^{\prime[\mu,\,\mu+1]}(\xi),\,\sim\mu+1}^{[\mu]}}\delta_{b_{\xi,\,\sim\mu}^{[\mu+1]}\beta_{P^{\prime[\mu,\,\mu+1]}(\xi),\,\sim\mu}^{[\mu+1]}}\right)\nonumber \\
\times & \prod_{\mu=1}^{L}\delta_{a_{1,\,\sim\mu-\sigma(\mu)}^{[\mu]}b_{1,\,\sim\mu+\sigma(\mu)}^{[\mu]}}\cdots\delta_{a_{t,\,\sim\mu-\sigma(\mu)}^{[\mu]}b_{t,\,\sim\mu+\sigma(\mu)}^{[\mu]}}\delta_{\alpha_{1,\,\sim\mu-\sigma(\mu)}^{[\mu]}\beta_{1,\,\sim\mu+\sigma(\mu)}^{[\mu]}}\cdots\delta_{\alpha_{t,\,\sim\mu-\sigma(\mu)}^{[\mu]}\beta_{t,\,\sim\mu+\sigma(\mu)}^{[\mu]}}.\label{eq:Wg-Manybody-Expand}
\end{align}
Note that the second equation implies that the Weingarten formula has been applied to the two-site unitary gates, regardless of whether their indices are contracted or free.
Using the identity
\begin{equation}
\prod_{\mu=1}^{L}\left\{ \sum_{P^{[\mu,\,\mu+1]},\,P^{\prime[\mu,\,\mu+1]}}F\left(P^{[\mu,\,\mu+1]},\,P^{\prime[\mu,\,\mu+1]}\right)\right\} =\sum_{P^{[12]},\,P^{\prime[12]}}\cdots\sum_{P^{[L1]},\,P^{\prime[L1]}}F\left(P^{[12]},\,P^{\prime[12]}\right)\cdots F\left(P^{[L1]},\,P^{\prime[L1]}\right),\label{eq:Prod-Sum-Id}
\end{equation}
where $F(P^{[\mu\nu]},\,P^{\prime[\mu\nu]})$ denotes some function
of the permutation $P^{[\mu\nu]}$ and $P^{\prime[\mu\nu]}$, the
last equation of Eq.~(\ref{eq:Wg-Manybody-Expand}) can be recast
into 
\begin{align}
 & \langle U_{\{a_{1,\,\sim\mu+\sigma(\mu)}^{[\mu]},\,b_{1,\,\sim\mu-\sigma(\mu)}^{[\mu]}\}}\cdots U_{\{a_{t,\,\sim\mu+\sigma(\mu)}^{[\mu]},\,b_{t,\,\sim\mu-\sigma(\mu)}^{[\mu]}\}}U_{\{\alpha_{1,\,\sim\mu+\sigma(\mu)}^{[\mu]},\,\beta_{1,\,\sim\mu-\sigma(\mu)}^{[\mu]}\}}^{*}\cdots U_{\{\alpha_{t,\,\sim\mu+\sigma(\mu)}^{[\mu]},\,\beta_{t,\,\sim\mu-\sigma(\mu)}^{[\mu]}\}}^{*}\rangle\nonumber \\
= & \sum_{P^{[12]},\,P^{\prime[12]}}\sum_{P^{[23]},\,P^{\prime[23]}}\cdots\sum_{P^{[L1]},\,P^{\prime[L1]}}\prod_{\mu=1}^{L}\left[V_{P^{-1[\mu,\,\mu+1]}P^{\prime[\mu,\,\mu+1]}}\prod_{\xi=1}^{t}\left(\delta_{a_{\xi,\,\sim\mu+1}^{[\mu]}\alpha_{P^{[\mu,\,\mu+1]}(\xi),\,\sim\mu+1}^{[\mu]}}\delta_{a_{\xi,\,\sim\mu}^{[\mu+1]}\alpha_{P^{[\mu,\,\mu+1]}(\xi),\,\sim\mu}^{[\mu+1]}}\delta_{b_{\xi,\,\sim\mu+1}^{[\mu]}\beta_{P^{\prime[\mu,\,\mu+1]}(\xi),\,\sim\mu+1}^{[\mu]}}\delta_{b_{\xi,\,\sim\mu}^{[\mu+1]}\beta_{P^{\prime[\mu,\,\mu+1]}(\xi),\,\sim\mu}^{[\mu+1]}}\right)\right]\nonumber \\
\times & \prod_{\mu=1}^{L}\delta_{a_{1,\,\sim\mu-\sigma(\mu)}^{[\mu]}b_{1,\,\sim\mu+\sigma(\mu)}^{[\mu]}}\cdots\delta_{a_{t,\,\sim\mu-\sigma(\mu)}^{[\mu]}b_{t,\,\sim\mu+\sigma(\mu)}^{[\mu]}}\delta_{\alpha_{1,\,\sim\mu-\sigma(\mu)}^{[\mu]}\beta_{1,\,\sim\mu+\sigma(\mu)}^{[\mu]}}\cdots\delta_{\alpha_{t,\,\sim\mu-\sigma(\mu)}^{[\mu]}\beta_{t,\,\sim\mu-\sigma(\mu)}^{[\mu]}}\nonumber \\
= & \sum_{P^{[12]},\,P^{\prime[12]}}\sum_{P^{[23]},\,P^{\prime[23]}}\cdots\sum_{P^{[L1]},\,P^{\prime[L1]}}\left[\prod_{\mu=1}^{L}V_{P^{-1[\mu,\,\mu+1]}P^{\prime[\mu,\,\mu+1]}}\right]\nonumber \\
\times & \prod_{\xi=1}^{t}\prod_{\mu=1}^{L}\left[\underbrace{\left(\delta_{a_{\xi,\,\sim\mu+1}^{[\mu]}\alpha_{P^{[\mu,\,\mu+1]}(\xi),\,\sim\mu+1}^{[\mu]}}\delta_{a_{\xi,\,\sim\mu}^{[\mu+1]}\alpha_{P^{[\mu,\,\mu+1]}(\xi),\,\sim\mu}^{[\mu+1]}}\delta_{b_{\xi,\,\sim\mu+1}^{[\mu]}\beta_{P^{\prime[\mu,\,\mu+1]}(\xi),\,\sim\mu+1}^{[\mu]}}\delta_{b_{\xi,\,\sim\mu}^{[\mu+1]}\beta_{P^{\prime[\mu,\,\mu+1]}(\xi),\,\sim\mu}^{[\mu+1]}}\right)}_{\text{(I)}}\underbrace{\left(\delta_{a_{\xi,\,\sim\mu-\sigma(\mu)}^{[\mu]}b_{\xi,\,\sim\mu+\sigma(\mu)}^{[\mu]}}\delta_{\alpha_{\xi,\,\sim\mu-\sigma(\mu)}^{[\mu]}\beta_{\xi,\,\sim\mu+\sigma(\mu)}^{[\mu]}}\right)}_{\text{(II)}}\right]\label{eq:Wg-Manybody-Middle1},
\end{align}
where we have rearranged the products $\prod_{\mu=1}^L \prod_{\xi=1}^t\cdots$ in the second last equation to the obtain the last equation.

Thanks to the periodic boundary conditions, we note the following
identity:

\begin{equation}
\prod_{\mu=1}^{L}\left(\delta_{a_{\xi,\,\sim\mu}^{[\mu+1]}\alpha_{P^{[\mu,\,\mu+1]}(\xi),\,\sim\mu}^{[\mu+1]}}\delta_{b_{\xi,\,\sim\mu}^{[\mu+1]}\beta_{P^{\prime[\mu,\,\mu+1]}(\xi),\,\sim\mu}^{[\mu+1]}}\right)=\prod_{\mu=1}^{L}\left(\delta_{a_{\xi,\,\sim\mu-1}^{[\mu]}\alpha_{P^{[\mu,\,\mu-1]}(\xi),\,\sim\mu-1}^{[\mu]}}\delta_{b_{\xi,\,\sim\mu-1}^{[\mu]}\beta_{P^{\prime[\mu,\,\mu-1]}(\xi),\,\sim\mu-1}^{[\mu]}}\right).
\end{equation}
Substituting above identity into the term $\prod_{\mu=1}^L\underbrace{\cdots}_\text{(I)}$ in Eq.~\eqref{eq:Wg-Manybody-Middle1}, we obtain
\begin{align}
&\prod_{\mu=1}^{L}\left(\underbrace{\delta_{a_{\xi,\,\sim\mu+1}^{[\mu]}\alpha_{P^{[\mu,\mu+1]}(\xi),\,\sim\mu+1}^{[\mu]}}\delta_{a_{\xi,\,\sim\mu}^{[\mu+1]}\alpha_{P^{[\mu,\mu+1]}(\xi),\,\sim\mu}^{[\mu+1]}}\delta_{b_{\xi,\,\sim\mu+1}^{[\mu]}\beta_{P^{\prime[\mu,\mu+1]}(\xi),\,\sim\mu+1}^{[\mu]}}\delta_{b_{\xi,\,\sim\mu}^{[\mu+1]}\beta_{P^{\prime[\mu,\mu+1]}(\xi),\,\sim\mu}^{[\mu+1]}}}_{\text{(I)}}\right) \nonumber \\
=&\prod_{\mu=1}^{L}\underbrace{\left(\delta_{a_{\xi,\,\sim\mu+1}^{[\mu]}\alpha_{P^{[\mu,\mu+1]}(\xi),\,\sim\mu+1}^{[\mu]}}\delta_{a_{\xi,\,\sim\mu-1}^{[\mu]}\alpha_{P^{[\mu-1,\,\mu]}(\xi),\,\sim\mu-1}^{[\mu]}}\delta_{b_{\xi,\,\sim\mu+1}^{[\mu]}\beta_{P^{\prime[\mu,\mu+1]}(\xi),\,\sim\mu+1}^{[\mu]}}\delta_{b_{\xi,\,\sim\mu-1}^{[\mu]}\beta_{P^{\prime[\mu-1,\mu]}(\xi),\,\sim\mu-1}^{[\mu]}}\right)}_\text{(III)}.
\end{align}
On the other hand, upon observing the following obvious identities
\begin{align}
\delta_{a_{\xi,\,\sim\mu+1}^{[\mu]}\alpha_{P^{[\mu,\,\mu+1]}(\xi),\,\sim\mu+1}^{[\mu]}}\delta_{a_{\xi,\,\sim\mu-1}^{[\mu]}\alpha_{P^{[\mu,\,\mu-1]}(\xi),\,\sim\mu-1}^{[\mu]}} & =\delta_{a_{\xi,\,\sim\mu+\sigma(\mu)}^{[\mu]}\alpha_{P^{[\mu,\,\mu+\sigma(\mu)]}(\xi),\,\sim\mu+\sigma(\mu)}^{[\mu]}}\delta_{a_{\xi,\,\sim\mu-\sigma(\mu)}^{[\mu]}\alpha_{P^{[\mu,\,\mu-\sigma(\mu)]}(\xi),\,\sim\mu-\sigma(\mu)}^{[\mu]}},\\
\delta_{b_{\xi,\,\sim\mu+1}^{[\mu]}\beta_{P^{\prime[\mu,\,\mu+1]}(\xi),\,\sim\mu+1}^{[\mu]}}\delta_{b_{\xi,\,\sim\mu-1}^{[\mu]}\beta_{P^{\prime[\mu,\,\mu-1]}(\xi),\,\sim\mu-1}^{[\mu]}} & =\delta_{b_{\xi,\,\sim\mu+\sigma(\mu)}^{[\mu]}\beta_{P^{\prime[\mu,\,\mu+\sigma(\mu)]}(\xi),\,\sim\mu+\sigma(\mu)}^{[\mu]}}\delta_{b_{\xi,\,\sim\mu-\sigma(\mu)}^{[\mu]}\beta_{P^{\prime[\mu,\,\mu-\sigma(\mu)]}(\xi),\,\sim\mu-\sigma(\mu)}^{[\mu]}},
\end{align}
we obtain
\begin{align}
 & \underbrace{\left(\delta_{a_{\xi,\,\sim\mu+1}^{[\mu]}\alpha_{P^{[\mu,\,\mu+1]}(\xi),\,\sim\mu+1}^{[\mu]}}\delta_{a_{\xi,\,\sim\mu-1}^{[\mu]}\alpha_{P^{[\mu,\,\mu-1]}(\xi),\,\sim\mu-1}^{[\mu]}}\delta_{b_{\xi,\,\sim\mu+1}^{[\mu]}\beta_{P^{\prime[\mu,\,\mu+1]}(\xi),\,\sim\mu+1}^{[\mu]}}\delta_{b_{\xi,\,\sim\mu-1}^{[\mu]}\beta_{P^{\prime[\mu,\,\mu-1]}(\xi),\,\sim\mu-1}^{[\mu]}}\right)}_\text{(III)}\underbrace{\left(\delta_{a_{\xi,\,\sim\mu-\sigma(\mu)}^{[\mu]}b_{\xi,\,\sim\mu+\sigma(\mu)}^{[\mu]}}\delta_{\alpha_{\xi,\,\sim\mu-\sigma(\mu)}^{[\mu]}\beta_{\xi,\,\sim\mu+\sigma(\mu)}^{[\mu]}}\right)}_\text{(II)}\nonumber \\
= & \left(\delta_{a_{\xi,\,\sim\mu+\sigma(\mu)}^{[\mu]}\alpha_{P^{[\mu,\,\mu+\sigma(\mu)]}(\xi),\,\sim\mu+\sigma(\mu)}^{[\mu]}}\underbrace{\delta_{a_{\xi,\,\sim\mu-\sigma(\mu)}^{[\mu]}\alpha_{P^{[\mu,\,\mu-\sigma(\mu)]}(\xi),\,\sim\mu-\sigma(\mu)}^{[\mu]}}\delta_{b_{\xi,\,\sim\mu+\sigma(\mu)}^{[\mu]}\beta_{P^{\prime[\mu,\,\mu+\sigma(\mu)]}(\xi),\,\sim\mu+\sigma(\mu)}^{[\mu]}}}_\text{(IV.1)}\delta_{b_{\xi,\,\sim\mu-\sigma(\mu)}^{[\mu]}\beta_{P^{\prime[\mu,\,\mu-\sigma(\mu)]}(\xi),\,\sim\mu-\sigma(\mu)}^{[\mu]}}\right)\left(\underbrace{\delta_{a_{\xi,\,\sim\mu-\sigma(\mu)}^{[\mu]}b_{\xi,\,\sim\mu+\sigma(\mu)}^{[\mu]}}}_{\text{(IV.2)}}\delta_{\alpha_{\xi,\,\sim\mu-\sigma(\mu)}^{[\mu]}\beta_{\xi,\,\sim\mu+\sigma(\mu)}^{[\mu]}}\right).\label{eq:Wg-Manybody-Middle2}
\end{align}
Furthermore, using the following identity 
\begin{equation}
\underbrace{\delta_{a_{\xi,\,\sim\mu-\sigma(\mu)}^{[\mu]}\alpha_{P^{[\mu,\,\mu-\sigma(\mu)]}(\xi),\,\sim\mu-\sigma(\mu)}^{[\mu]}}\delta_{b_{\xi,\,\sim\mu+\sigma(\mu)}^{[\mu]}\beta_{P^{\prime[\mu,\,\mu+\sigma(\mu)]}(\xi),\,\sim\mu+\sigma(\mu)}^{[\mu]}}}_\text{(IV.1)}\underbrace{\delta_{a_{\xi,\,\sim\mu-\sigma(\mu)}^{[\mu]}b_{\xi,\,\sim\mu+\sigma(\mu)}^{[\mu]}}}_\text{(IV.2)}=\underbrace{\delta_{\alpha_{P^{[\mu,\,\mu-\sigma(\mu)]}(\xi),\,\sim\mu-\sigma(\mu)}^{[\mu]}\beta_{P^{\prime[\mu,\,\mu+\sigma(\mu)]}(\xi),\,\sim\mu+\sigma(\mu)}^{[\mu]}}}_\text{(V)},\label{eq:delta-id-3}
\end{equation}
Eq.~\eqref{eq:Wg-Manybody-Middle2} can be further simplified as follows:
\begin{align}
 & \left(\delta_{a_{\xi,\,\sim\mu+\sigma(\mu)}^{[\mu]}\alpha_{P^{[\mu,\,\mu+\sigma(\mu)]}(\xi),\,\sim\mu+\sigma(\mu)}^{[\mu]}}\underbrace{\delta_{a_{\xi,\,\sim\mu-\sigma(\mu)}^{[\mu]}\alpha_{P^{[\mu,\,\mu-\sigma(\mu)]}(\xi),\,\sim\mu-\sigma(\mu)}^{[\mu]}}\delta_{b_{\xi,\,\sim\mu+\sigma(\mu)}^{[\mu]}\beta_{P^{\prime[\mu,\,\mu+\sigma(\mu)]}(\xi),\,\sim\mu+\sigma(\mu)}^{[\mu]}}}_\text{(IV.1)}\delta_{b_{\xi,\,\sim\mu-\sigma(\mu)}^{[\mu]}\beta_{P^{\prime[\mu,\,\mu-\sigma(\mu)]}(\xi),\,\sim\mu-\sigma(\mu)}^{[\mu]}}\right)\left(\underbrace{\delta_{a_{\xi,\,\sim\mu-\sigma(\mu)}^{[\mu]}b_{\xi,\,\sim\mu+\sigma(\mu)}^{[\mu]}}}_{\text{(IV.2)}}\delta_{\alpha_{\xi,\,\sim\mu-\sigma(\mu)}^{[\mu]}\beta_{\xi,\,\sim\mu+\sigma(\mu)}^{[\mu]}}\right)\nonumber \\
= & \delta_{a_{\xi,\,\sim\mu+\sigma(\mu)}^{[\mu]}\alpha_{P^{[\mu,\,\mu+\sigma(\mu)]}(\xi),\,\sim\mu+\sigma(\mu)}^{[\mu]}}\underbrace{\delta_{\alpha_{P^{[\mu,\,\mu-\sigma(\mu)]}(\xi),\,\sim\mu-\sigma(\mu)}^{[\mu]}\beta_{P^{\prime[\mu,\,\mu+\sigma(\mu)]}(\xi),\,\sim\mu+\sigma(\mu)}^{[\mu]}}}_{\text{(V)}}\delta_{b_{\xi,\,\sim\mu-\sigma(\mu)}^{[\mu]}\beta_{P^{\prime[\mu,\,\mu-\sigma(\mu)]}(\xi),\,\sim\mu-\sigma(\mu)}^{[\mu]}}\delta_{\alpha_{\xi,\,\sim\mu-\sigma(\mu)}^{[\mu]}\beta_{\xi,\,\sim\mu+\sigma(\mu)}^{[\mu]}}\nonumber \\
= & \delta_{a_{\xi,\,\sim\mu+\sigma(\mu)}^{[\mu]}\alpha_{P^{[\mu,\,\mu+\sigma(\mu)]}(\xi),\,\sim\mu+\sigma(\mu)}^{[\mu]}}\delta_{b_{\xi,\,\sim\mu-\sigma(\mu)}^{[\mu]}\beta_{P^{\prime[\mu,\,\mu-\sigma(\mu)]}(\xi),\,\sim\mu-\sigma(\mu)}^{[\mu]}}\underbrace{\delta_{\alpha_{\xi,\,\sim\mu-\sigma(\mu)}^{[\mu]}\beta_{\xi,\,\sim\mu+\sigma(\mu)}^{[\mu]}}\delta_{\alpha_{P^{[\mu,\,\mu-\sigma(\mu)]}(\xi),\,\sim\mu-\sigma(\mu)}^{[\mu]}\beta_{P^{\prime[\mu,\,\mu+\sigma(\mu)]}(\xi),\,\sim\mu+\sigma(\mu)}^{[\mu]}}}_{\text{(VI)}},
\end{align}
where  we have used Eq.~\eqref{eq:delta-id-3} in the first equation and rearrange the product in the last equation.
To summarize,  the average of the homogeneous product of the Floquet operators can be expressed as follows: 
\begin{align}
 & \langle U_{\{a_{1,\,\sim\mu+\sigma(\mu)}^{[\mu]},\,b_{1,\,\sim\mu-\sigma(\mu)}^{[\mu]}\}}\cdots U_{\{a_{t,\,\sim\mu+\sigma(\mu)}^{[\mu]},\,b_{t,\,\sim\mu-\sigma(\mu)}^{[\mu]}\}}U_{\{\alpha_{1,\,\sim\mu+\sigma(\mu)}^{[\mu]},\,\beta_{1,\,\sim\mu-\sigma(\mu)}^{[\mu]}\}}^{*}\cdots U_{\{\alpha_{t,\,\sim\mu+\sigma(\mu)}^{[\mu]},\,\beta_{t,\,\sim\mu-\sigma(\mu)}^{[\mu]}\}}^{*}\rangle\nonumber \\
= & \sum_{P^{[12]},\,P^{\prime[12]}}\sum_{P^{[23]},\,P^{\prime[23]}}\cdots\sum_{P^{[L1]},\,P^{\prime[L1]}}\left[\prod_{\mu=1}^{L}V_{P^{-1[\mu,\,\mu+1]}P^{\prime[\mu,\,\mu+1]}}\right]\nonumber \\
\times & \prod_{\xi=1}^{t}\prod_{\mu=1}^{L}\left(\delta_{a_{\xi,\,\sim\mu+\sigma(\mu)}^{[\mu]}\alpha_{P^{[\mu,\,\mu+\sigma(\mu)]}(\xi),\,\sim\mu+\sigma(\mu)}^{[\mu]}}\delta_{b_{\xi,\,\sim\mu-\sigma(\mu)}^{[\mu]}\beta_{P^{\prime[\mu,\,\mu-\sigma(\mu)]}(\xi),\,\sim\mu-\sigma(\mu)}^{[\mu]}}\underbrace{\delta_{\alpha_{\xi,\,\sim\mu-\sigma(\mu)}^{[\mu]}\beta_{\xi,\,\sim\mu+\sigma(\mu)}^{[\mu]}}\delta_{\alpha_{P^{[\mu,\,\mu-\sigma(\mu)]}(\xi),\,\sim\mu-\sigma(\mu)}^{[\mu]}\beta_{P^{\prime[\mu,\,\mu+\sigma(\mu)]}(\xi),\,\sim\mu+\sigma(\mu)}^{[\mu]}}}_\text{(VI)}\right).\label{eq:Weingarten-Manybody-del-sum}
\end{align}
We observe that, as shown in the Fig.~\ref{fig:dummy indices}, the indices in the expression $\prod_{\xi=1}^t\underbrace{\cdots}_{\text{(VI)}}$
are dummy because they appear twice. Therefore, the free indices left in Eq.~\eqref{eq:Weingarten-Manybody-del-sum}
are the $\{a_{\xi,\,\sim\mu+\sigma(\mu)}^{[\mu]},\,b_{\xi,\,\sim\mu-\sigma(\mu)}^{[\mu]}\}$. 
Furthermore, upon noting that 
\begin{equation}
\delta_{\alpha_{P(\xi)},\,\beta_{P^{\prime}(\xi)}}=\delta_{\alpha_{\xi},\,\beta_{P^{-1}P^{\prime}(\xi)}}, \quad \prod_{\xi=1}^{t} \delta_{\beta_{\xi},\,\beta_{P^{-1}P^{\prime}(\xi)}}=q^{\text{\# of cycles in } P^{-1}P^{\prime} },
\end{equation}
we observe
\begin{equation}
\prod_{\xi=1}^{t}\underbrace{\delta_{\alpha_{\xi,\,\sim\mu-\sigma(\mu)}^{[\mu]}\beta_{\xi,\,\sim\mu+\sigma(\mu)}^{[\mu]}}\delta_{\alpha_{P^{[\mu,\,\mu-\sigma(\mu)]}(\xi),\,\sim\mu-\sigma(\mu)}^{[\mu]}\beta_{P^{\prime[\mu,\,\mu+\sigma(\mu)]}(\xi),\,\sim\mu+\sigma(\mu)}^{[\mu]}}}_\text{(VI)}=\prod_{\xi=1}^{t}\delta_{\beta_{\xi,\,\sim\mu+\sigma(\mu)}^{[\mu]}\beta_{P^{-1[\mu,\,\mu-\sigma(\mu)]}P^{\prime[\mu,\,\mu+\sigma(\mu)]}(\xi),\,\sim\mu+\sigma(\mu)}^{[\mu]}}=q^{\bar{s}^{[\mu,\,\mu\pm\sigma(\mu)]}}\label{eq:delta-sum},
\end{equation}
where $\bar{s}^{[\mu,\,\mu\pm\sigma(\mu)]}$ is the cycles in the permutation $P^{-1[\mu,\,\mu-\sigma(\mu)]}P^{\prime[\mu,\,\mu+\sigma(\mu)]}$. 
As a consequence,  we obtain Eq.~(\ref{eq:Wg-Manybody})
in the main text. 
\begin{figure}
\includegraphics[scale=1]{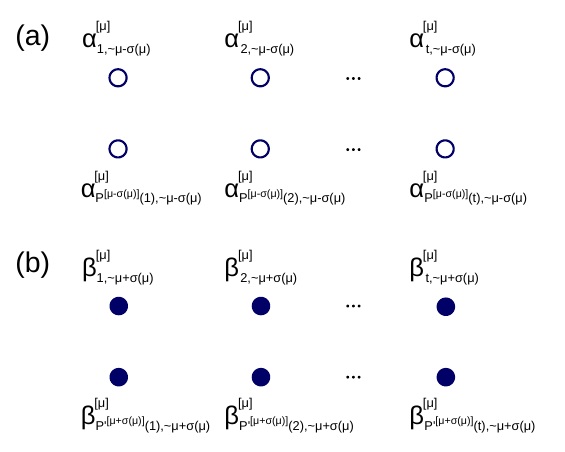}
\caption{A list of the subscripts in Eq.~\eqref{eq:delta-sum}. \label{fig:dummy indices} }
\end{figure}

\subsection{Proof of the factorization principle~\ref{prp:many-body} }

Let us first consider the case of two sites. In the large $q$ limit, the leading
order of $V_{P^{-1[12]}P^{\prime[12]}}\sim(q^{2})^{-2t+s}$, where
we recall that $s$ is the number of disjoint cycles in the permutation
$P^{-1[12]}P^{\prime[12]}$. In this limit, the global Weingarten
function factorizes into the local contributions from site $1$ and
$2$ respectively. More precisely,
\begin{equation}
V_{P^{-1[12]}P^{\prime[12]}}=V_{P^{-1[1]}P^{\prime[1]}}V_{P^{-1[2]}P^{\prime[2]}}+\mathscr{O}(q^{2(-2t+s-2)}),\label{eq:Wg-factorization}
\end{equation}
where $V_{P^{-1[\mu]}P^{\prime[\mu]}}$
is the Weingarten function for the Hilbert space at site $\mu$, which is
of dimension $q$. It is worth noting that Eq.~(\ref{eq:Wg-factorization})
generally does not hold in the non-asymptotic regime where $q$ is finite.

On the other hand, since $\delta_{\bm{a}_{\xi}\bm{\alpha}_{P(\xi)}}=\delta_{a_{\xi}^{(1)}\alpha_{\xi}^{(1)}}\delta_{a_{\xi}^{(2)}\alpha_{\xi}^{(2)}}$,
Eq.~\eqref{eq:two-site-Wg} in the main text can be alternatively written as 
\begin{align}
 & \langle\mathsf{U}_{\bm{a}_{1}^{[12]}\bm{b}_{1}^{[12]}}\cdots\mathsf{U}_{\bm{a}_{t}^{[12]}\bm{b}_{t}^{[12]}}\mathsf{U}_{\bm{\alpha}_{1}^{[12]}\bm{\beta}_{1}}^{*}\cdots\mathsf{U}_{\bm{\alpha}_{t}^{[12]}\bm{\beta}_{t}^{[12]}}^{*}\rangle\nonumber \\
= & \sum_{P^{[1]},\,P^{\prime[1]}}\sum_{Q^{[2]},\,Q^{\prime[2]}}\left(V_{P^{-1[1]}P^{\prime[1]}}\prod_{\xi=1}^{t}\delta_{a_{\xi,\sim2}^{[1]}\alpha_{P^{[1]}(\xi),\sim2}^{[1]}}\delta_{b_{\xi,\sim2}^{[1]}\beta_{P^{\prime[1]}(\xi),\sim2}^{[1]}}\right)\left(V_{Q^{-1[2]}Q^{\prime[2]}}\prod_{\xi=1}^{t}\delta_{a_{\xi,\sim1}^{[2]}\alpha_{Q^{[2]}(\xi),\sim1}^{[2]}}\delta_{b_{\xi,\sim1}^{[2]}\beta_{Q^{\prime[2]}(\xi),\sim1}^{[2]}}\right)\delta_{P^{[1]},\,Q^{[2]}}\delta_{P^{\prime[1]},\,Q^{\prime[2]}}.\label{eq:Two-site-effective-Weingarten}
\end{align}
Similarly,
\begin{align}
 & \langle\mathsf{U}_{\bm{a}_{1}^{[\mu\nu]}\bm{b}_{1}^{[\mu\nu]}}\cdots\mathsf{U}_{\bm{a}_{t}^{[\mu\nu]}\bm{b}_{t}^{[\mu\nu]}}\mathsf{U}_{\bm{\alpha}_{1}^{[\mu\nu]}\bm{\beta}_{1}^{[\mu\nu]}}^{*}\cdots\mathsf{U}_{\bm{\alpha}_{t}^{[\mu\nu]}\bm{\beta}_{t}^{[\mu\nu]}}^{*}\rangle\nonumber \\
= & \sum_{P^{[\mu]},\,P^{\prime[\mu]}}\sum_{Q^{[\nu]},\,Q^{\prime[\nu]}}\left(V_{P^{-1[\mu]}P^{\prime[\mu]}}\prod_{\xi=1}^{t}\delta_{a_{\xi,\,\sim\nu}^{[\mu]}\alpha_{P^{[\mu]}(\xi),\,\sim\nu}^{[\mu]}}\delta_{b_{\xi,\,\sim\nu}^{[\mu]}\beta_{P^{\prime[\mu]}(\xi),\,\sim\nu}^{[\mu]}}\right)\left(V_{Q^{-1[\nu]}Q^{\prime[\nu]}}\prod_{\xi=1}^{t}\delta_{a_{\xi,\,\sim\mu}^{[\nu]}\alpha_{Q^{[\nu]}(\xi),\,\sim\mu}^{[\nu]}}\delta_{b_{\xi,\,\sim\mu}^{[\nu]}\beta_{Q^{\prime[\nu]}(\xi),\,\sim\mu}^{[\nu]}}\right)\delta_{P^{[\mu]},\,Q^{[\nu]}}\delta_{P^{\prime[\mu]},\,Q^{\prime[\nu]}}.
\end{align}

We consider operators that are supported on a single-site, i.e., $A_{\xi}^{(12)}=A_{\xi}^{(1)}\otimes A_{\xi}^{(2)}$, $B_{\xi}^{(12)}=B_{\xi}^{(1)}\otimes B_{\xi}^{(2)}$.
Intuitively, Eq.~(\ref{eq:Two-site-effective-Weingarten}) simply
implies that one can first perform single-site contractions at sites
$1$ and $2$, respectively. $\delta_{P^{[1]},\,Q^{[2]}}\delta_{P^{\prime[1]},\,Q^{\prime(2)}}$
imposes a restriction that the contractions at site $1$ and $2$
must be the same, which can be diagrammatically interpreted as the
thin solids lines at both sides must be joined in the same pattern.
We note that this restriction is nothing but the bond constraint put
forward in Ref.~\citep{chan2018solution}. To see above intuition
more clear, let us consider $f=\mathrm{Tr}(UA_{1}\cdots UA_{t}U^{\dagger}B_{1}\cdots U^{\dagger}B_{t})$.
Then 
\begin{align}
 & \mathrm{Tr}(UA_{1}\cdots UA_{t}U^{\dagger}B_{1}\cdots U^{\dagger}B_{t})\nonumber \\
= & \mathsf{U}_{\bm{a}_{1}^{[12]}\bm{b}_{1}^{[12]}}\cdots\mathsf{U}_{\bm{a}_{t}^{[12]}\bm{b}_{t}^{[12]}}\mathsf{U}_{\bm{\alpha}_{1}^{[12]}\bm{\beta}_{1}^{[12]}}^{*}\cdots\mathsf{U}_{\bm{\alpha}_{t}^{[12]}\bm{\beta}_{t}^{[12]}}^{*}\nonumber \\
\times & \prod_{\mu=1}^{2}[A_{1}^{[\mu]}]_{b_{1,\sim\mu+1}^{[\mu]}a_{2,\sim\mu+1}^{[\mu]}}[A_{2}^{[\mu]}]_{b_{2,\sim\mu+1}^{[\mu]}a_{3,\sim\mu+1}^{[\mu]}}\cdots[A_{t}^{[\mu]}]_{b_{t,\sim\mu+1}^{[\mu]}\beta_{1,\sim\mu+1}^{[\mu]}}[B_{1}^{[\mu]}]_{\alpha_{1,\sim\mu+1}^{[\mu]}\beta_{2,\sim\mu+1}^{[\mu]}}[B_{2}^{[\mu]}]_{\alpha_{2,\sim\mu+1}^{[\mu]}\beta_{3,\sim\mu+1}^{[\mu]}}\cdots[B_{t}^{[\mu]}]_{\alpha_{t,\sim\mu+1}^{[\mu]}a_{1,\sim\mu+1}^{[\mu]}},
\end{align}
where the periodic boundary condition is assumed. Now applying Eq.~(\ref{eq:Two-site-effective-Weingarten}),
we obtain 
\begin{align}
 & \langle\mathrm{Tr}(UA_{1}\cdots UA_{t}U^{\dagger}B_{1}\cdots U^{\dagger}B_{t})\rangle\nonumber \\
= & \prod_{\mu=1}^{L}\langle\mathsf{U}_{\bm{a}_{1}^{[12]}\bm{b}_{1}^{[12]}}\cdots\mathsf{U}_{\bm{a}_{t}^{[12]}\bm{b}_{t}^{[12]}}\mathsf{U}_{\bm{\alpha}_{1}^{[12]}\bm{\beta}_{1}^{[12]}}^{*}\cdots\mathsf{U}_{\bm{\alpha}_{t}^{[12]}\bm{\beta}_{t}^{[12]}}^{*}\rangle\nonumber \\
\times & \prod_{\mu=1}^{2}[A_{1}^{[\mu]}]_{b_{1,\sim\mu+1}^{[\mu]}a_{2,\sim\mu+1}^{[\mu]}}[A_{2}^{[\mu]}]_{b_{2,\sim\mu+1}^{[\mu]}a_{3,\sim\mu+1}^{[\mu]}}\cdots[A_{t}^{[\mu]}]_{b_{t,\sim\mu+1}^{[\mu]}\beta_{1,\sim\mu+1}^{[\mu]}}[B_{1}^{[\mu]}]_{\alpha_{1,\sim\mu+1}^{[\mu]}\beta_{2,\sim\mu+1}^{[\mu]}}[B_{2}^{[\mu]}]_{\alpha_{2,\sim\mu+1}^{[\mu]}\beta_{3,\sim\mu+1}^{[\mu]}}\cdots[B_{t}^{[\mu]}]_{\alpha_{t,\sim\mu+1}^{[\mu]}a_{1,\sim\mu+1}^{[\mu]}}\nonumber \\
= & \sum_{P^{[1]},\,P^{\prime[1]}}\sum_{Q^{[2]},\,Q^{\prime[2]}}\left(V_{P^{-1[1]}P^{\prime[1]}}\prod_{\xi=1}^{t}\delta_{a_{\xi,\,\sim2}^{[1]}\alpha_{P^{[1]}(\xi),\,\sim2}^{[1]}}\delta_{b_{\xi,\,\sim2}^{[1]}\beta_{P^{\prime[1]}(\xi),\,\sim2}^{[1]}}[A_{1}^{[1]}]_{b_{1,\sim2}^{[1]}a_{2,\sim2}^{[1]}}[A_{2}^{[1]}]_{b_{2,\sim2}^{[1]}a_{3,\sim2}^{[1]}}\cdots[A_{t}^{[1]}]_{b_{t,\sim2}^{[1]}\beta_{1,\sim2}^{[1]}}[B_{1}^{[1]}]_{\alpha_{1,\sim2}^{[1]}\beta_{2,\sim2}^{[1]}}[B_{2}^{[1]}]_{\alpha_{2,\sim2}^{[1]}\beta_{3,\sim2}^{[1]}}\cdots[B_{t}^{[1]}]_{\alpha_{t,\sim2}^{[1]}a_{1,\sim2}^{[1]}}\right)\nonumber \\
\times & \left(V_{Q^{-1[2]}Q^{\prime[2]}}\prod_{\xi=1}^{t}\delta_{a_{\xi,\,\sim1}^{[2]}\alpha_{Q^{[2]}(\xi),\,\sim1}^{[2]}}\delta_{b_{\xi}^{[2]}\beta_{Q^{\prime[2]}(\xi),\sim1}^{[2]}}[A_{1}^{[2]}]_{b_{1,\sim1}^{[2]}a_{2,\sim1}^{[2]}}[A_{2}^{[2]}]_{b_{2,\sim1}^{[2]}a_{3,\sim1}^{[2]}}\cdots[A_{t}^{[2]}]_{b_{t,\sim1}^{[2]}\beta_{1,\sim1}^{[2]}}[B_{1}^{[2]}]_{\alpha_{1,\sim1}^{[2]}\beta_{2,\sim1}^{[2]}}[B_{2}^{[2]}]_{\alpha_{2,\sim1}^{[2]}\beta_{3,\sim1}^{[2]}}\cdots[B_{t}^{[2]}]_{\alpha_{t,\sim1}^{[2]}a_{1,\sim1}^{[2]}}\right)\nonumber \\
\times & \delta_{P^{[1]},\,Q^{[2]}}\delta_{P^{\prime[1]},\,Q^{\prime[2]}}.
\end{align}
Upon defining 
\begin{align}
T_{P^{[\mu]},\,P^{\prime[\mu]}}=&\prod_{\xi=1}^{t}\delta_{a_{\xi,\,\sim\mu+1}^{[\mu]}\alpha_{P^{[\mu]}(\xi),\,\sim\mu+1}^{[\mu]}}\delta_{b_{\xi,\,\sim\mu+1}^{[\mu]}\beta_{P^{\prime[\mu]}(\xi),\,\sim\mu+1}^{[\mu]}} \nonumber \\
&\times[A_{1}^{[\mu]}]_{b_{1,\sim\mu+1}^{[\mu]}a_{2,\sim\mu+1}^{[\mu]}}[A_{2}^{[\mu]}]_{b_{2,\sim\mu+1}^{[\mu]}a_{3,\sim\mu+1}^{[\mu]}}\cdots[A_{t}^{[\mu]}]_{b_{t,\sim\mu+1}^{[\mu]}\beta_{1,\sim\mu+1}^{[\mu]}}[B_{1}^{[\mu]}]_{\alpha_{1,\sim\mu+1}^{[\mu]}\beta_{2,\sim\mu+1}^{[\mu]}}[B_{2}^{[\mu]}]_{\alpha_{2,\sim\mu+1}^{[\mu]}\beta_{3,\sim\mu+1}^{[\mu]}}\cdots[B_{t}^{[\mu]}]_{\alpha_{t,\sim\mu+1}^{[\mu]}a_{1,\sim\mu+1}^{[\mu]}},
\end{align}
where $\mu=1,\,2$, we can rewrite 
\begin{equation}
\mathrm{Tr}(UA_{1}\cdots UA_{t}U^{\dagger}B_{1}\cdots U^{\dagger}B_{t})=\sum_{P^{[1]},\,P^{\prime[1]}}\sum_{Q^{[2]},\,Q^{\prime[2]}}V_{P^{-1[1]}P^{\prime[1]}}T_{P^{[1]},\,P^{\prime[1]}}V_{Q^{-1[2]}Q^{\prime[2]}}T_{Q^{[2]},\,Q^{\prime[2]}}\delta_{P^{[1]},\,Q^{[2]}}\delta_{P^{\prime[1]},\,Q^{\prime[2]}}.
\label{eq:Wg-ProdTr-Two-site}
\end{equation}
Eq.~(\ref{eq:Wg-ProdTr-Two-site}) is the generalization of the Weingarten
trace formula to the case of two-site CUE and can
be further generalized to generic $L>2$. 
\begin{figure}
\centering
\includegraphics[scale=0.8]{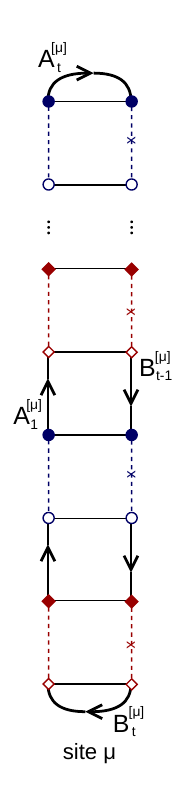}
\caption{\label{fig:many-body-diagram}The many-body precontraction diagram
of $\mathrm{Tr}(UA_{1}\cdots UA_{t}U^{\dagger}B_{1}\cdots U^{\dagger}B_{t})$ is a product of the single-site diagrams shown in the figure,
where $A_{\xi}=\prod_{\mu}A_{\xi}^{[\mu]}$ and $B_{\xi}=\prod_{\mu}B_{\xi}^{[\mu]}$ are made of  product of operators
acting on a single site. }
\end{figure}

Unitary gates acting on different pair of sites must be considered
as separate and averaged over independently. Therefore, in the most generic
case, one needs to apply the Weingarten formula a number of
times that equals the number of different unitary gates. For simplicity,
we consider the a linear functional of degree $t$ with only one group
of trace products, i.e., $f=\mathrm{Tr}(UA_{1}\cdots UA_{t}U^{\dagger}B_{1}\cdots U^{\dagger}B_{t})$.
The argument and intuition generalizes to an arbitrary number
of groups of trace products, as long as the assumption that all the
operators are supported on single sites holds. Similarly to the RMM
case, due to the fact that $A_{\xi}$ and $B_{\xi}$ are only supported
on a single site, we can write
\begin{align}
\mathrm{Tr}(UA_{1}\cdots UA_{t}U^{\dagger}B_{1}\cdots U^{\dagger}B_{t}) & =\prod_{\mu=1}^{L}\mathsf{U}_{\bm{a}_{1}^{[\mu,\,\mu+1]}\bm{b}_{1}^{[\mu,\,\mu+1]}}\cdots\prod_{\mu=1}^{L}\mathsf{U}_{\bm{a}_{t}^{[\mu,\,\mu+1]}\bm{b}_{t}^{[\mu,\,\mu+1]}}\prod_{\mu=1}^{L}\mathsf{U}_{\bm{\alpha}_{1}^{[\mu,\,\mu+1]}\bm{\beta}_{1}^{[\mu,\,\mu+1]}}^{*}\cdots\prod_{\mu=1}^{L}\mathsf{U}_{\bm{\alpha}_{t}^{[\mu,\,\mu+1]}\bm{\beta}_{t}^{[\mu,\,\mu+1]}}^{*}\nonumber \\
 & \times\prod_{\mu=1}^{L}\delta_{a_{1,\,\sim\mu-\sigma(\mu)}^{[\mu]}b_{1,\,\sim\mu+\sigma(\mu)}^{[\mu]}}\cdots\prod_{\mu=1}^{L}\delta_{a_{t,\,\sim\mu-\sigma(\mu)}^{[\mu]}b_{t,\,\sim\mu+\sigma(\mu)}^{[\mu]}}\prod_{\mu=1}^{L}\delta_{\alpha_{1,\,\sim\mu-\sigma(\mu)}^{[\mu]}\beta_{1,\,\sim\mu+\sigma(\mu)}^{[\mu]}}\cdots\prod_{\mu=1}^{L}\delta_{\alpha_{t,\,\sim\mu-\sigma(\mu)}^{[\mu]}\beta_{t,\,\sim\mu+\sigma(\mu)}^{[\mu]}}\nonumber \\
 & \times\prod_{\mu}[A_{1}^{[\mu]}]_{b_{1,\sim\mu-\sigma(\mu)}^{[\mu]}a_{2,\sim\mu+\sigma(\mu)}^{[\mu]}}\prod_{\mu}[A_{2}^{[\mu]}]_{b_{2,\sim\mu-\sigma(\mu)}^{[\mu]}a_{3,\sim\mu+\sigma(\mu)}^{[\mu]}}\cdots\prod_{\mu}[A_{t}^{[\mu]}]_{b_{t,\sim\mu-\sigma(\mu)}^{[\mu]}\beta_{1,\sim\mu-\sigma(\mu)}^{[\mu]}}\nonumber \\
 & \times\prod_{\mu}[B_{1}^{[\mu]}]_{\alpha_{1,\sim\mu+\sigma(\mu)}^{[\mu]}\beta_{2,\sim\mu-\sigma(\mu)}^{[\mu]}}\prod_{\mu}[B_{2}^{[\mu]}]_{\alpha_{2,\sim\mu+\sigma(\mu)}^{[\mu]}\beta_{3,\sim\mu-\sigma(\mu)}^{[\mu]}}\cdots\prod_{\mu}[B_{t}^{[\mu]}]_{\alpha_{t,\sim\mu+\sigma(\mu)}^{[\mu]}a_{1,\sim\mu+\sigma(\mu)}^{[\mu]}},
\end{align}
where we note that the indices of the matrix elements of the operators
are:
\begin{equation}
\prod_{\mu}[A_{1}^{[\mu]}]_{b_{1,\sim\mu-\sigma(\mu)}^{[\mu]}a_{2,\sim\mu+\sigma(\mu)}^{[\mu]}},\,\cdots,\,\prod_{\mu}[A_{t-1}^{[\mu]}]_{b_{t-1,\sim\mu-\sigma(\mu)}^{[\mu]}a_{t,\sim\mu+\sigma(\mu)}^{[\mu]}},\,\prod_{\mu}[A_{t}^{[\mu]}]_{b_{t,\sim\mu-\sigma(\mu)}^{[\mu]}\beta_{1,\sim\mu-\sigma(\mu)}^{[\mu]}} ,
\end{equation}
\begin{equation}
\prod_{\mu}[B_{1}^{[\mu]}]_{\alpha_{1,\sim\mu+\sigma(\mu)}^{[\mu]}\beta_{2,\sim\mu-\sigma(\mu)}^{[\mu]}},\,\cdots,\,\prod_{\mu}[B_{t-1}^{[\mu]}]_{\alpha_{t-1,\sim\mu+\sigma(\mu)}^{[\mu]}\beta_{2,\sim\mu-\sigma(\mu)}^{[\mu]}},\,\prod_{\mu}[B_{t}^{[\mu]}]_{\alpha_{t,\sim\mu+\sigma(\mu)}^{[\mu]}a_{1,\sim\mu+\sigma(\mu)}^{[\mu]}} .
\end{equation}
Note that
\begin{align}
 & \prod_{\mu=1}^{L}\langle\mathsf{U}_{\bm{a}_{1}^{[\mu,\,\mu+1]}\bm{b}_{1}^{[\mu,\,\mu+1]}}\cdots\mathsf{U}_{\bm{a}_{t}^{[\mu,\,\mu+1]}\bm{b}_{t}^{[\mu,\,\mu+1]}}\mathsf{U}_{\bm{\alpha}_{1}^{[\mu,\,\mu+1]}\bm{\beta}_{1}^{[\mu,\,\mu+1]}}^{*}\cdots\mathsf{U}_{\bm{\alpha}_{t}^{[\mu,\,\mu+1]}\bm{\beta}_{t}^{[\mu,\,\mu+1]}}^{*}\rangle\nonumber \\
= & \prod_{\mu=1}^{L}\left\{ \sum_{P^{[\mu]},\,P^{\prime[\mu]}}\sum_{Q^{[\mu+1]},\,Q^{\prime[\mu+1]}}\left(V_{P_{\sim\mu+1}^{-1[\mu]}P_{\sim\mu+1}^{\prime[\mu]}}\prod_{\xi=1}^{t}\delta_{a_{\xi,\,\sim\mu+1}^{[\mu]}\alpha_{P^{[\mu]}(\xi),\,\sim\mu+1}^{[\mu]}}\delta_{b_{\xi,\,\sim\mu+1}^{[\mu]}\beta_{P^{\prime[\mu]}(\xi),\,\sim\mu+1}^{[\mu]}}\right)\right.\nonumber \\
 & \left.\times\left(V_{Q_{\sim\mu}^{-1[\mu+1]}Q_{\sim\mu}^{\prime[\mu+1]}}\prod_{\xi=1}^{t}\delta_{a_{\xi,\,\sim\mu}^{[\mu+1]}\alpha_{Q^{[\mu+1]}(\xi),\,\sim\mu}^{[\mu+1]}}\delta_{b_{\xi,\,\sim\mu}^{[\mu+1]}\beta_{Q^{\prime[\mu+1]}(\xi),\,\sim\mu}^{[\mu+1]}}\right)\delta_{P^{[\mu]},\,Q^{[\mu+1]}}\delta_{P^{\prime[\mu]},\,Q^{\prime[\mu+1]}}\right\} .
\end{align}
Using the identity
\begin{align}
 & \prod_{\mu=1}^{L}\left\{ \sum_{P^{[\mu]},\,P^{\prime[\mu]}}\sum_{Q^{[\mu+1]},\,Q^{\prime[\mu+1]}}F\left(P^{[\mu]},\,P^{\prime[\mu]},\,Q^{[\mu+1]},\,Q^{\prime[\mu+1]}\right)\right\} \nonumber \\
= & \sum_{P^{[1]},\,P^{\prime[1]}}\sum_{Q^{[1]},\,Q^{\prime[1]}}\cdots\sum_{P^{[L]},\,P^{\prime[L]}}\sum_{Q^{[L]},\,Q^{\prime[L]}}F\left(P^{[1]},\,P^{\prime[1]},\,Q^{[2]},\,Q^{\prime[2]}\right)\cdots F\left(P^{[L]},\,P^{\prime[L]},\,Q^{[L+1]},\,Q^{\prime[L+1]}\right), \label{eq:Prod-Sum-Id-PQ}
\end{align}
we find 
\begin{align}
 & \prod_{\mu=1}^{L}\langle\mathsf{U}_{\bm{a}_{1}^{[\mu,\,\mu+1]}\bm{b}_{1}^{[\mu,\,\mu+1]}}\cdots\mathsf{U}_{\bm{a}_{t}^{[\mu,\,\mu+1]}\bm{b}_{t}^{[\mu,\,\mu+1]}}\mathsf{U}_{\bm{\alpha}_{1}^{[\mu,\,\mu+1]}\bm{\beta}_{1}^{[\mu,\,\mu+1]}}^{*}\cdots\mathsf{U}_{\bm{\alpha}_{t}^{[\mu,\,\mu+1]}\bm{\beta}_{t}^{[\mu,\,\mu+1]}}^{*}\rangle\nonumber \\
= & \sum_{P^{[1]},\,P^{\prime[1]}}\sum_{Q^{[1]},\,Q^{\prime[1]}}\cdots\sum_{P^{[L]},\,P^{\prime[L]}}\sum_{Q^{[L]},\,Q^{\prime[L]}}\prod_{\mu=1}^{L}\left(V_{P^{-1[\mu]}P^{\prime[\mu]}}\prod_{\xi=1}^{t}\delta_{a_{\xi,\,\sim\mu+1}^{[\mu]}\alpha_{P^{[\mu]}(\xi),\,\sim\mu+1}^{[\mu]}}\delta_{b_{\xi,\,\sim\mu+1}^{[\mu]}\beta_{P^{\prime[\mu]}(\xi),\,\sim\mu+1}^{[\mu]}}\right)\nonumber \\
\times & \left(V_{Q^{-1[\mu]}Q^{\prime[\mu]}}\prod_{\xi=1}^{t}\delta_{a_{\xi,\,\sim\mu-1}^{[\mu]}\alpha_{Q^{[\mu]}(\xi),\,\sim\mu-1}^{[\mu]}}\delta_{b_{\xi,\,\sim\mu-1}^{[\mu]}\beta_{Q^{\prime[\mu]}(\xi),\,\sim\mu-1}^{[\mu]}}\right)\left(\prod_{\mu=1}^{L}\delta_{P^{[\mu]},\,Q^{[\mu+1]}}\delta_{P^{\prime[\mu]},\,Q^{\prime[\mu+1]}}\right).
\end{align}
For the sake of convenience, we introduce the short hand notation
\begin{equation}
\mathcal{B}[\{P^{[\mu]},\,Q^{[\mu]}\}]=\prod_{\mu=1}^{L}\delta_{P^{[\mu]},\,Q^{[\mu+1]}}\delta_{P^{\prime[\mu]},\,Q^{\prime[\mu+1]}}.
\end{equation}
When the multi-site Weingarten formula is applied to the linear map
$f$, we get
\begin{align}
 & \langle\mathrm{Tr}(UA_{1}\cdots UA_{t}U^{\dagger}B_{1}\cdots U^{\dagger}B_{t})\rangle\nonumber \\
= & \prod_{\mu=1}^{L}\langle U_{\bm{a}_{1}^{[\mu,\,\mu+1]}\bm{b}_{1}^{[\mu,\,\mu+1]}}\cdots U_{\bm{a}_{t}^{[\mu,\,\mu+1]}\bm{b}_{t}^{[\mu,\,\mu+1]}}U_{\bm{\alpha}_{1}^{[\mu,\,\mu+1]}\bm{\beta}_{1}^{[\mu,\,\mu+1]}}^{*}\cdots U_{\bm{\alpha}_{t}^{[\mu,\,\mu+1]}\bm{\beta}_{t}^{[\mu,\,\mu+1]}}^{*}\rangle\nonumber \\
\times & \prod_{\mu=1}^{L}\delta_{a_{1,\,\sim\mu-\sigma(\mu)}^{[\mu]}b_{1,\,\sim\mu+\sigma(\mu)}^{[\mu]}}\cdots\prod_{\mu=1}^{L}\delta_{a_{t,\,\sim\mu-\sigma(\mu)}^{[\mu]}b_{t,\,\sim\mu+\sigma(\mu)}^{[\mu]}}\prod_{\mu=1}^{L}\delta_{\alpha_{1,\,\sim\mu-\sigma(\mu)}^{[\mu]}\beta_{1,\,\sim\mu+\sigma(\mu)}^{[\mu]}}\cdots\prod_{\mu=1}^{L}\delta_{\alpha_{t,\,\sim\mu-\sigma(\mu)}^{[\mu]}\beta_{t,\,\sim\mu+\sigma(\mu)}^{[\mu]}}\nonumber \\
\times & \prod_{\mu=1}^{L}[A_{1}^{[\mu]}]_{b_{1,\sim\mu-\sigma(\mu)}^{[\mu]}a_{2,\sim\mu+\sigma(\mu)}^{[\mu]}}\prod_{\mu=1}^{L}[A_{2}^{[\mu]}]_{b_{2,\sim\mu-\sigma(\mu)}^{[\mu]}a_{3,\sim\mu+\sigma(\mu)}^{[\mu]}}\cdots\prod_{\mu=1}^{L}[A_{t}^{[\mu]}]_{b_{t,\sim\mu-\sigma(\mu)}^{[\mu]}\beta_{1,\sim\mu-\sigma(\mu)}^{[\mu]}}\nonumber \\
\times & \prod_{\mu=1}^{L}[B_{1}^{[\mu]}]_{\alpha_{1,\sim\mu+\sigma(\mu)}^{[\mu]}\beta_{2,\sim\mu-\sigma(\mu)}^{[\mu]}}\prod_{\mu=1}^{L}[B_{2}^{[\mu]}]_{\alpha_{2,\sim\mu+\sigma(\mu)}^{[\mu]}\beta_{3,\sim\mu-\sigma(\mu)}^{[\mu]}}\cdots\prod_{\mu=1}^{L}[B_{t}^{[\mu]}]_{\alpha_{t,\sim\mu+\sigma(\mu)}^{[\mu]}a_{1,\sim\mu+\sigma(\mu)}^{[\mu]}}.\nonumber \\
= & \sum_{P^{[1]},\,P^{\prime[1]}}\sum_{Q^{[1]},\,Q^{\prime[1]}}\cdots\sum_{P^{[L]},\,P^{\prime[L]}}\sum_{Q^{[L]},\,Q^{\prime[L]}}\mathcal{B}[\{P^{[\mu]},\,Q^{[\mu]}\}]\nonumber \\
\times & \prod_{\mu=1}^{L}\left\{ V_{P^{-1[\mu]}P^{\prime[\mu]}}V_{Q^{-1[\mu]}Q^{\prime[\mu]}}\prod_{\xi=1}^{t}\delta_{a_{\xi,\,\sim\mu+1}^{[\mu]}\alpha_{P^{[\mu]}(\xi),\,\sim\mu+1}^{[\mu]}}\delta_{b_{\xi,\,\sim\mu+1}^{[\mu]}\beta_{P^{\prime[\mu]}(\xi),\,\sim\mu+1}^{[\mu]}}\prod_{\xi=1}^{t}\delta_{a_{\xi,\,\sim\mu-1}^{[\mu]}\alpha_{Q^{[\mu]}(\xi),\,\sim\mu-1}^{[\mu]}}\delta_{b_{\xi,\,\sim\mu-1}^{[\mu]}\beta_{Q^{\prime[\mu]}(\xi),\,\sim\mu-1}^{[\mu]}}\right.\nonumber \\
\times & \delta_{a_{1,\,\sim\mu-\sigma(\mu)}^{[\mu]}b_{1,\,\sim\mu+\sigma(\mu)}^{[\mu]}}\cdots\delta_{a_{t,\,\sim\mu-\sigma(\mu)}^{[\mu]}b_{t,\,\sim\mu+\sigma(\mu)}^{[\mu]}}\delta_{\alpha_{1,\,\sim\mu-\sigma(\mu)}^{[\mu]}\beta_{1,\,\sim\mu+\sigma(\mu)}^{[\mu]}}\cdots\delta_{\alpha_{t,\,\sim\mu-\sigma(\mu)}^{[\mu]}\beta_{t,\,\sim\mu+\sigma(\mu)}^{[\mu]}}\nonumber\\
\times & [A_{1}^{[\mu]}]_{b_{1,\sim\mu-\sigma(\mu)}^{[\mu]}a_{2,\sim\mu+\sigma(\mu)}^{[\mu]}}[A_{2}^{[\mu]}]_{b_{2,\sim\mu-\sigma(\mu)}^{[\mu]}a_{3,\sim\mu+\sigma(\mu)}^{[\mu]}}\cdots[A_{t}^{[\mu]}]_{b_{t,\sim\mu-\sigma(\mu)}^{[\mu]}\beta_{1,\sim\mu-\sigma(\mu)}^{[\mu]}}[B_{1}^{[\mu]}]_{\alpha_{1,\sim\mu+\sigma(\mu)}^{[\mu]}\beta_{2,\sim\mu-\sigma(\mu)}^{[\mu]}}[B_{2}^{[\mu]}]_{\alpha_{2,\sim\mu+\sigma(\mu)}^{[\mu]}\beta_{3,\sim\mu-\sigma(\mu)}^{[\mu]}}[B_{t}^{[\mu]}]_{\alpha_{t,\sim\mu+\sigma(\mu)}^{[\mu]}a_{1,\sim\mu+\sigma(\mu)}^{[\mu]}}.
\end{align}
Upon introducing the compact notation 
\begin{align}
T_{P^{[\mu]}P^{\prime[\mu]}Q^{[\mu]}Q^{\prime[\mu]}} & \equiv\prod_{\xi=1}^{t}\delta_{a_{\xi,\,\sim\mu}^{[\mu]}\alpha_{P^{[\mu]}(\xi),\,\sim\mu+1}^{[\mu]}}\delta_{b_{\xi,\,\sim\mu+1}^{[\mu]}\beta_{P^{\prime[\mu]}(\xi),\,\sim\mu+1}^{[\mu]}}\prod_{\xi=1}^{t}\delta_{a_{\xi,\,\sim\mu-1}^{[\mu]}\alpha_{Q^{[\mu]}(\xi),\,\sim\mu-1}^{[\mu]}}\delta_{b_{\xi,\,\sim\mu-1}^{[\mu]}\beta_{Q^{\prime[\mu]}(\xi),\,\sim\mu-1}^{[\mu]}}\nonumber \\
 & \times\prod_{\xi=1}^{t}\delta_{a_{\xi,\,\sim\mu-\sigma(\mu)}^{[\mu]}b_{\xi\,\sim\mu+\sigma(\mu)}^{[\mu]}}\prod_{\xi=1}^{t}\delta_{\alpha_{\xi,\,\sim\mu-\sigma(\mu)}^{[\mu]}\beta_{\xi,\,\sim\mu+\sigma(\mu)}^{[\mu]}}\nonumber \\
 & \times[A_{1}^{[\mu]}]_{b_{1,\sim\mu-\sigma(\mu)}^{[\mu]}a_{2,\sim\mu+\sigma(\mu)}^{[\mu]}}[A_{2}^{[\mu]}]_{b_{2,\sim\mu-\sigma(\mu)}^{[\mu]}a_{3,\sim\mu+\sigma(\mu)}^{[\mu]}}\cdots[A_{t}^{[\mu]}]_{b_{t,\sim\mu-\sigma(\mu)}^{[\mu]}\beta_{1,\sim\mu-\sigma(\mu)}^{[\mu]}}[B_{1}^{[\mu]}]_{\alpha_{1,\sim\mu+\sigma(\mu)}^{[\mu]}\beta_{2,\sim\mu-\sigma(\mu)}^{[\mu]}}[B_{2}^{[\mu]}]_{\alpha_{2,\sim\mu+\sigma(\mu)}^{[\mu]}\beta_{3,\sim\mu-\sigma(\mu)}^{[\mu]}}[B_{t}^{[\mu]}]_{\alpha_{t,\sim\mu+\sigma(\mu)}^{[\mu]}a_{1,\sim\mu+\sigma(\mu)}^{[\mu]}},\label{eq:T-PP-Many-body}
\end{align}
we find 
\begin{equation}
\langle\mathrm{Tr}(UA_{1}\cdots UA_{t}U^{\dagger}B_{1}\cdots U^{\dagger}B_{t})\rangle=\sum_{P^{[1]},\,P^{\prime[1]}}\sum_{Q^{[1]},\,Q^{\prime[1]}}\cdots\sum_{P^{[L]},\,P^{\prime[L]}}\sum_{Q^{[L]},\,Q^{\prime[L]}}\mathcal{B}[\{P^{[\mu]},\,Q^{[\mu]}\}]\prod_{\mu=1}^{L}V_{P^{-1[\mu]}P^{\prime[\mu]}}V_{Q^{-1[\mu]}Q^{\prime[\mu]}}T_{P^{[\mu]}P^{\prime[\mu]}Q^{[\mu]}Q^{\prime[\mu[}}.\label{eq:Wg-ProdTr-Manybody-g1}
\end{equation}
\end{widetext}
This observation also generalize to the case $f$
is product of $g$ group trace-products, i.e., Eq.~(\ref{eq:Wg-ProdTr-Manybody}).

\section{Derivation of the upper bounds of the Lipschitz constant for the RMM and RQC cases\label{sec:upper-bounds}}
\subsection{The RMM case}

For the control protocol, we need to consider the QFI as a function
of the Haar random unitary $U$ and apply Observation~\ref{obs:bi-invariance}
to absorb the sensing gate $W$ into $U$ to get the following
function
\begin{equation}
\tilde{F}^{Q}[U]=4\mathrm{Var}\left(\tilde{G}[U]\right)_{\rho_{0}},
\end{equation}
where 
\begin{equation}
\tilde{G}[U]\equiv\sum_{s=1}^{t}U^{\dagger s}H_{0}U^{s}=\sum_{s=1}^{t}H_{\mathrm{sp}}(s)\label{eq:tilde-G-def},
\end{equation}
and $H_{\mathrm{sp}}(s)\equiv U^{\dagger s}H_{0}U^{s}$.

Note that if there was only one summation in Eq.~\eqref{eq:tilde-G-def},
$\tilde{G}[U]$ and $\tilde{F}^{Q}[U]$ would become the generator and the
QFI for the state-preparation protocol, respectively.
Upon introducing
\begin{equation}
Y[U]\equiv\tilde{G}[U]-\mathrm{Tr}\left(\rho_{0}\tilde{G}[U]\right),\label{eq:Y-def}
\end{equation}
 we can express 
\begin{equation}
\tilde{F}^{Q}[U]=4\mathrm{Tr}\left(\rho_{0}Y^{2}[U]\right).
\end{equation}
It is straightforward to calculate that
\begin{align}
\frac{d\tilde{F}^{Q}[e^{-\mathrm{i}X\epsilon}U]}{d\epsilon}\bigg|_{\epsilon=0} & =4\mathrm{Tr}\left[\{\rho_{0},\,Y[U]\}Y^{(1)}[U,X]\right]\nonumber \\
 & =4\mathrm{Tr}\left[\{\rho_{0},\,Y[U]\}\tilde{G}^{(1)}[U,\,X]\right],
\end{align}
where we have used the fact that $\mathrm{Tr}\left[\{\rho_{0},\,Y[U]\}\right]=0$,
\begin{align}
Y^{(1)}[U,X] & \equiv\frac{dY[e^{-\mathrm{i}X\epsilon}U)}{d\epsilon}\big|_{\epsilon=0}\nonumber \\
 & =\tilde{G}^{(1)}[U,X]-\mathrm{Tr}\left(\rho_{0}\tilde{G}^{(1)}[U,X]\right),\label{eq:Y1-def}
\end{align}
and
\begin{equation}
\tilde{G}^{(1)}[U,X]\equiv\frac{d\tilde{G}[e^{-\mathrm{i}X\epsilon}U]}{d\epsilon}\big|_{\epsilon=0}.\label{eq:G1-def}
\end{equation}
Now instead of computing $\tilde{G}^{(1)}[U,X]$ by brute force, let
us employ a trick that is generalizable to the RQC case. Let
us define 
\begin{equation}
U(\epsilon)\equiv e^{-\mathrm{i}X\epsilon}U\approx\left(\mathbb{I}-\mathrm{i}\epsilon X\right)U.
\end{equation}
Then it can be readily seen that 
\begin{align}
U^{s}(\epsilon) & =(\mathbb{I}-\mathrm{i}\epsilon X)U(\mathbb{I}-\mathrm{i}\epsilon X)U\cdots(\mathbb{I}-\mathrm{i}\epsilon X)U\nonumber \\
 & =U^{s}-\mathrm{i}\epsilon\sum_{\tau=0}^{s-1}U^{\tau}XU^{s-\tau}\nonumber \\
 & =U^{s}(\mathbb{I}-\mathrm{i}\epsilon\sum_{\tau=0}^{s-1}X(s-\tau)),
\end{align}
where 
\begin{equation}
X(\tau)\equiv U^{\dagger\tau}XU^{\tau}.
\end{equation}
Therefore 
\begin{equation}
\tilde{G}[e^{-\mathrm{i}X\epsilon}U]\approx G[U]+\mathrm{i}\epsilon\sum_{s=1}^{t}\sum_{\tau=0}^{s-1}X(s-\tau)H_{\mathrm{sp}}(s),
\end{equation}
and 
\begin{equation}
\tilde{G}^{(1)}[U,X]=\mathrm{i}\epsilon\sum_{s=1}^{t}\sum_{\tau=0}^{s-1}[X(s-\tau),\,H_{\mathrm{sp}}(s)].
\end{equation}
Using the H\"older inequalities, $|\mathrm{Tr}(A^{\dagger}B)|\le\|AB\|_{1}\leq\|A\|_{1}\|B\|_{\infty}$,
we obtain
\begin{align}
\bigg|\frac{d\tilde{F}^{Q}[e^{-\mathrm{i}X\epsilon}U]}{d\epsilon}\bigg|_{\epsilon=0} & =4\big|\mathrm{Tr}\left[\{\rho_{0},\,Y[U]\}\tilde{G}^{(1)}[U,\,X]\right]\big|\nonumber \\
 & \leq4\|\{\rho_{0},\,Y[U]\}\|_{1}\|\tilde{G}^{(1)}[U,\,X]\|_{\infty}\nonumber \\
 & \leq8\|\rho_{0}\|_{1}\|Y[U]\|_{\infty}\|\tilde{G}^{(1)}[U,\,X]\|_{\infty}.
\end{align}
Furthermore, 
\begin{align}
 & \|Y[U]\|_{\infty}\nonumber \\
= & \|\tilde{G}[U]-\mathrm{Tr}\left(\rho_{0}\tilde{G}[U]\right)\|_{\infty}\leq\|\tilde{G}[U]\|_{\infty}+\big|\mathrm{Tr}\left(\rho_{0}\tilde{G}[U]\right)\big|\nonumber \\
\leq & \|\tilde{G}[U]\|_{\infty}+\|\rho_{0}\|_{1}\|\tilde{G}[U]\|_{\infty}=2\|\sum_{s=1}^{t}U^{\dagger s}H_{0}U^{s}\|_{\infty}\leq2t\|H_{0}\|_{\infty},
\end{align}
where we have used the triangle inequality, the H\"older inequality
and the isometric invariance of the Schatten norms. Similarly, from Eq.~\eqref{eq:G1-def}
we can obtain
\begin{align}
 & \|\tilde{G}^{(1)}[U,\,X]\|_{\infty}\nonumber \\
\leq & \sum_{s=1}^{t}\sum_{\tau=0}^{s-1}\left(\|X(s-\tau)H_{\mathrm{sp}}(s)\|_{\infty}+\|H_{\mathrm{sp}}(s)X(s-\tau)\|_{\infty}\right)\nonumber \\
= & t(t+1)\|X\|_{\infty}\|\|H_{0}\|_{\infty}.
\end{align}
Therefore, upon noting the monotone property of the Schatten norm,
we conclude that
\begin{equation}
\bigg|\frac{d\tilde{F}^{Q}[e^{-\mathrm{i}X\epsilon}U]}{d\epsilon}\bigg|_{\epsilon=0}\leq16t^{2}(t+1)\|H_{0}\|_{\infty}\|X\|_{\mathrm{HS}}.
\end{equation}
On the the other hand, since $\tilde{F}^{Q}$ is invariant under the
shifting of $H_{0}$ by a constant, for the control protocol we know that
\begin{equation}
\mathcal{L}_{\mathrm{ctr-RMM}}\leq16t^{2}(t+1)\inf_{c}\|H_{0}-c\|_{\infty}=4t^{2}(t+1)\Delta_{H_{0}},
\end{equation}
where $\Delta_{H_{0}}$ is the spectrum gap. For the state-preparation protocol, one just need to replace
$\tilde{G}[U]$ with $U^{\dagger t}H_{0}U^{t}$ and perform the same
analysis to obtain
\begin{equation}
\|Y[U]\|_{\infty}\leq2\|H_{0}\|_{\infty},\:\|\tilde{G}^{(1)}[U,\,X]\|_{\infty}\leq2t\|X\|_{\infty}\|H_{0}\|_{\infty},
\end{equation}
and finally arrive at
\begin{equation}
\mathcal{L}_{\mathrm{sp-RMM}}\leq32t\inf_{c}\|H_{0}-c\|_{\infty}=8t\Delta_{H_{0}}.
\end{equation}

\subsection{The RQC case}

Let us now consider the control protocol in the RQC case. Similarly
to the control protocol of the RMT case, according to Observation~\ref{obs:bi-invariance}, to compute the Lipschitz constant,
we can absorb the local sensing gate into the two-site unitary and
consider 
\begin{equation}
\tilde{F}^{Q}[\{\mathsf{U}^{[\mu,\mu+1]}\}_{\mu=1}^{L}]=4\mathrm{Var}\left(\tilde{G}[\{\mathsf{U}^{[\mu,\mu+1]}\}_{\mu=1}^{L})]\right)_{\rho_{0}},
\end{equation}
where 
\begin{align}
\tilde{G}[\{\mathsf{U}^{[\mu,\mu+1]}\}_{\mu=1}^{L}] & \equiv\sum_{s=1}^{t}\sum_{\mu=1}^{L}\left(\prod_{\ell=1}^{L/2}\mathsf{U}^{[2\ell-1,2\ell]\dagger}\prod_{k=1}^{L/2}\mathsf{U}^{[2k,2k+1]\dagger}\right)^{s}h_{0}^{[\mu]}\nonumber \\
 & \times\left(\prod_{k=1}^{L/2}\mathsf{U}^{[2k,2k+1]}\prod_{\ell=1}^{L/2}\mathsf{U}^{[2\ell-1,2\ell]}\right)^{s}.
\end{align}
For the sake of simplicity, we introduce 
\begin{equation}
U=U_{\mathrm{e}}U_{\mathrm{o}},\,U(\epsilon)=e^{\mathrm{i}X_{\mathrm{e}}\epsilon}U_{\mathrm{e}}e^{\mathrm{i}X_{\mathrm{o}}\epsilon}U_{\mathrm{o}},\,H_{0}=\sum_{\mu=1}^L h_{0}^{[\mu]},\label{eq:U-H0-RQC}
\end{equation}
where 
\begin{align}
U_{\mathrm{e}}\equiv\prod_{k=1}^{L/2}\mathsf{U}^{[2k,2k+1]},\, & U_{\mathrm{o}}\equiv\prod_{\ell=1}^{L/2}\mathsf{U}^{[2\ell-1,2\ell]},\\
X_{\mathrm{e}}\equiv\sum_{k=1}^{L/2}X_{[2k,2k+1]},\, & X_{\mathrm{o}}\equiv\sum_{\ell=1}^{L/2}X_{[2\ell-1,2\ell]}.
\end{align}
It is straightforward to calculate 
\begin{align}
U(\epsilon) & =e^{-\mathrm{i}X_{\mathrm{e}}\epsilon}U_{\mathrm{e}}e^{-\mathrm{i}X_{\mathrm{o}}\epsilon}U_{\mathrm{o}}\approx(\mathbb{I}-\mathrm{i}\epsilon X_{\mathrm{e}})U_{\mathrm{e}}(\mathbb{I}-\mathrm{i}\epsilon X_{\mathrm{o}})U_{\mathrm{o}}\nonumber \\
 & \approx(\mathbb{I}-\mathrm{i}\epsilon X_{\mathrm{eo}})U\approx e^{-\mathrm{i}\epsilon X_{\mathrm{eo}}U},
\end{align}
where 
\begin{equation}
X_{\mathrm{eo}}\equiv X_{\mathrm{e}}+U_{\mathrm{e}}X_{\mathrm{o}}U_{\mathrm{e}}^{\dagger}\label{eq:Xeo-def}.
\end{equation}

With above notations, one can identity $\tilde{G}[\{\mathsf{U}^{[\mu,\mu+1]}\}_{\mu=1}^{L}]$,
$\tilde{G}[\{e^{-\mathrm{i}X_{[\mu,\mu+1]}\epsilon}\mathsf{U}^{[\mu,\mu+1]}\}]$
with $\tilde{G}[U]$ and $\tilde{G}[e^{-\mathrm{i}\epsilon X_{\mathrm{eo}}U}]$ respectively, where
$\tilde{G}[U]$ is defined in Eq.~\eqref{eq:tilde-G-def}, but with
$U$, $H_{0}$, and $X_{\mathrm{eo}}$ defined in Eq.~\eqref{eq:U-H0-RQC}
and Eq.~\eqref{eq:Xeo-def}, respectively. Then we follow similar
calculations as in the RMM case and find
\begin{align}
\tilde{G}^{(1)}[U,X_{\mathrm{eo}}] & \equiv\frac{d\tilde{G}[e^{-\mathrm{i}\epsilon X_{\mathrm{eo}}U}]}{d\epsilon}\big|_{\epsilon=0}\nonumber \\
 & =\mathrm{i}\epsilon\sum_{s=1}^{t}\sum_{\tau=0}^{s-1}[X_{\mathrm{eo}}(s-\tau),\,H_{\mathrm{sp}}(s)],\label{eq:G1-def-RQC}
\end{align}
and
\begin{equation}
\bigg|\frac{d\tilde{F}^{Q}[e^{-\mathrm{i}X_{\mathrm{eo}}\epsilon}U]}{d\epsilon}\bigg|_{\epsilon=0}\leq16t^{2}(t+1)\|H_{0}\|_{\infty}\|X_{\mathrm{eo}}\|_{\infty}.
\end{equation}
On the other hand, we know that
\begin{align}
\|X_{\mathrm{eo}}\|_{\infty} & =\|X_{\mathrm{e}}+U_{\mathrm{e}}X_{\mathrm{o}}U_{\mathrm{e}}^{\dagger}\|_{\infty}\nonumber \\
 & \leq\|X_{\mathrm{e}}\|_{\infty}+\|X_{\mathrm{o}}\|_{\infty}\leq\sum_{\mu=1}^{L}\|X_{[\mu,\mu+1]}\|_{\infty},
\end{align}
\begin{equation}
\|\{X_{[\mu,\mu+1]}\}_{\mu=1}^{L}\|_{\mathrm{HS}}=\sum_{\mu=1}^{L}\|X_{[\mu,\mu+1]}\|_{\mathrm{HS}}\ge\sum_{\mu=1}^{L}\|X_{[\mu,\mu+1]}\|_{\infty},
\end{equation}
and $\|H_{0}\|_{\infty}=L\|h_{0}\|_{\infty}$, where $h_{0}\equiv h_{0}^{[\mu]}$. Therefore we find that for
the control protocol for the RQC case, 
\begin{equation}
\mathcal{L}_{\mathrm{ctr-RQC}}\leq4t^{2}(t+1)L\Delta_{h_{0}},
\end{equation}
and for the state-preparation protocol for the RQC case,
\begin{equation}
\mathcal{L}_{\mathrm{sp-RQC}}\leq8tL\Delta_{h_{0}}.
\end{equation}

\bibliography{RQC-Metrology}

\end{document}